\documentclass[a4paper,11pt]{article}
\usepackage{amsmath,amsthm,amssymb}
\usepackage{graphicx}
\usepackage{mathrsfs}
\usepackage[all]{xy}
\makeatletter
\@addtoreset{equation}{section}
\makeatother

\newtheorem{theorem}{Theorem}[section]
\newtheorem*{theorem*}{Theorem}

\newtheorem{proposition}[theorem]{Proposition}
\newtheorem{corollary}[theorem]{Corollary}
\newtheorem*{corollary*}{Corollary}
\newtheorem{lemma}[theorem]{Lemma}
\newtheorem{example}[theorem]{Example}
\newtheorem{remark}[theorem]{Remark}
\newtheorem*{remark*}{Remark}

\newcommand{\BZ}{\mathbb{Z}}
\newcommand{\BR}{\mathbb{R}}
\newcommand{\BC}{\mathbb{C}}

\newcommand{\bD}{\mathbf{D}}
\newcommand{\bH}{\mathbf{H}}
\newcommand{\bI}{\mathbf{I}}

\newcommand{\bn}{\mathbf{n}}
\newcommand{\bsigma}{\boldsymbol{\sigma}}

\newcommand{\cB}{\mathcal{B}}
\newcommand{\cC}{\mathcal{C}}

\newcommand{\cO}{\mathcal{O}}
\newcommand{\cH}{\mathcal{H}}
\newcommand{\cI}{\mathcal{I}}
\newcommand{\cL}{\mathcal{L}}
\newcommand{\cP}{\mathcal{P}}
\newcommand{\cCL}{\mathcal{CL}}
\newcommand{\fg}{\mathfrak{g}}

\newcommand{\ra}{\mathrm{a}}
\newcommand{\rA}{\mathrm{A}}
\newcommand{\rB}{\mathrm{B}}
\newcommand{\rd}{\mathrm{d}}
\newcommand{\rH}{\mathrm{H}}
\newcommand{\rh}{\mathrm{h}}
\newcommand{\rQ}{\mathrm{Q}}

\newcommand{\Tr}{\operatorname{Tr}}
\newcommand{\tr}{\operatorname{tr}}
\newcommand{\even}{\mathrm{even}}
\newcommand{\odd}{\mathrm{odd}}
\newcommand{\NCHO}{\mathrm{NCHO}}
\newcommand{\QRM}{\mathrm{QRM}}
\newcommand{\op}{\mathrm{op}}
\newcommand{\HS}{\mathrm{HS}}

\newcommand{\sech}{\operatorname{sech}}
\newcommand{\artanh}{\operatorname{artanh}}
\renewcommand{\Re}{\operatorname{Re}}
\renewcommand{\Im}{\operatorname{Im}}
\newcommand{\hsum}{\sideset{}{^\oplus}{\sum}}

\newcommand{\eqspace}{\mathrel{\phantom{=}}}
\allowdisplaybreaks[3]

\title{Special values of spectral zeta functions of \\ one-, two-photon quantum Rabi models \\ and non-commutative harmonic oscillators}
\author{Ryosuke Nakahama\thanks{This work was supported by JST CREST Grant Number JPMJCR2113, Japan.}
\thanks{Email: ryosuke.nakahama@ntt.com} \\
\textit{NTT Institute for Fundamental Mathematics,} \\ \textit{Communication Science Laboratories, NTT, Inc.} \\
\textit{3-9-11 Midori-cho, Musashino-shi, Tokyo 180-8585, Japan}}
\date{\today}

\begin{document}

\maketitle

\begin{abstract}
We find explicit expressions of the special values of the Hurwitz-type spectral zeta function $\zeta(\rH;n,\lambda)$ for the Hamiltonians $\rH$ of the one-photon quantum Rabi model (1pQRM), 
the two-photon quantum Rabi model (2pQRM), and the non-commutative harmonic oscillator (NCHO), at positive integers $n$. 
Then the 1st term of the spectral zeta function of 1pQRM gives a generalization of Beukers' integral used for the proof of the irrationality of $\zeta(2)$ after Ap\'ery's work. 
A similar expression of the 1st term of that of 2pQRM is also discussed. \medskip

\noindent \textbf{Keywords}: quantum Rabi models; non-commutative harmonic oscillators; spectral zeta functions; Beukers' integrals; Ap\'ery-like numbers. 

\noindent \textbf{2020 Mathematics Subject Classification}: 81Q10; 11M41. 
\end{abstract}

\section{Introduction}

The \emph{quantum Rabi model} (QRM) is a light-matter interaction model introduced by Jaynes and Cummings \cite{JC}, 
with the Hamiltonian given by the formally self-adjoint operator on $L^2(\BR)\otimes\BC^2$, 
\[ \rH_{1\QRM}^{(g,\Delta,0)}:=\bI\ra^\dagger\ra+g\bsigma_1(\ra+\ra^\dagger)+\Delta\bsigma_3, \]
where $g,\Delta\in\BR$, $\ra^\dagger,\ra$ are the creation and annihilation operators on $L^2(\BR)$, and $\bsigma_1,\bsigma_3\in M(2,\BC)$ are the Pauli matrices. 
This model describes a linear interaction of an electromagnetic mode and a two-level system. 
The integrability of this model is proved by Braak \cite{B1, B2} by using its $\BZ_2$-symmetry. 
As a generalization, the asymmetric quantum Rabi model 
\[ \rH_{1\QRM}^{(g,\Delta,\varepsilon)}:=\bI\ra^\dagger\ra+g\bsigma_1(\ra+\ra^\dagger)+\Delta\bsigma_3+\varepsilon\bsigma_1, \]
where $\varepsilon\in\BR$, which breaks the $\BZ_2$-symmetry of the original QRM, is also well-studied (see, e.g., \cite{B1, B2, KRW, XZBL}). 
Similarly, the \emph{two-photon quantum Rabi model} (2pQRM)
\[ \rH_{2\QRM}^{(g,\Delta,0)}:=\bI\ra^\dagger\ra+g\bsigma_1(\ra^2+(\ra^\dagger)^2)+\Delta\bsigma_3 \]
describes a non-linear interaction of an electromagnetic mode and a two-level system. 
The one-photon QRM (1pQRM) for all $g\in\BR$ and 2pQRM for $|g|<1/2$ have only discrete spectra, but this is not the case for 2pQRM with $|g|\ge 1/2$. 
Indeed, for $g=1/2$ case, 2pQRM may contain both discrete and continuous spectra (see, e.g., \cite{B1, B2, B3, DXBC, L, XC}). 
We can also consider the asymmetric generalization of 2pQRM, 
\[ \rH_{2\QRM}^{(g,\Delta,\varepsilon)}:=\bI\ra^\dagger\ra+g\bsigma_1(\ra^2+(\ra^\dagger)^2)+\Delta\bsigma_3+\varepsilon\bsigma_1. \]

Independently, as a purely mathematical model, motivated by number-theoretic interests and lower bound estimates of systems of (pseudo)differential operators, 
Parmeggiani and Wakayama \cite{PW0, PW1, PW2} introduced the \emph{non-commutative harmonic oscillator} (NCHO), 
with the Hamiltonian operator on $L^2(\BR)\allowbreak\otimes\BC^2$, 
\begin{align*}
\rH_\NCHO^{(\alpha,\beta,0)}:\hspace{-3pt}&=\begin{bmatrix}\alpha&0\\0&\beta\end{bmatrix}\biggl(-\frac{1}{2}\frac{d^2}{dx^2}+\frac{1}{2}x^2\biggr)
+\begin{bmatrix}0&-1\\1&0\end{bmatrix}\biggl(x\frac{d}{dx}+\frac{1}{2}\biggr) \\
&=\begin{bmatrix}\alpha&0\\0&\beta\end{bmatrix}\biggl(\ra^\dagger\ra+\frac{1}{2}\biggr)+\frac{1}{2}\begin{bmatrix}0&-1\\1&0\end{bmatrix}(\ra^2-(\ra^\dagger)^2), 
\end{align*}
where $\alpha,\beta\in\BR$. If $\alpha\beta>1$, then this has discrete spectra \cite{P2014, PV}, and in this article, we always assume $\alpha\beta>1$. 
The spectra of NCHO are studied mathematically via the spectral zeta functions (see \cite{IW2, KW2, KW4, O3} and references therein), 
via holomorphic differential equations (see e.g., \cite{O1, O2, RW, W}), and via pseudodifferential operators (see e.g., \cite{M, MP}).  
A generalization of NCHO, the shifted NCHO 
\begin{align*}
\rH_\NCHO^{(\alpha,\beta,\eta)}
&=\begin{bmatrix}\alpha&0\\0&\beta\end{bmatrix}\biggl(\ra^\dagger\ra+\frac{1}{2}\biggr)
+\frac{1}{2}\begin{bmatrix}0&-1\\1&0\end{bmatrix}(\ra^2-(\ra^\dagger)^2)+2\eta\sqrt{\alpha\beta-1}\sqrt{-1}\begin{bmatrix}0&-1\\1&0\end{bmatrix}, 
\end{align*}
where $\eta\in\BR$, is also considered by \cite{RW}. 
By \cite{N2}, the eigenvalue problem of NCHO is equivalent to that of 2pQRM, although the origins of these two models are different. 
Also, again by \cite{N2}, a deformation of 2pQRM converges to 1pQRM by taking a suitable limit. This is a refinement of the confluence process given in \cite{RW, W}. 

In this article, we renormalize $\rH_{1\QRM}^{(g,\Delta,\varepsilon)}$, $\rH_{2\QRM}^{(g,\Delta,\varepsilon)}$ as 
\begin{align*}
\widetilde{\rH}_{1\QRM}^{(g,\Delta,\varepsilon)}:\hspace{-3pt}&=\rH_{1\QRM}^{(g,\Delta,\varepsilon)}+g^2\bI
=\bI(\ra^\dagger\ra+g^2)+g\bsigma_1(\ra+\ra^\dagger)+\Delta\bsigma_3+\varepsilon\bsigma_1, \\
\widetilde{\rH}_{2\QRM}^{(g,\Delta,\varepsilon)}:\hspace{-3pt}&=\cosh(2g)\biggl(\rH_{2\QRM}^{(\tanh(2g)/2,\Delta\sech(2g),\varepsilon\sech(2g))}+\frac{1}{2}\bI\biggr) \\
&=\cosh(2g)\bI\biggl(\ra^\dagger\ra+\frac{1}{2}\biggr)+\frac{\sinh(2g)}{2}\bsigma_1(\ra^2+(\ra^\dagger)^2)+\Delta\bsigma_3+\varepsilon\bsigma_1, 
\end{align*}
where $g,\Delta,\varepsilon\in\BR$, and renormalize $\rH_\NCHO^{(\alpha,\beta,\eta)}$ as 
\begin{align*}
&\widetilde{\rH}_\NCHO^{(\alpha,\beta,\eta)}:=\frac{\alpha+\beta}{2\sqrt{\alpha\beta(\alpha\beta-1)}}\rH_\NCHO^{(\alpha,\beta,\eta)} \\
&=\frac{\alpha+\beta}{2\sqrt{\alpha\beta(\alpha\beta-1)}} \\
&\eqspace{}\times\biggl(\begin{bmatrix}\alpha&0\\0&\beta\end{bmatrix}\biggl(\ra^\dagger\ra+\frac{1}{2}\biggr)
+\frac{1}{2}\begin{bmatrix}0&-1\\1&0\end{bmatrix}(\ra^2-(\ra^\dagger)^2)+2\eta\sqrt{\alpha\beta-1}\sqrt{-1}\begin{bmatrix}0&-1\\1&0\end{bmatrix}\biggr), 
\end{align*}
where $\alpha,\beta>0$, $\alpha\beta>1$, $\eta\in\BR$, to make the later formulas simpler. Then these have only discrete spectra, 
and by using the eigenvalues $\{\mu_j\}$ of $\rH=\widetilde{\rH}_{1\QRM}^{(g,\Delta,\varepsilon)}, \widetilde{\rH}_{2\QRM}^{(g,\Delta,\varepsilon)}, \widetilde{\rH}_\NCHO^{(\alpha,\beta,\eta)}$, 
its Hurwitz-type spectral zeta function is defined by 
\[ \zeta(\rH;s,\lambda):=\sum_j \frac{1}{(\mu_j+\lambda)^s}=\Tr\bigl((\rH+\lambda)^{-s}\bigr) \]
for $s,\lambda\in\BC$ when it converges. Such spectral zeta functions are studied by, e.g., \cite{HS, IW2, KW2, KW4, O3, RW2, RW3, S}. 
The purpose of this article is to give explicit expressions of the special values $\zeta(\rH;n,\lambda)$ for $n\in\BZ_{\ge 2}$. 
To describe the main results, for $g\in\BR$, $\lambda,\varepsilon\in\BC$ with $\lambda\pm\varepsilon\notin -\BZ_{\ge 0}$, $m\in\BZ_{\ge 1}$, let 
\begin{align*}
R_m^\flat(\lambda;g,\varepsilon)&:=\Tr\bigl(\bigl((\ra^\dagger\ra+g(\ra+\ra^\dagger)+g^2+\varepsilon+\lambda)^{-1}(\ra^\dagger\ra-g(\ra+\ra^\dagger)+g^2-\varepsilon+\lambda)^{-1}\bigr)^m\bigr), \\
C^\flat(\lambda,\varepsilon)&:=\max_{\delta\in\{\pm \},\ k\in\BZ_{\ge 0}}\frac{1}{|k+\lambda+\delta\varepsilon|}, 
\end{align*}
and for $g\in\BR$, $\lambda,\varepsilon\in\BC$ with $\lambda\pm\varepsilon\notin -\bigl(\BZ_{\ge 0}+\frac{1}{2}\bigr)$, $m\in\BZ_{\ge 1}$, let 
\begin{align*}
R_m^+(\lambda;g,\varepsilon):\hspace{-3pt}&=\Tr\biggl(\biggl(\biggl(\cosh(2g)\biggl(\ra^\dagger\ra+\frac{1}{2}\biggr)+\frac{\sinh(2g)}{2}(\ra^2+(\ra^\dagger)^2)+\varepsilon+\lambda\biggr)^{-1} \\*
&\eqspace{}\times\biggl(\cosh(2g)\biggl(\ra^\dagger\ra+\frac{1}{2}\biggr)-\frac{\sinh(2g)}{2}(\ra^2+(\ra^\dagger)^2)-\varepsilon+\lambda\biggr)^{-1}\biggr)^m\biggr), \\ 
C(\lambda,\varepsilon):\hspace{-3pt}&=\max_{\delta\in\{\pm \},\ k\in\BZ_{\ge 0}}\frac{1}{\bigl|k+\lambda+\delta\varepsilon+\frac{1}{2}\bigr|}. 
\end{align*}
Then the main result of Section \ref{section_zeta_sum} is given as follows. 

\begin{theorem}[Corollary \ref{cor_zeta_sum}]\label{thm_intro1}
\begin{enumerate}
\item Let $g,\Delta,\varepsilon\in\BR$, $\lambda\in\BC$. We assume $\lambda\pm\varepsilon\notin -\BZ_{\ge 0}$ and $|\Delta|<C_\flat(\lambda,\varepsilon)^{-1}$. 
Then for $n\in\BZ_{\ge 2}$, $(\widetilde{\rH}_{1\QRM}^{(g,\Delta,\varepsilon)}+\lambda)^{-n}$ is of trace class, and we have 
\begin{align*}
\zeta\bigl(\widetilde{\rH}_{1\QRM}^{(g,\Delta,\varepsilon)};n,\lambda\bigr) 
=\zeta(n,\lambda+\varepsilon)+\zeta(n,\lambda-\varepsilon)
+\frac{(-1)^n}{(n-1)!}\sum_{m=1}^\infty \frac{\Delta^{2m}}{m}\frac{\partial^n R_m^\flat}{\partial\lambda^n}(\lambda;g,\varepsilon). 
\end{align*}
\item Let $g,\Delta,\varepsilon\in\BR$, $\lambda\in\BC$. 
We assume $\lambda\pm\varepsilon\notin -\bigl(\BZ_{\ge 0}+\frac{1}{2}\bigr)$ and $|\Delta|<C(\lambda,\varepsilon)^{-1}$. 
Then for $n\in\BZ_{\ge 2}$, $(\widetilde{\rH}_{2\QRM}^{(g,\Delta,\varepsilon)}+\lambda)^{-n}$ is of trace class, and we have 
\begin{align*}
&\zeta\bigl(\widetilde{\rH}_{2\QRM}^{(g,\Delta,\varepsilon)};n,\lambda\bigr) \\
&=\zeta\biggl(n,\lambda+\varepsilon+\frac{1}{2}\biggr)+\zeta\biggl(n,\lambda-\varepsilon+\frac{1}{2}\biggr)
+\frac{(-1)^n}{(n-1)!}\sum_{m=1}^\infty \frac{\Delta^{2m}}{m}\frac{\partial^n R_m^+}{\partial\lambda^n}(\lambda;g,\varepsilon). 
\end{align*}
\item Let $\alpha,\beta,\eta\in\BR$, $\lambda\in\BC$. We assume $\alpha\beta>1$, $\lambda\pm 2\eta\notin -\bigl(\BZ_{\ge 0}+\frac{1}{2}\bigr)$ 
and $\bigl|\lambda\frac{\alpha-\beta}{\alpha+\beta}\bigr|C(\lambda,2\eta)\allowbreak<1$. 
Then for $n\in\BZ_{\ge 2}$, $(\widetilde{\rH}_\NCHO^{(\alpha,\beta,\eta)}+\lambda)^{-n}$ is of trace class, and we have 
\begin{multline*}
\zeta\bigl(\widetilde{\rH}_\NCHO^{(\alpha,\beta,\eta)};n,\lambda\bigr)
=\zeta\biggl(n,\lambda+2\eta+\frac{1}{2}\biggr)+\zeta\biggl(n,\lambda-2\eta+\frac{1}{2}\biggr) \\*
{}+\frac{(-1)^n}{(n-1)!}\sum_{m=1}^\infty \frac{1}{m}\biggl(\frac{\alpha-\beta}{\alpha+\beta}\biggr)^{2m}
\frac{\partial^n}{\partial \lambda^n}\lambda^{2m}R_m^+\biggl(\lambda;\frac{1}{2}\artanh\biggl(\frac{1}{\sqrt{\alpha\beta}}\biggr),2\eta\biggr). 
\end{multline*}
\end{enumerate}
\end{theorem}

Here, $\zeta(s,\lambda):=\sum_{k=0}^\infty (k+\lambda)^{-s}$ is the Hurwitz zeta function. 
Next, in Section \ref{section_int_exp}, we find the integral expressions of $R_m^\flat(\lambda;g,\varepsilon)$, $R_m^+(\lambda;g,\varepsilon)$, by using the Lie semigroup theory. 

\begin{theorem}\label{thm_intro2}
\begin{enumerate}
\item (Theorem \ref{thm_int_exp}\,(1)) Let $g\in\BR$, $\lambda,\varepsilon\in\BC$, $m\in\BZ_{\ge 1}$. If $\Re\lambda-|\Re\varepsilon|>0$, then we have 
\begin{align*}
R_m^\flat(\lambda;g,\varepsilon)
&=\int_{[0,1]^{2m}} \exp\biggl(-4g^2\frac{\Phi_m(u_1,\ldots,u_{2m})}{1-\prod_{j=1}^{2m}u_j}\biggr)\frac{\prod_{j=1}^{2m} u_j^{\lambda-(-1)^j\varepsilon-1}}{1-\prod_{j=1}^{2m}u_j}
\,\prod_{j=1}^{2m}\rd u_j,
\end{align*}
where 
\[ \Phi_m(u_1,\ldots,u_{2m}):=m\biggl(1+\prod_{j=1}^{2m}u_j\biggr)+\sum_{l=1}^{2m-1}(-1)^l\sum_{k=1}^{2m}\prod_{j=k}^{k+l-1}u_j, \]
with $u_{2m+j}:=u_j$ for $j=1,\ldots,2m$. 
\item (Corollary \ref{cor_int_exp}) Let $g\in\BR$, $\lambda,\varepsilon\in\BC$, $m\in\BZ_{\ge 1}$. If $\Re\lambda-|\Re\varepsilon|+\frac{1}{2}>0$, then we have 
\begin{align*}
R_m^+(\lambda;g,\varepsilon)&=\int_{[0,1]^{2m}} \frac{1}{\sqrt{\Psi^g_m(u_1,\ldots,u_{2m})-2}}\prod_{j=1}^{2m} u_j^{\lambda-(-1)^j\varepsilon-1}\,\rd u_j, 
\end{align*}
where 
\begin{align*}
\Psi^g_m(u_1,\ldots,u_{2m}):=\tr\biggl(\prod_{1\le j\le 2m}^{\longrightarrow}\biggl(\begin{bmatrix} \cosh(2g)&(-1)^j\sinh(2g) \\ (-1)^j\sinh(2g)&\cosh(2g) \end{bmatrix}
\begin{bmatrix}u_j^{-1}&0\\0&u_j\end{bmatrix}\biggr)\biggr). 
\end{align*}
\end{enumerate}
\end{theorem}

Here, $\prod_{1\le j\le 2m}^{\rightarrow}A_j:=A_1A_2\cdots A_{2m}$ for non-commutative elements $A_j$. From these theorems, we immediately get the following. 
\begin{corollary}\label{cor_intro}
\begin{enumerate}
\item Let $g,\Delta,\varepsilon\in\BR$, $\lambda\in\BC$. We assume $\Re\lambda-|\varepsilon|>0$ and $|\Delta|<|\lambda-|\varepsilon||$. 
Then for $n\in\BZ_{\ge 2}$, we have 
\begin{multline*}
\zeta\bigl(\widetilde{\rH}_{1\QRM}^{(g,\Delta,\varepsilon)};n,\lambda\bigr) 
=\zeta(n,\lambda+\varepsilon)+\zeta(n,\lambda-\varepsilon)+\frac{1}{\Gamma(n)}\sum_{m=1}^\infty \frac{\Delta^{2m}}{m} \\
{}\times\int_{[0,1]^{2m}} \exp\biggl(-4g^2\frac{\Phi_m(u_1,\ldots,u_{2m})}{1-\prod_{j=1}^{2m}u_j}\biggr)
\frac{\bigl(-\log\prod_{j=1}^{2m}u_j\bigr)^n}{1-\prod_{j=1}^{2m}u_j}\prod_{j=1}^{2m} u_j^{\lambda-(-1)^j\varepsilon-1}\,\rd u_j. 
\end{multline*}
\item Let $g,\Delta,\varepsilon\in\BR$, $\lambda\in\BC$. 
We assume $\Re\lambda-|\varepsilon|+\frac{1}{2}>0$ and $|\Delta|<\bigl|\lambda-|\varepsilon|+\frac{1}{2}\bigr|$. 
Then for $n\in\BZ_{\ge 2}$, we have 
\begin{multline*}
\zeta\bigl(\widetilde{\rH}_{2\QRM}^{(g,\Delta,\varepsilon)};n,\lambda\bigr)
=\zeta\biggl(n,\lambda+\varepsilon+\frac{1}{2}\biggr)+\zeta\biggl(n,\lambda-\varepsilon+\frac{1}{2}\biggr)+\frac{1}{\Gamma(n)}\sum_{m=1}^\infty \frac{\Delta^{2m}}{m} \\
{}\times\int_{[0,1]^{2m}} \frac{1}{\sqrt{\Psi^g_m(u_1,\ldots,u_{2m})-2}}\biggl(-\log\prod_{j=1}^{2m}u_j\biggr)^n\prod_{j=1}^{2m} u_j^{\lambda-(-1)^j\varepsilon-1}\,\rd u_j. 
\end{multline*}
\end{enumerate}
\end{corollary}

These theorems refine the result by Kimoto--Wakayama \cite{KW4} on the spectral zeta function of NCHO. 
We can also compute the spectral zeta function of 1pQRM by using the heat kernel given by Reyes-Bustos--Wakayama \cite{RW2, RW1, RW3}, 
and it is desirable to verify the coincidence of our special values and those computed by using the heat kernel for 1pQRM.  

Especially, the 1st term of $\zeta\bigl(\widetilde{\rH}_{1\QRM}^{(g,\Delta,\varepsilon)};n,\lambda\bigr)$ is given by the differential with respect to $\lambda$ of 
\begin{align*}
R_1^\flat(\lambda;g,\varepsilon)&=\int_{[0,1]^2}\exp\biggl(-4g^2\frac{(1-u)(1-v)}{1-uv}\biggr)\frac{u^{\lambda+\varepsilon-1}v^{\lambda-\varepsilon-1}}{1-uv}\,\rd u\rd v \\
&=\sum_{n=0}^\infty \frac{(-4g^2)^n}{n!}\int_{[0,1]^2} \frac{(1-u)^n(1-v)^n}{(1-uv)^{n+1}}u^{\lambda+\varepsilon-1}v^{\lambda-\varepsilon-1}\,\rd u\rd v. 
\end{align*}
Then each term of this sum is regarded as a generalization of Beukers' integral \cite{Be1}, 
\begin{equation}\label{intro_Beukers}
(-1)^n\int_{[0,1]^2} \frac{(1-u)^n(1-v)^nu^nv^n}{(1-uv)^{n+1}}\,\rd u\rd v=A_n\zeta(2)-B_n, 
\end{equation}
used for the proof of the irrationality of $\zeta(2)$, clarifying the Ap\'ery's work \cite{Ap, vdP}. 
This work is also studied via the modularity of the generating function $\sum_{n=0}^\infty A_nt^n$ of $\{A_n\}$ by \cite{Be2, Be3}. 
In this article, we seek an expression analogous to (\ref{intro_Beukers}) of each term of $R_1^\flat(\lambda;g,\varepsilon)$ for general $\lambda,\varepsilon$ in Section \ref{section_1pQRM}. 

Similarly, the 1st term of $\zeta\bigl(\widetilde{\rH}_{2\QRM}^{(g,\Delta,\varepsilon)};n,\lambda\bigr)$ is given by the differential with respect to $\lambda$ of 
\begin{align*}
R_1^+(\lambda;g,\varepsilon)&=\sech(2g)\int_{[0,1]^2} \frac{u^{\lambda+\varepsilon-\frac{1}{2}}v^{\lambda-\varepsilon-\frac{1}{2}}}{\sqrt{(1-uv)^2-\tanh^2(2g)(u-v)^2}}\,\rd u\rd v \\
&=\sech(2g)\sum_{n=0}^\infty \frac{\bigl(\frac{1}{2}\bigr)_n}{n!}\tanh^{2n}(2g)\int_{[0,1]^2} \frac{(u-v)^{2n}}{(1-uv)^{2n+1}}
u^{\lambda+\varepsilon-\frac{1}{2}}v^{\lambda-\varepsilon-\frac{1}{2}}\,\rd u\rd v, 
\end{align*}
where $\bigl(\frac{1}{2}\bigr)_n:=\frac{1}{2}\cdot\frac{3}{2}\cdots\bigl(n-\frac{1}{2}\bigr)$. 
Then again each term of this sum is regarded as an analogue of Beukers' integral (``Ap\'ery-like numbers'' in the terminology of Kimoto--Wakayama \cite{KW4}). 
Its modularity is studied by \cite{KW2} for $\lambda=\varepsilon=0$, and by \cite{KW4} for $\varepsilon=0$. 
In this article, we seek an expression analogous to (\ref{intro_Beukers}) of each term of $R_1^+(\lambda;g,\varepsilon)$ for general $\lambda,\varepsilon$ in Section \ref{section_2pQRM}. 

The Hamiltonian $\widetilde{\rH}_{1\QRM}^{(g,\Delta,\varepsilon)}$ on $L^2(\BR)\otimes\BC^2$ is unitarily equivalent to the $\BC^2$-valued differential operator 
$\widetilde{\rH}_{\flat}^{(g,\Delta,\varepsilon)}$ on $\BC$ (see (\ref{def_Hflat})) via the Bargmann transform. 
Similarly, the restrictions of the Hamiltonian $\widetilde{\rH}_{2\QRM}^{(g,\Delta,\varepsilon)}$ on $L^2(\BR)_\even\otimes\BC^2$ and $L^2(\BR)_\odd\otimes\BC^2$ are 
unitarily equivalent to the $\BC^2$-valued differential operators $\widetilde{\rH}_{\nu}^{(g,\Delta,\varepsilon)}$ on the unit disc $\bD$ (see (\ref{def_Hnu})) with $\nu=1/2, 3/2$ respectively, 
via the Laplace and Cayley transforms, so that we have 
\begin{align*}
\zeta\bigl(\widetilde{\rH}_{1\QRM}^{(g,\Delta,\varepsilon)};s,\lambda\bigr)&=\zeta\bigl(\widetilde{\rH}_{\flat}^{(g,\Delta,\varepsilon)};s,\lambda\bigr), \\
\zeta\bigl(\widetilde{\rH}_{2\QRM}^{(g,\Delta,\varepsilon)};s,\lambda\bigr)
&=\zeta\bigl(\widetilde{\rH}_{1/2}^{(g,\Delta,\varepsilon)};s,\lambda\bigr)+\zeta\bigl(\widetilde{\rH}_{3/2}^{(g,\Delta,\varepsilon)};s,\lambda\bigr). 
\end{align*}
The operator $\widetilde{\rH}_{\nu}^{(g,\Delta,\varepsilon)}$ is defined for all $\nu>0$, and by using the dilation map 
\[ \rho_\nu\colon\cO(\bD)\longrightarrow\cO(\sqrt{\nu}\bD), \qquad f(z)\longmapsto f(w/\sqrt{\nu}), \]
the composition $\rho_\nu\circ\widetilde{\rH}_{\nu}^{(g,\Delta,\varepsilon)}\circ\rho_\nu^{-1}$ defines a $\BC^2$-valued differential operator on $\sqrt{\nu}\bD$. 
Then by \cite{N2}, as a formal differential operator, we have the limit 
\[ \lim_{\nu\to\infty}\frac{1}{2}\bigl(\rho_\nu\circ\widetilde{\rH}_{\nu}^{(g/\sqrt{\nu},2\Delta,2\varepsilon)}\circ\rho_\nu^{-1}-\nu\bigr)
=\widetilde{\rH}_{\flat}^{(g,\Delta,\varepsilon)}. \]
This refines the confluence process given in \cite{RW, W}. Especially, it is expected that 
\[ 2^s\lim_{\nu\to\infty}\zeta\bigl(\widetilde{\rH}_{\nu}^{(g/\sqrt{\nu},2\Delta,2\varepsilon)};s,2\lambda-\nu\bigr)
=\zeta\bigl(\widetilde{\rH}_{\flat}^{(g,\Delta,\varepsilon)};s,\lambda\bigr) \]
holds. This is verified in Corollary \ref{cor_zeta}\,(3) for $s=n\in\BZ_{\ge 2}$ and for $\lambda,\Delta,\varepsilon$ satisfying the conditions in Corollary \ref{cor_intro}\,(1). 

In this article, the connection of $\zeta\bigl(\widetilde{\rH}_{\nu}^{(g,\Delta,\varepsilon)};n,\lambda\bigr)$ and the ``Ap\'ery-like numbers'' are discussed only when $\nu=1/2,3/2$, 
which correspond to the original 2pQRM on $L^2(\BR)\otimes\BC^2$. However, if $\nu\in\frac{1}{2}\BZ_{\ge 1}$, then $\widetilde{\rH}_{\nu}^{(g,\Delta,\varepsilon)}$ is unitarily equivalent to 
a restriction of a multivariate generalization of 2pQRM, by using the theory of spherical harmonics (see \cite{N} and Remarks \ref{rem_sph_harm}, \ref{rem_multivariate} below). 
Hence it is desirable to find a connection with the ``Ap\'ery-like numbers'' for $\nu$ other than $1/2,3/2$.

\section{Preliminaries}
\subsection{Trace class norm, Hilbert--Schmidt norm and Schatten $p$-norm}

In this subsection, we review the trace class, the Hilbert--Schmidt, and the Schatten $p$-norms from, e.g., \cite{Sch}. 
Let $\cH$ be a separable Hilbert space equipped with the inner product $\langle\cdot,\cdot\rangle$ and the norm $\Vert v\Vert:=\sqrt{\langle v,v\rangle}$. 
Throughout the article, for a bounded linear operator $\rA$ on $\cH$, its operator norm, trace class norm, and Hilbert--Schmidt norm are denoted by 
$\Vert \rA\Vert_\op$, $\Vert \rA\Vert_{\Tr}$, and $\Vert \rA\Vert_\HS$ respectively. 
\begin{align*}
\Vert \rA\Vert_\op:=\sup_{v\in\cH,\; \Vert v\Vert=1} \Vert \rA v\Vert, \quad 
\Vert \rA\Vert_{\Tr}:=\sum_k \bigl\langle \sqrt{\rA^\dagger \rA}v_k,v_k\bigr\rangle, \quad
\Vert \rA\Vert_\HS:=\biggl(\sum_k \Vert \rA v_k\Vert^2\biggr)^{1/2}. 
\end{align*}
Here, $\{v_k\}\subset\cH$ is a fixed orthonormal basis. Then $\Vert \rA\Vert_{\Tr}, \Vert \rA\Vert_\HS$ do not depend on the choice of $\{v_k\}$. 
If $\Vert \rA\Vert_{\Tr}<\infty$, then its trace is denoted by $\Tr(\rA)$. 
\[ \Tr(\rA):=\sum_k \langle \rA v_k,v_k\rangle. \]
Again this does not depend on the choice of $\{v_k\}$. We use $\tr(A)$ for the usual matrix trace, and distinguish with the operator trace $\Tr(\rA)$. 
For two bounded operators $\rA,\rB$, the following inequalities hold. 
\begin{align*}
\Vert\rA\rB\Vert_{\Tr}\le \Vert\rA\Vert_{\Tr}\Vert\rB\Vert_{\op}, \qquad
\Vert\rA\rB\Vert_{\Tr}\le \Vert\rA\Vert_{\op}\Vert\rB\Vert_{\Tr}, \qquad
\Vert\rA\rB\Vert_{\Tr}\le \Vert\rA\Vert_{\HS}\Vert\rB\Vert_{\HS}. 
\end{align*}
More generally, for $1\le p\le\infty$, let $\Vert\rA\Vert_p$ denote the Schatten $p$-norm, 
\[ \Vert\rA\Vert_p:=\bigl(\Tr((\rA^\dagger\rA)^{p/2})\bigr)^{1/p}, \]
so that $\Vert\rA\Vert_1=\Vert\rA\Vert_{\Tr}$, $\Vert\rA\Vert_2=\Vert\rA\Vert_\HS$, $\Vert\rA\Vert_\infty=\Vert\rA\Vert_\op$ hold. 
When $\frac{1}{p}+\frac{1}{q}=\frac{1}{r}$, we have H\"older's inequality, 
\[ \Vert\rA\rB\Vert_r\le \Vert\rA\Vert_p\Vert\rB\Vert_q \]
(see, e.g., \cite[Corollaire 4.4]{Fac}, \cite[Theorem 4.2]{FacKos}), and recursively, when $\frac{1}{p_1}+\cdots+\frac{1}{p_n}=1$, we have 
\begin{equation}\label{Holder_ineq}
\Vert\rA_1\rA_2\cdots\rA_n\Vert_{\Tr}\le \Vert\rA_1\Vert_{p_1}\Vert\rA_2\Vert_{p_2}\cdots\Vert\rA_n\Vert_{p_n}. 
\end{equation}

Next, suppose that $\rA$ is a self-adjoint operator densely defined on $\cH$ which has only discrete spectra $\{\mu_j\}$ with no accumulation points. 
Then for a resolvent $\lambda\in\BC$ of $-\rA$ and for $s\in\BC$, the Hurwitz-type spectral zeta function of $\rA$ is defined as 
\[ \zeta(\rA;s,\lambda):=\sum_j\frac{1}{(\mu_j+\lambda)^s}=\Tr\bigl((\rA+\lambda)^{-s}\bigr) \]
if it converges. Especially, if $(\rA+\lambda)^{-s_0}$ is of trace class for some $s_0\in\BR$, then $\zeta(\rA;s,\lambda)$ converges for $\Re s\ge s_0$. 
Its analytic continuation for $\Re s<s_0$ is also denoted by $\zeta(\rA;s,\lambda)$.

\subsection{Semigroup acting on Fock space}

We recall the \textit{Fock space} from, e.g., \cite[Section 1.6]{F}, 
\[ \cH_\flat(\BC):=\{ f\in\cO(\BC)\mid \Vert f\Vert_\flat^2:=\langle f,f\rangle_\flat<\infty\}, \]
defined by the inner product 
\[ \langle f,g\rangle_\flat:=\frac{1}{\pi}\int_\BC f(w)\overline{g(w)}e^{-|w|^2}\,\rd w \qquad (f,g\in\cO(\BC)). \]
Then $\cH_\flat(\BC)$ is unitarily isomorphic to $L^2(\BR)$ via the \textit{Bergmann transform}
\[ \cB\colon L^2(\BR)\longrightarrow\cH_\flat(\BC), \qquad (\cB f)(w):=\pi^{-1/4}\int_\BR f(x)e^{-(x^2+w^2)/2+\sqrt{2}xw}\,\rd x, \]
and the annihilation and creation operators on $L^2(\BR)$, 
\[ \ra:=\frac{1}{\sqrt{2}}\biggl(x+\frac{\rd}{\rd x}\biggr), \qquad \ra^\dagger:=\frac{1}{\sqrt{2}}\biggl(x-\frac{\rd}{\rd x}\biggr) \]
satisfy 
\[ (\cB\ra f)(w)=\frac{\rd}{\rd w}(\cB f)(w), \qquad (\cB\ra^\dagger f)(w)=w(\cB f)(w). \]

Next, let $G^\BC_\flat$ and $\fg^\BC_\flat$ be the complex Lie group and its Lie algebra given by 
\[ G^\BC_\flat:=\left\{ \begin{bmatrix} 1&a&b\\0&c&d\\0&0&1 \end{bmatrix}\ \middle|\ \begin{matrix} a,b,d\in\BC,\\ c\in\BC^\times \end{matrix} \right\}, \qquad 
\fg^\BC_\flat:=\left\{ \begin{bmatrix} 0&a&b\\0&c&d\\0&0&0 \end{bmatrix}\ \middle|\ \begin{matrix} a,b,c,d\in\BC \end{matrix} \right\}. \]
$G^\BC_\flat$ is the semidirect product of $\Bigl\{\Bigl[\begin{smallmatrix}1&0&0\\0&c&0\\0&0&1\end{smallmatrix}\Bigr]\Bigr\}\simeq\BC^\times$ 
and the Heisenberg group $\Bigl\{\Bigl[\begin{smallmatrix}1&a&b\\0&1&d\\0&0&1\end{smallmatrix}\Bigr]\Bigr\}$. 
Next we define the involutive anti-automorphism $\ddagger$ on $G^\BC_\flat$ and $\fg^\BC_\flat$ by 
\[ \begin{bmatrix} 1&a&b\\0&c&d\\0&0&1 \end{bmatrix}^\ddagger :=\begin{bmatrix} 1&-\overline{d}&\overline{b}\\0&\overline{c}&-\overline{a}\\0&0&1 \end{bmatrix}, \qquad
\begin{bmatrix} 0&a&b\\0&c&d\\0&0&0 \end{bmatrix}^\ddagger :=\begin{bmatrix} 0&-\overline{d}&\overline{b}\\0&\overline{c}&-\overline{a}\\0&0&0 \end{bmatrix}, \]
and let 
\[ G_\flat:=\{ \gamma\in G^\BC_\flat \mid \gamma^\ddagger=\gamma^{-1} \}, \qquad \fg_\flat:=\{ \gamma\in \fg^\BC_\flat \mid \gamma^\ddagger=-\gamma \}. \]
Then $G_\flat^\BC$ and $\fg_\flat^\BC$ act on the space $\cO(\BC)$ of entirely holomorphic functions on $\BC$ by 
\begin{align*}
\left(\tau_\flat\left(\begin{bmatrix} 1&a&b\\0&c&d\\0&0&1 \end{bmatrix}\right)f\right)(w)&:=e^{dw-b}f(cw-a), \\
\left(\rd\tau_\flat\left(\begin{bmatrix} 0&a&b\\0&c&d\\0&0&0 \end{bmatrix}\right)f\right)(w)&:=\biggl[-a\frac{\rd}{\rd w}-b+cw\frac{\rd}{\rd w}+dw\biggr]f(w), 
\end{align*}
and $G_\flat$ acts unitarily on $\cH_\flat(\BC)$, namely, 
for $\gamma=\left[\begin{smallmatrix}1&a&|a|^2/2+\sqrt{-1}b\\0&c&\overline{a}c\\0&0&1\end{smallmatrix}\right]\in G_\flat$ with $a\in\BC$, $b\in\BR$, $|c|=1$, we have 
\begin{align*}
\Vert \tau_\flat(\gamma)f\Vert_\flat^2&=\frac{1}{\pi}\int_\BC \bigl|e^{\overline{a}cw-(|a|^2/2+\sqrt{-1}b)}f(cw-a)\bigr|^2 e^{-|w|^2}\,\rd w \\
&=\frac{1}{\pi}\int_\BC |f(cw-a)|^2 e^{-|cw-a|^2}\,\rd w=\frac{1}{\pi}\int_\BC |f(w)|^2 e^{-|w|^2}\,\rd w=\Vert f\Vert_\flat^2. 
\end{align*}
Also, as an unbounded operator on $\cH_\flat(\BC)$, $\rd\tau_\flat(\gamma)^\dagger=\rd\tau_\flat(\gamma^\ddagger)$ holds for all $\gamma\in\fg^\BC_\flat$. 

Next, let $\Gamma_\flat^\circ\subset G^\BC_\flat$ be the semigroup given by 
\[ \Gamma_\flat^\circ:=\left\{\gamma=\begin{bmatrix}1&a&b\\0&c&d\\0&0&1\end{bmatrix}\in G^\BC_\flat\ \middle|\ |\det(\gamma)|=|c|<1\right\}. \]
Then the following holds. 

\begin{lemma}
For $\gamma\in\Gamma_\flat^\circ$, we have $\tau_\flat(\gamma)(\cH_\flat(\BC))\subset\cH_\flat(\BC)$, 
and as an operator on $\cH_\flat(\BC)$, we have $\tau_\flat(\gamma)^\dagger=\tau_\flat(\gamma^\ddagger)$. 
\end{lemma}

\begin{proof}
For $r>0$, let 
\begin{align*}
\cH_\flat(\BC)(r)&:=\biggl\{ f\in\cO(\BC)\biggm| \frac{1}{\pi}\int_\BC f(w)\overline{g(w)}e^{-|w|^2/r^2}\,\rd w<\infty\biggr\}, \\
\Gamma_\flat^\circ(r)&:=\left\{\gamma=\begin{bmatrix}1&a&b\\0&c&d\\0&0&1\end{bmatrix}\in G^\BC_\flat\ \middle|\ |\det(\gamma)|=|c|<r\right\}, 
\end{align*}
so that $\cH_\flat(\BC)(1)=\cH_\flat(\BC)$, $\Gamma_\flat^\circ(1)=\Gamma_\flat^\circ$. Then for $r>0$, $f\in\cH_\flat(\BC)(r)$, 
$\gamma=\Bigl[\begin{smallmatrix}1&a&b\\0&c&d\\0&0&1\end{smallmatrix}\Bigr]\in\Gamma_\flat^\circ(r)$, we have 
\begin{align*}
&\Vert\tau_\flat(\gamma)f\Vert_\flat^2 \\
&=\frac{1}{\pi}\int_\BC \bigl|e^{dw-b}f(cw-a)\bigr|^2e^{-|w|^2}\,\rd w
=\frac{1}{|c|^2\pi}\int_\BC \bigl|e^{d(w+a)/c-b}f(w)\bigr|^2e^{-|w+a|^2/|c|^2}\,\rd w \\
&\le \frac{1}{|c|^2}\max_{w\in\BC} \exp\biggl(2\Re\biggl(\frac{d(w+a)}{c}-b\biggr)-\biggl(\frac{1}{|c|^2}-\frac{1}{r^2}\biggr)|w|^2-2\frac{\Re w\overline{a}}{|c|^2}-\frac{|a|^2}{|c|^2}\biggr) \\
&\eqspace{}\times\frac{1}{\pi}\int_\BC |f(w)|^2e^{-|w|^2/r^2}\,\rd w<\infty, 
\end{align*}
and hence $\tau_\flat(\Gamma_\flat^\circ(r))(\cH_\flat(\BC)(r))\subset\cH_\flat(\BC)$. Next, since $G_\flat$ acts unitarily on $\cH_\flat(\BC)$, we have 
\[ \langle \tau_\flat(\gamma)f,g\rangle_\flat=\langle f,\tau_\flat(\gamma^{-1})g\rangle_\flat=\langle f,\tau_\flat(\gamma^\ddagger)g\rangle_\flat 
\qquad (f,g\in\cH_\flat(\BC),\; \gamma\in G_\flat), \]
and for $r>1$, since $\cH_\flat(\BC)(r)\subset\cH_\flat(\BC)$ holds and $G_\flat\subset\Gamma_\flat^\circ(r)$ is totally real, we have 
\[ \langle \tau_\flat(\gamma)f,g\rangle_\flat=\langle f,\tau_\flat(\gamma^\ddagger)g\rangle_\flat 
\qquad (f,g\in\cH_\flat(\BC)(r),\; \gamma\in \Gamma_\flat^\circ(r)). \]
Then since $\Gamma_\flat^\circ\subset\Gamma_\flat^\circ(r)$ and $\cH_\flat(\BC)(r)$ is dense in $\cH_\flat(\BC)$, we get 
\[ \langle \tau_\flat(\gamma)f,g\rangle_\flat=\langle f,\tau_\flat(\gamma^\ddagger)g\rangle_\flat 
\qquad (f,g\in\cH_\flat(\BC),\; \gamma\in \Gamma_\flat^\circ), \]
namely, $\tau_\flat(\gamma)^\dagger=\tau_\flat(\gamma^\ddagger)$ holds. 
\end{proof}

Next we consider the trace of $\tau_\flat(\gamma)$ for $\gamma\in\Gamma_\flat^\circ$. 

\begin{theorem}\label{thm_Tr_tauflat}
For $\gamma=\Bigl[\begin{smallmatrix}1&a&b\\0&c&d\\0&0&1\end{smallmatrix}\Bigr]\in\Gamma_\flat^\circ$, the operator $\tau_\flat(\gamma)$ on $\cH_\flat(\BC)$ is of trace class, 
and we have 
\[ \Tr(\tau_\flat(\gamma))=\frac{\exp\bigl(-\frac{ad+b(1-c)}{1-c}\bigr)}{1-c}. \]
\end{theorem}

\begin{proof}
First, we can directly check that the map 
\[ \{\gamma\in G^\BC_\flat\mid \gamma=\gamma^\ddagger,\; \det(\gamma)>0\}\longrightarrow\{\gamma\in G^\BC_\flat\mid \gamma=\gamma^\ddagger,\; \det(\gamma)>0\}, \qquad \gamma\longmapsto\gamma^2 \]
is bijective, and let $\gamma\mapsto\sqrt{\gamma}$ be its inverse. Then since we have 
\[ \det\bigl(\sqrt{\gamma^\ddagger\gamma}\bigr)=|\det(\gamma)| \qquad (\gamma\in G_\flat^\BC), \]
$\gamma\in\Gamma_\flat^\circ$ holds if and only if $\sqrt{\gamma^\ddagger\gamma}\in\Gamma_\flat^\circ$ holds, and we have 
\[ \Vert\tau_\flat(\gamma)\Vert_{\Tr}=\Tr\Bigl(\sqrt{\tau_\flat(\gamma)^\dagger\tau_\flat(\gamma)}\Bigr)=\Tr\Bigl(\tau_\flat\bigl(\sqrt{\gamma^\ddagger\gamma}\bigr)\Bigr). \]
Hence it suffices to compute $\Tr(\tau_\flat(\gamma))$ for $\gamma\in\Gamma_\flat^\circ$. 
We consider the orthonormal basis $\{f_k\}_{k=0}^\infty\subset\cH_\flat(\BC)$, 
\[ f_k(w):=\frac{1}{\sqrt{k!}}w^k \qquad (k\in\BZ_{\ge 0}). \]
Then if we write 
\[ (\tau_\flat(\gamma)f_k)(w)=\sum_{j=0}^\infty a_jf_j(w)=\sum_{j=0}^\infty \frac{a_j}{\sqrt{j!}}w^j, \]
then we have 
\begin{align*}
\langle\tau_\flat(\gamma)f_k,f_k\rangle_\flat&=a_k=\frac{\sqrt{k!}}{2\pi\sqrt{-1}}\oint_{|w|=r}(\tau_\flat(\gamma)f_k)(w)w^{-k-1}\,\rd w \\
&=\frac{1}{2\pi\sqrt{-1}}\oint_{|w|=r}e^{dw-b}(cw-a)^kw^{-k-1}\,\rd w
\end{align*}
for any $r>0$. Now we fix $r>\frac{|a|}{1-|c|}$. Then since $|c|=|\det(\gamma)|<1$, for $|w|=r$, we have 
\[ \biggl|\frac{cw-a}{w}\biggr|\le \frac{|c||w|+|a|}{|w|}=\frac{r|c|+|a|}{r}<1, \]
and hence we get 
\begin{align*}
\Tr(\tau_\flat(\gamma))&=\sum_{k=0}^\infty\langle\tau_\flat(\gamma)f_k,f_k\rangle_\flat=\sum_{k=0}^\infty \frac{1}{2\pi\sqrt{-1}}\oint_{|w|=r}e^{dw-b}(cw-a)^kw^{-k-1}\,\rd w \\
&=\frac{1}{2\pi\sqrt{-1}}\oint_{|w|=r}\sum_{k=0}^\infty e^{dw-b}(cw-a)^kw^{-k-1}\,\rd w \\
&=\frac{1}{2\pi\sqrt{-1}}\oint_{|w|=r}\frac{e^{dw-b}}{1-(cw-a)w^{-1}}\frac{\rd w}{w} 
=\frac{1}{2\pi\sqrt{-1}}\oint_{|w|=r}\frac{\exp(dw-b)}{(1-c)w+a}\,\rd w \\
&=\frac{\exp\bigl(-d\frac{a}{1-c}-b\bigr)}{1-c}=\frac{\exp\bigl(-\frac{ad+b(1-c)}{1-c}\bigr)}{1-c}. \qedhere
\end{align*}
\end{proof}

\begin{example}\label{ex_Xg}
For $g\in\BR$, $u\in(0,1)$, let 
\begin{align*}
X(g):\hspace{-3pt}&=\begin{bmatrix}0&-g&-g^2\\0&1&g\\0&0&0\end{bmatrix}
=\begin{bmatrix}1&-g&g^2/2\\0&1&-g\\0&0&1\end{bmatrix}\begin{bmatrix}0&0&0\\0&1&0\\0&0&0\end{bmatrix}\begin{bmatrix}1&g&g^2/2\\0&1&g\\0&0&1\end{bmatrix}\in\fg^\BC_\flat, \\ 
u^{X(g)}:\hspace{-3pt}&=\exp(X(g)\log u)=\begin{bmatrix}1&-g&g^2/2\\0&1&-g\\0&0&1\end{bmatrix}\begin{bmatrix}1&0&0\\0&u&0\\0&0&1\end{bmatrix}\begin{bmatrix}1&g&g^2/2\\0&1&g\\0&0&1\end{bmatrix}
\in\Gamma_\flat^\circ, 
\end{align*}
so that 
\begin{align}
\rd\tau_\flat(X(g))&=w\frac{\rd}{\rd w}+g\biggl(\frac{\rd}{\rd w}+w\biggr)+g^2=\cB\bigl(\ra^\dagger\ra+g(\ra+\ra^\dagger)+g^2\bigr)\cB^{-1}, \label{formula_dtauXg}\\
\tau_\flat(u^{X(g)})&=\exp\biggl(\log u\biggl(w\frac{\rd}{\rd w}+g\biggl(\frac{\rd}{\rd w}+w\biggr)+g^2\biggr)\biggr) \notag
\end{align}
are self-adjoint on $\cH_\flat(\BC)$. Then the sets of their eigenvalues are $\BZ_{\ge 0}$ and $\{u^k\mid k\in\BZ_{\ge 0}\}$ respectively, with the common eigenvectors 
$\left\{\tau_\flat\left(\left[\begin{smallmatrix}1&-g&g^2/2\\0&1&-g\\0&0&1\end{smallmatrix}\right]\right)w^k\ \middle|\ k\in\BZ_{\ge 0}\right\}$. 
Especially, for $\lambda\in\BC$, $\lambda\notin-\BZ_{\ge 0}$, the resolvent $(\rd\tau_\flat(X(g))+\lambda)^{-1}$ is of Hilbert--Schmidt class, with 
\[ \bigl\Vert(\rd\tau_\flat(X(g))+\lambda)^{-1}\bigr\Vert_\HS^2=\sum_{k=0}^\infty\frac{1}{|k+\lambda|^2}, \qquad 
\bigl\Vert(\rd\tau_\flat(X(g))+\lambda)^{-1}\bigr\Vert_\op=\max_{k\in\BZ_{\ge 0}}\frac{1}{|k+\lambda|}, \]
and for $n\in\BZ_{\ge 2}$, we have 
\[ \Tr\bigl((\rd\tau_\flat(X(g))+\lambda)^{-n}\bigr)=\sum_{k=0}^\infty\frac{1}{(k+\lambda)^n}=\zeta(n,\lambda), \]
where $\zeta(s,\lambda)$ is the Hurwitz zeta function. Similarly, $\tau_\flat(u^{X(g)})$ is of trace class, with 
\begin{gather*}
\Tr\bigl(\tau_\flat(u^{X(g)})\bigr)=\bigl\Vert\tau_\flat(u^{X(g)})\bigr\Vert_{\Tr}=\sum_{k=0}^\infty u^k=\frac{1}{1-u}, \\
\bigl\Vert\tau_\flat(u^{X(g)})\bigr\Vert_\HS^2=\sum_{k=0}^\infty (u^k)^2=\frac{1}{1-u^2}, \qquad 
\bigl\Vert\tau_\flat(u^{X(g)})\bigr\Vert_\op=\max_{k\in\BZ_{\ge 0}} u^k=1. 
\end{gather*}
\end{example}

\subsection{Semigroup acting on unit disc}

Let $\bD$ and $\bH$ be the unit disc and the upper half plane, 
\[ \bD:=\{z\in\BC\mid |z|<1\}, \qquad \bH:=\{y\in\BC\mid \Im y>0\}. \]
For $\nu>0$, we consider the \textit{weighted Bergman space}
\[ \cH_\nu(\bD):=\{ f\in\cO(\bD)\mid \Vert f\Vert_\nu^2:=\langle f,f\rangle_\nu<\infty\}, \]
defined by the inner product given by, for $f(z)=\sum_{k=0}^\infty a_kz^k$, $g(z)=\sum_{k=0}^\infty b_kz^k\in\cO(\bD)$, 
\begin{align*}
\langle f,g\rangle_\nu:\hspace{-3pt}&=\sum_{k=0}^\infty \frac{k!}{(\nu)_k}a_k\overline{b_k} && (\nu>0) \\
&=\frac{\nu-1}{\pi}\int_\bD f(z)\overline{g(z)}(1-|z|^2)^{\nu-2}\,\rd z && (\nu>1), 
\end{align*}
where $(\nu)_k:=\nu(\nu+1)(\nu+2)\cdots(\nu+k-1)$. Then $L^2(\BR)_\even$ and $L^2(\BR)_\odd$ are unitarily isomorphic to $\cH_{1/2}(\bD)$ and $\cH_{3/2}(\bD)$ as follows. 
Let $\cL_\even$, $\cL_\odd$ be the Laplace transforms, 
\begin{align*}
\cL_\even&\colon L^2(\BR)_\even\longrightarrow \cO(\bH), & (\cL_\even f)(y)&:=\pi^{-1/4}\int_\BR f(x)e^{x^2y\sqrt{-1}/2}\,\rd x, \\
\cL_\odd&\colon L^2(\BR)_\odd\longrightarrow \cO(\bH), & (\cL_\odd f)(y)&:=\sqrt{2}\pi^{-1/4}\int_\BR f(x)xe^{x^2y\sqrt{-1}/2}\,\rd x, 
\end{align*}
let $\cI_\even$, $\cI_\odd$ be the integral transforms, 
\begin{align*}
\cI_\even&\colon \cH_\flat(\BC)_\even\longrightarrow \cO(\bD), & (\cI_\even f)(z)&:=\frac{1}{\pi}\int_\BC f(w)e^{z\overline{w}^2/2}e^{-|w|^2}\,\rd w, \\
\cI_\odd&\colon \cH_\flat(\BC)_\odd\longrightarrow \cO(\bD), & (\cI_\odd f)(z)&:=\frac{1}{\pi}\int_\BC f(w)\overline{w}e^{z\overline{w}^2/2}e^{-|w|^2}\,\rd w, 
\end{align*}
and for $\nu>0$, let $\cC_\nu$ be the weighted Cayley transform, 
\[ \cC_\nu\colon\cO(\bH)\longrightarrow\cO(\bD), \qquad (\cC_\nu f)(z):=(1+z)^{-\nu}f\biggl(\sqrt{-1}\frac{1-z}{1+z}\biggr). \]
Then the following diagrams commute. 
\[ \xymatrix{ L^2(\BR)_\even \ar[r]^{\cB} \ar[d]_{\cL_\even} & \cH_\flat(\BC)_\even \ar[d]^{\cI_\even} \\ \cO(\bH) \ar[r]_{\cC_{1/2}} & \cO(\bD), } \qquad  
\xymatrix{ L^2(\BR)_\odd \ar[r]^{\cB} \ar[d]_{\cL_\odd} & \cH_\flat(\BC)_\odd \ar[d]^{\cI_\odd} \\ \cO(\bH) \ar[r]_{\cC_{3/2}} & \cO(\bD). } \]
Moreover, the compositions 
\begin{align*}
\cC\cL_{1/2}&:=\cC_{1/2}\circ\cL_\even=\cI_\even\circ\cB\colon L^2(\BR)_\even\longrightarrow\cH_{1/2}(\bD), \\ 
\cC\cL_{3/2}&:=\cC_{3/2}\circ\cL_\odd=\cI_\odd\circ\cB\colon L^2(\BR)_\odd\longrightarrow\cH_{3/2}(\bD)
\end{align*}
give the unitary isomorphisms satisfying, for $\nu=1/2$, $3/2$, 
\begin{gather*}
\cC\cL_\nu\circ\biggl(\ra^\dagger\ra+\frac{1}{2}\biggr)=\biggl(2z\frac{\rd}{\rd z}+\nu\biggr)\circ\cC\cL_\nu, \\
\cC\cL_\nu\circ\frac{1}{2}\ra^2=\frac{\rd}{\rd z}\circ\cC\cL_\nu, \qquad 
\cC\cL_\nu\circ\frac{1}{2}(\ra^\dagger)^2=\biggl(z^2\frac{\rd}{\rd z}+\nu z\biggr)\circ\cC\cL_\nu. 
\end{gather*}
See e.g., \cite[Sections 4.5, 4.6]{F}. 

\begin{remark}\label{rem_sph_harm}
More generally, for $\nu>0$, let 
\begin{gather*}
L^2_\nu(\BR_{>0}):=\biggl\{f\colon\BR_{>0}\to\BC\biggm|\Vert f\Vert_{\nu,\BR_{>0}}^2 :=\frac{2^\nu}{\Gamma(\nu)}\int_0^\infty |f(t)|^2 t^{\nu-1}\,dt<\infty\biggr\}, \\
\cL_\nu\colon L^2_\nu(\BR_{>0})\longrightarrow \cO(\bH), \qquad 
(\cL_\nu f)(y):=\frac{2^\nu}{\Gamma(\nu)}\int_0^\infty f(t)e^{\sqrt{-1}yt}t^{\nu-1}\,dt. 
\end{gather*}
Then the composition 
\[ \cC\cL_{\nu}:=\cC_\nu\circ\cL_\nu\colon L^2_\nu(\BR_{>0})\longrightarrow \cH_\nu(\bD) \]
gives the unitary isomorphism. Also, for $n\in\BZ_{>0}$, $k\in\BZ_{\ge 0}$, let $\cP(\BR^n)$ be the space of complex-valued polynomials on $\BR^n$, and let 
\[ \cH\cP_k(\BR^n):=\biggl\{h\in\cP(\BR^n)\biggm| h(tx)=t^kh(x) \quad(t\in\BR),\quad \sum_{j=1}^n\frac{\partial^2h}{\partial x_j^2}=0\biggr\} \]
be the space of homogeneous harmonic polynomials of degree $k$, equipped with the inner product 
\[ \langle h_1,h_2\rangle_{\cH\cP}:=\int_{\BR^n}h_1(x)\overline{h_2(x)}e^{-r^2(x)}\,dx, \]
where $r^2(x):=x_1^2+\cdots+x_n^2$. Then $L^2(\BR^n)$ has the orthogonal decomposition 
\[ L^2(\BR^n)\simeq \hsum_{k\in\BZ_{\ge 0}} L^2_{k+\frac{n}{2}}(\BR_{>0})\otimes \cH\cP_k(\BR^n), \]
where the embedding is given by $f(t)\otimes h(x)\mapsto f(r^2(x)/2)h(x)$. 
Let $\ra_j$, $\ra_j^\dagger$ be the annihilation and creation operators with respect to the $j$-th variable on $L^2(\BR^n)$. Then for $\nu=k+\frac{n}{2}$, $f(t)\in L^2_\nu(\BR_{>0})$, $h(x)\in\cH\cP_k(\BR^n)$, we have 
\begin{gather*}
(\cC\cL_\nu\otimes Id_{\cP\cH_k(\BR^n)}) \biggl(\sum_{j=1}^n\biggl(\ra_j^\dagger\ra_j+\frac{1}{2}\biggr)\biggl(f\biggl(\frac{r^2(x)}{2} \biggr)h(x)\biggr)\biggr)
=\biggl(2z\frac{\rd}{\rd z}+\nu\biggr)(\cC\cL_\nu f)(z)\otimes h(x), \\
(\cC\cL_\nu\otimes Id_{\cP\cH_k(\BR^n)}) \biggl(\frac{1}{2}\sum_{j=1}^n\ra_j^2 \biggl(f\biggl(\frac{r^2(x)}{2} \biggr)h(x)\biggr)\biggr)
=\frac{\rd}{\rd z}(\cC\cL_\nu f)(z)\otimes h(x), \\
(\cC\cL_\nu\otimes Id_{\cP\cH_k(\BR^n)}) \biggl(\frac{1}{2}\sum_{j=1}^n(\ra_j^\dagger)^2 \biggl(f\biggl(\frac{r^2(x)}{2} \biggr)h(x)\biggr)\biggr)
=\biggl(z^2\frac{\rd}{\rd z}+\nu z\biggr)(\cC\cL_\nu f)(z)\otimes h(x). 
\end{gather*}
This gives a special case of Howe's dual pair correspondence. See, e.g., \cite{H, KV}. See also \cite{N}. 
\end{remark}

Next, let 
\[ G^\BC:=SL(2,\BC)=\biggl\{\begin{bmatrix}a&b\\c&d\end{bmatrix}\biggm|\begin{matrix}a,b,c,d\in\BC, \\ ad-bc=1\end{matrix}\biggr\}, \quad 
\fg^\BC:=\mathfrak{sl}(2,\BC)=\biggl\{\begin{bmatrix}a&b\\c&-a\end{bmatrix}\biggm|a,b,c\in\BC\biggr\}. \]
We define the involutive anti-automorphism $\ddagger$ on $G^\BC$ and $\fg^\BC$ by 
\[ \begin{bmatrix}a&b\\c&d\end{bmatrix}^\ddagger:=\begin{bmatrix}1&0\\0&-1\end{bmatrix}\begin{bmatrix}a&b\\c&d\end{bmatrix}^\dagger\begin{bmatrix}1&0\\0&-1\end{bmatrix}
=\begin{bmatrix}\overline{a}&-\overline{c}\\-\overline{b}&\overline{d}\end{bmatrix}, \]
and let 
\[ G:=SU(1,1)=\{\gamma\in G^\BC\mid \gamma^\ddagger=\gamma^{-1}\}, \qquad \fg:=\mathfrak{su}(1,1)=\{\gamma\in\fg^\BC\mid \gamma^\ddagger=-\gamma\}. \]
Then for $\nu>0$, the universal covering group $\widetilde{G}$ of $G$ acts unitarily on $\cH_\nu(\bD)$ by 
\begin{equation}\label{def_taunu}
\biggl(\tau_\nu\biggl(\begin{bmatrix}a&b\\c&d\end{bmatrix}^{-1}\biggr)f\biggr)(z):=(cz+d)^{-\nu}f\biggl(\frac{az+b}{cz+d}\biggr). 
\end{equation}
We note that $(cz+d)^{-\nu}$ is not well-defined on $G\times \bD$ if $\nu\notin\BZ$, but is well-defined on the universal covering space $\widetilde{G}\times \bD$ for all $\nu>0$. 
Its complexified differential action of $\fg^\BC$ is given by the unbounded operator on $\cH_\nu(\bD)$, 
\[ \biggl(\rd\tau_\nu\biggl(\begin{bmatrix}a&b\\c&-a\end{bmatrix}\biggr)f\biggr)(z):=-a\biggl(2z\frac{\rd}{\rd z}+\nu\bigg)-b\frac{\rd}{\rd z}+c\biggl(z^2\frac{\rd}{\rd z}+\nu z\biggr), \]
so that $\rd\tau_\nu(\gamma)^\dagger=\rd\tau_\nu(\gamma^\ddagger)$ holds. 

Next, by using the action of $G^\BC=SL(2,\BC)$ on $\BC\cup\{\infty\}$, 
\[ \gamma.z:=\frac{az+b}{cz+d} \qquad \biggl(\gamma=\begin{bmatrix}a&b\\c&d\end{bmatrix}\in G^\BC,\; z\in\BC\cup\{\infty\}\biggr), \]
we define the semigroups $\Gamma,\Gamma^\circ\subset G^\BC$ by 
\begin{align*}
\Gamma&:=\{\gamma\in G^\BC\mid \gamma^{-1}.z\in \bD \text{ for all }z\in \bD\}, \\
\Gamma^\circ&:=\{\gamma\in G^\BC\mid \gamma^{-1}.z\in \bD \text{ for all }z\in \overline{\bD}\}, 
\end{align*}
where $\overline{\bD}\subset\BC$ is the closure of $\bD$. Also, let $C,C^\circ\subset\sqrt{-1}\fg\subset\fg^\BC$ be the cones given by 
\begin{align*}
C&:=\biggl\{\begin{bmatrix}a&b\\-\overline{b}&-a\end{bmatrix} \biggm| \begin{matrix}a\in\BR_{\ge 0},\;b\in\BC, \\ a^2\ge |b|^2\end{matrix} \biggr\}
=Ad(G)\BR_{\ge 0}\begin{bmatrix}1&0\\0&-1\end{bmatrix}, \\
C^\circ&:=\biggl\{\begin{bmatrix}a&b\\-\overline{b}&-a\end{bmatrix} \biggm| \begin{matrix}a\in\BR_{>0},\;b\in\BC, \\ a^2>|b|^2\end{matrix} \biggr\}
=Ad(G)\BR_{>0}\begin{bmatrix}1&0\\0&-1\end{bmatrix}. 
\end{align*}
Then the following holds. 

\begin{theorem}\label{thm_semigroup}
\begin{enumerate}
\item (\cite[Part I, Theorem V.1.1]{FKKLR})
\[ \Gamma=G\exp(C), \qquad \Gamma^\circ=G\exp(C^\circ). \]
\item (\cite[Part I, Proof of Proposition V.3.1]{FKKLR}) Let $\gamma\in\Gamma^\circ\subset M(2,\BC)$. If $\tr(\gamma)\in\BR$, then 
\[ |\tr(\gamma)|>2. \]
\item (\cite[Part I, Proposition V.3.1]{FKKLR}) If $\gamma\in\Gamma^\circ$, then one of the eigenvalues $\mu(\gamma)$ of $\gamma\in M(2,\BC)$ satisfies 
\[ |\mu(\gamma)|<1. \]
The function $\gamma\mapsto\mu(\gamma)$ is holomorphic on $\Gamma^\circ$. 
\end{enumerate}
\end{theorem}

The universal covering semigroups $\widetilde{\Gamma}, \widetilde{\Gamma}^\circ$ of $\Gamma,\Gamma^\circ$ act on $\cH_\nu(\bD)$ by the same formula as (\ref{def_taunu}). 
If $\gamma\in\exp(C^\circ)\subset\widetilde{\Gamma}^\circ$, then $\tau_\nu(\gamma)$ is self-adjoint. 
Moreover, the trace of $\tau_\nu(\gamma)$ for $\gamma\in\widetilde{\Gamma}^\circ$ satisfies the following. 
\begin{theorem}[{\cite[Part I, Theorem V.3.2]{FKKLR}}]\label{thm_Tr_taunu}
For $\gamma\in\widetilde{\Gamma}^\circ$, the operator $\tau_\nu(\gamma)$ on $\cH_\nu(\bD)$ is of trace class, and we have 
\[ \Tr(\tau_\nu(\gamma))=\frac{\mu(\gamma)^\nu}{1-\mu(\gamma)^2}, \]
where $\mu(\gamma)$ is the eigenvalue of $\gamma$ such that $|\mu(\gamma)|<1$. 
\end{theorem}


\begin{example}\label{ex_Yg}
For $g\in\BR$, $u\in(0,1)$, let 
\begin{align*}
Y(g):\hspace{-3pt}&=\begin{bmatrix} -\cosh(2g)&-\sinh(2g) \\ \sinh(2g)&\cosh(2g) \end{bmatrix} \\
&=\begin{bmatrix} \cosh(g)&-\sinh(g) \\ -\sinh(g)&\cosh(g) \end{bmatrix}
\begin{bmatrix}-1&0\\0&1\end{bmatrix}\begin{bmatrix} \cosh(g)&\sinh(g) \\ \sinh(g)&\cosh(g) \end{bmatrix}\in\fg^\BC, \\
u^{Y(g)}:\hspace{-3pt}&=\exp(Y(g)\log u) \\
&=\begin{bmatrix} \cosh(g)&-\sinh(g) \\ -\sinh(g)&\cosh(g) \end{bmatrix}
\begin{bmatrix}u^{-1}&0\\0&u\end{bmatrix}\begin{bmatrix} \cosh(g)&\sinh(g) \\ \sinh(g)&\cosh(g) \end{bmatrix}\in\Gamma^\circ, 
\end{align*}
so that 
\begin{align*}
\rd\tau_\nu(Y(g))&=\cosh(2g)\biggl(2z\frac{\rd}{\rd z}+\nu\biggr)+\sinh(2g)\biggl((1+z^2)\frac{\rd}{\rd z}+\nu z\biggr), \\
\tau_\nu(u^{Y(g)})&=\exp\biggl(\log u\biggl(\cosh(2g)\biggl(2z\frac{\rd}{\rd z}+\nu\biggr)+\sinh(2g)\biggl((1+z^2)\frac{\rd}{\rd z}+\nu z\biggr)\biggr)\biggr) 
\end{align*}
are self-adjoint on $\cH_\nu(\bD)$. Especially, for $\nu=1/2, 3/2$, we have 
\begin{align}
\rd\tau_{1/2}(Y(g))&=\cC\cL_{1/2}\biggl(\cosh(2g)\biggl(\ra^\dagger\ra+\frac{1}{2}\biggr)+\frac{\sinh(2g)}{2}(\ra^2+(\ra^\dagger)^2)\biggr)\biggr|_{L^2(\BR)_\even}\cC\cL_{1/2}^{-1}, \notag\\
\rd\tau_{3/2}(Y(g))&=\cC\cL_{3/2}\biggl(\cosh(2g)\biggl(\ra^\dagger\ra+\frac{1}{2}\biggr)+\frac{\sinh(2g)}{2}(\ra^2+(\ra^\dagger)^2)\biggr)\biggr|_{L^2(\BR)_\odd}\cC\cL_{3/2}^{-1}. \label{formula_dtauYg}
\end{align}
The sets of eigenvalues of $\rd\tau_\nu(Y(g))$ and $\tau_\nu(u^{Y(g)})$ are $2\BZ_{\ge 0}+\nu$ and $\{u^{2k+\nu}\mid k\in\BZ_{\ge 0}\}$ respectively, with the common eigenvectors 
$\Bigl\{\tau_\nu\Bigl(\Bigl[\begin{smallmatrix}\cosh(g)&-\sinh(g) \\ -\sinh(g)&\cosh(g)\end{smallmatrix}\Bigr]\Bigr)z^k \Bigm| k\in\BZ_{\ge 0}\Bigr\}$. 
Especially, for $\lambda\in\BC$, $\lambda\notin-(2\BZ_{\ge 0}+\nu)$, the resolvent $(\rd\tau_\nu(Y(g))+\lambda)^{-1}$ is of Hilbert--Schmidt class, with 
\[ \bigl\Vert(\rd\tau_\nu(Y(g))+\lambda)^{-1}\bigr\Vert_\HS^2=\sum_{k=0}^\infty\frac{1}{|2k+\lambda+\nu|^2}, \quad 
\bigl\Vert(\rd\tau_\nu(Y(g))+\lambda)^{-1}\bigr\Vert_\op=\max_{k\in\BZ_{\ge 0}}\frac{1}{|2k+\lambda+\nu|}, \]
and for $n\in\BZ_{\ge 2}$, we have 
\[ \Tr\bigl((\rd\tau_\nu(Y(g))+\lambda)^{-n}\bigr)=\sum_{k=0}^\infty\frac{1}{(2k+\lambda+\nu)^n}=2^{-n}\zeta\biggl(n,\frac{\lambda+\nu}{2}\biggr), \]
where $\zeta(s,\lambda)$ is the Hurwitz zeta function. Similarly, $\tau_\nu(u^{Y(g)})$ is of trace class, with 
\begin{gather*}
\Tr\bigl(\tau_\nu(u^{Y(g)})\bigr)=\bigl\Vert\tau_\nu(u^{Y(g)})\bigr\Vert_{\Tr}=\sum_{k=0}^\infty u^{2k+\nu}=\frac{u^\nu}{1-u^2}, \\
\bigl\Vert\tau_\nu(u^{Y(g)})\bigr\Vert_\HS^2=\sum_{k=0}^\infty (u^{2k+\nu})^2=\frac{u^{2\nu}}{1-u^4}, \qquad 
\bigl\Vert\tau_\nu(u^{Y(g)})\bigr\Vert_\op=\max_{k\in\BZ_{\ge 0}} u^{2k+\nu}=u^\nu. 
\end{gather*}
\end{example}

\subsection{Hamiltonians treated in this article}

In this article, we consider the renormalized 1-photon quantum Rabi model 
\[ \widetilde{\rH}_{1\QRM}^{(g,\Delta,\varepsilon)}:=\rH_{1\QRM}^{(g,\Delta,\varepsilon)}+g^2\bI
=\bI(\ra^\dagger\ra+g^2)+g\bsigma_1(\ra+\ra^\dagger)+\Delta\bsigma_3+\varepsilon\bsigma_1 \]
(where $g,\Delta,\varepsilon\in\BR$, $\bI=\bigl[\begin{smallmatrix}1&0\\0&1\end{smallmatrix}\bigr]$, $\bsigma_1=\bigl[\begin{smallmatrix}0&1\\1&0\end{smallmatrix}\bigr]$,  
$\bsigma_2=\bigl[\begin{smallmatrix}0&-\sqrt{-1}\\\sqrt{-1}&0\end{smallmatrix}\bigr]$, $\bsigma_3=\bigl[\begin{smallmatrix}1&0\\0&-1\end{smallmatrix}\bigr]$ are the Pauli matrices), 
the renormalized 2-photon quantum Rabi model 
\begin{align*}
\widetilde{\rH}_{2\QRM}^{(g,\Delta,\varepsilon)}:\hspace{-3pt}&=\cosh(2g)\biggl(\rH_{2\QRM}^{(\tanh(2g)/2,\Delta\sech(2g),\varepsilon\sech(2g))}+\frac{1}{2}\bI\biggr) \\
&=\cosh(2g)\bI\biggl(\ra^\dagger\ra+\frac{1}{2}\biggr)+\frac{\sinh(2g)}{2}\bsigma_1(\ra^2+(\ra^\dagger)^2)+\Delta\bsigma_3+\varepsilon\bsigma_1 
\end{align*}
(where $g,\Delta,\varepsilon\in\BR$, $\bI,\bsigma_1,\bsigma_2,\bsigma_3$ are as above), 
and the non-commutative harmonic oscillator 
\begin{align*}
&\widetilde{\rH}_\NCHO^{(\alpha,\beta,\eta)}:=\frac{\alpha+\beta}{2\sqrt{\alpha\beta(\alpha\beta-1)}}\rH_\NCHO^{(\alpha,\beta,\eta)} \\
&=\frac{\alpha+\beta}{2\sqrt{\alpha\beta(\alpha\beta-1)}} \\
&\eqspace{}\times\biggl(\begin{bmatrix}\alpha&0\\0&\beta\end{bmatrix}\biggl(\ra^\dagger\ra+\frac{1}{2}\biggr)
+\frac{1}{2}\begin{bmatrix}0&-1\\1&0\end{bmatrix}(\ra^2-(\ra^\dagger)^2)+2\eta\sqrt{\alpha\beta-1}\sqrt{-1}\begin{bmatrix}0&-1\\1&0\end{bmatrix}\biggr) 
\end{align*}
(where $\alpha,\beta>0$, $\alpha\beta>1$, $\eta\in\BR$). 
For $g,\Delta,\varepsilon\in\BR$, we define the operator $\widetilde{\rH}_\flat^{(g,\Delta,\varepsilon)}$ on $\cH_\flat(\BC)$ by 
\begin{align}
\widetilde{\rH}_{\flat}^{(g,\Delta,\varepsilon)}:\hspace{-3pt}&=\bI\biggl(w\frac{\rd}{\rd w}+g^2\biggr)+g\bsigma_3\biggl(\frac{\rd}{\rd w}+w\biggr)+\Delta\bsigma_1+\varepsilon\bsigma_3 \notag\\
&=\begin{bmatrix} \rd\tau_\flat(X(g))+\varepsilon & \Delta \\ \Delta & \rd\tau_\flat(X(-g))-\varepsilon \end{bmatrix}, \label{def_Hflat}
\end{align}
where $X(g)\in\fg^\BC_\flat$ is as in Example \ref{ex_Xg}. Then since 
\[ \widetilde{\rH}_{\flat}^{(g,\Delta,\varepsilon)}=\frac{1}{2}\begin{bmatrix}1&1\\1&-1\end{bmatrix}
\cB\widetilde{\rH}_{1\QRM}^{(g,\Delta,\varepsilon)}\cB^{-1}\begin{bmatrix}1&1\\1&-1\end{bmatrix} \]
holds, we have 
\[ \zeta(\widetilde{\rH}_{1\QRM}^{(g,\Delta,\varepsilon)};s,\lambda)=\zeta(\widetilde{\rH}_{\flat}^{(g,\Delta,\varepsilon)};s,\lambda). \]
Next, for $\nu>0$, $g,\Delta,\varepsilon\in\BR$, we define the operator $\widetilde{\rH}_\nu^{(g,\Delta,\varepsilon)}$ on $\cH_\nu(\bD)$ by 
\begin{align}
\widetilde{\rH}_\nu^{(g,\Delta,\varepsilon)}:\hspace{-3pt}
&=\cosh(2g)\bI\biggl(2z\frac{\rd}{\rd z}+\nu\biggr)+\sinh(2g)\bsigma_3\biggl((1+z^2)\frac{\rd}{\rd z}+\nu z\biggr)+\Delta\bsigma_1+\varepsilon\bsigma_3 \notag\\
&=\begin{bmatrix} \rd\tau_\nu(Y(g))+\varepsilon & \Delta \\ \Delta & \rd\tau_\nu(Y(-g))-\varepsilon \end{bmatrix}, \label{def_Hnu}
\end{align}
where $Y(g)\in\fg^\BC$ is as in Example \ref{ex_Yg}. Then since 
\begin{align*}
\widetilde{\rH}_{1/2}^{(g,\Delta,\varepsilon)}
&=\frac{1}{2}\begin{bmatrix}1&1\\1&-1\end{bmatrix}\cCL_{1/2}\bigl(\widetilde{\rH}_{2\QRM}^{(g,\Delta,\varepsilon)}\bigr|_{L^2(\BR)_{\even}}\bigr)\cCL_{1/2}^{-1}\begin{bmatrix}1&1\\1&-1\end{bmatrix}, \\
\widetilde{\rH}_{3/2}^{(g,\Delta,\varepsilon)}
&=\frac{1}{2}\begin{bmatrix}1&1\\1&-1\end{bmatrix}\cCL_{3/2}\bigl(\widetilde{\rH}_{2\QRM}^{(g,\Delta,\varepsilon)}\bigr|_{L^2(\BR)_{\odd}}\bigr)\cCL_{3/2}^{-1}\begin{bmatrix}1&1\\1&-1\end{bmatrix} 
\end{align*}
hold, we have 
\[ \zeta(\widetilde{\rH}_{2\QRM}^{(g,\Delta,\varepsilon)};s,\lambda)=\zeta(\widetilde{\rH}_{1/2}^{(g,\Delta,\varepsilon)};s,\lambda)+\zeta(\widetilde{\rH}_{3/2}^{(g,\Delta,\varepsilon)};s,\lambda). \]
Similarly, for $\nu,\alpha,\beta>0$, $\alpha\beta>1$, $\eta\in\BR$, we define the operator $\widetilde{\rQ}_\nu^{(\alpha,\beta,\eta)}$ on $\cH_\nu(\bD)$ by 
\begin{align}
\widetilde{\rQ}_\nu^{(\alpha,\beta,\eta)}:\hspace{-3pt}&=\frac{\alpha+\beta}{2\sqrt{\alpha\beta(\alpha\beta-1)}}
\biggl(\begin{bmatrix}\alpha&0\\0&\beta\end{bmatrix}\biggl(2z\frac{\rd}{\rd z}+\nu\biggr) \notag\\
&\eqspace{}+\begin{bmatrix}0&1\\1&0\end{bmatrix}\biggl((1+z^2)\frac{\rd}{\rd z}+\nu z\biggr)+2\eta\sqrt{\alpha\beta-1}\begin{bmatrix}0&1\\1&0\end{bmatrix}\biggr). \label{def_Qnu}
\end{align}
Then again since 
\begin{align*}
&\widetilde{\rQ}_{1/2}^{(\alpha,\beta,\eta)} \\*
&=e^{\pi\sqrt{-1}\bsigma_3/4}\tau_\nu\bigl(e^{\pi\sqrt{-1}\bsigma_3/4}\bigr)
\cCL_{1/2}\bigl(\widetilde{\rH}_\NCHO^{(\alpha,\beta,\eta)}\bigr|_{L^2(\BR)_{\even}}\bigr)\cCL_{1/2}^{-1}\tau_\nu\bigl(e^{-\pi\sqrt{-1}\bsigma_3/4}\bigr)e^{-\pi\sqrt{-1}\bsigma_3/4}, \\
&\widetilde{\rQ}_{3/2}^{(\alpha,\beta,\eta)} \\*
&=e^{\pi\sqrt{-1}\bsigma_3/4}\tau_\nu\bigl(e^{\pi\sqrt{-1}\bsigma_3/4}\bigr)
\cCL_{3/2}\bigl(\widetilde{\rH}_\NCHO^{(\alpha,\beta,\eta)}\bigr|_{L^2(\BR)_{\odd}}\bigr)\cCL_{3/2}^{-1}\tau_\nu\bigl(e^{-\pi\sqrt{-1}\bsigma_3/4}\bigr)e^{-\pi\sqrt{-1}\bsigma_3/4} 
\end{align*}
hold, where $e^{\pm\pi\sqrt{-1}\bsigma_3/4}=\left[\begin{smallmatrix}e^{\pm\pi\sqrt{-1}/4}&0\\0&e^{\mp\pi\sqrt{-1}/4}\end{smallmatrix}\right]$, we have 
\[ \zeta(\widetilde{\rH}_\NCHO^{(\alpha,\beta,\eta)};s,\lambda)=\zeta(\widetilde{\rQ}_{1/2}^{(\alpha,\beta,\eta)};s,\lambda)+\zeta(\widetilde{\rQ}_{3/2}^{(\alpha,\beta,\eta)};s,\lambda). \]

\begin{remark}\label{rem_multivariate}
We can consider multivariate generalizations of the 2-photon quantum Rabi model and the non-commutative harmonic oscillator on $L^2(\BR^n)\otimes\BC^2$ as 
\begin{align*}
\widetilde{\rH}_{2\QRM,\BR^n}^{(g,\Delta,\varepsilon)}:\hspace{-3pt}
&=\cosh(2g)\bI\sum_{j=1}^n\biggl(\ra_j^\dagger\ra_j+\frac{1}{2}\biggr)+\frac{\sinh(2g)}{2}\bsigma_1\sum_{j=1}^n(\ra_j^2+(\ra_j^\dagger)^2)+\Delta\bsigma_3+\varepsilon\bsigma_1, \\
\widetilde{\rH}_{\NCHO,\BR^n}^{(\alpha,\beta,\eta)}:\hspace{-3pt}&=\frac{\alpha+\beta}{2\sqrt{\alpha\beta(\alpha\beta-1)}}
\biggl(\begin{bmatrix}\alpha&0\\0&\beta\end{bmatrix}\sum_{j=1}^n\biggl(\ra_j^\dagger\ra_j+\frac{1}{2}\biggr) \\*
&\eqspace{}+\frac{1}{2}\begin{bmatrix}0&-1\\1&0\end{bmatrix}\sum_{j=1}^n(\ra_j^2-(\ra_j^\dagger)^2) +2\eta\sqrt{\alpha\beta-1}\sqrt{-1}\begin{bmatrix}0&-1\\1&0\end{bmatrix}\biggr). 
\end{align*}
Then by the decomposition of $L^2(\BR^n)$ in Remark \ref{rem_sph_harm}, we have 
\begin{align*}
\zeta\bigl(\widetilde{\rH}_{2\QRM,\BR^n}^{(g,\Delta,\varepsilon)}|_{L^2_{k+\frac{n}{2}}(\BR_{>0})\otimes\cH\cP_k(\BR^n)};s,\lambda\bigr)
&=\zeta\bigl(\widetilde{\rH}_{k+\frac{n}{2}}^{(g,\Delta,\varepsilon)};s,\lambda\bigr) \dim\cH\cP_k(\BR^n), \\
\zeta\bigl(\widetilde{\rH}_{\NCHO,\BR^n}^{(\alpha,\beta,\eta)}|_{L^2_{k+\frac{n}{2}}(\BR_{>0})\otimes\cH\cP_k(\BR^n)};s,\lambda\bigr)
&=\zeta\bigl(\widetilde{\rQ}_{k+\frac{n}{2}}^{(\alpha,\beta,\eta)};s,\lambda\bigr) \dim\cH\cP_k(\BR^n), 
\end{align*}
where $\dim\cH\cP_k(\BR^n)=\bigl(\begin{smallmatrix}k+n-1\\n-1\end{smallmatrix}\bigr) -\bigl(\begin{smallmatrix}k+n-3\\n-1\end{smallmatrix}\bigr)$, and hence we have 
\begin{align*}
\zeta(\widetilde{\rH}_{2\QRM,\BR^n}^{(g,\Delta,\varepsilon)};s,\lambda)
&=\sum_{k=0}^\infty\zeta\bigl(\widetilde{\rH}_{k+\frac{n}{2}}^{(g,\Delta,\varepsilon)};s,\lambda\bigr) \dim\cH\cP_k(\BR^n), \\
\zeta(\widetilde{\rH}_{\NCHO,\BR^n}^{(\alpha,\beta,\eta)};s,\lambda)
&=\sum_{k=0}^\infty\zeta\bigl(\widetilde{\rQ}_{k+\frac{n}{2}}^{(\alpha,\beta,\eta)};s,\lambda\bigr) \dim\cH\cP_k(\BR^n). 
\end{align*}
\end{remark}

Next, for $\nu>1$, we consider the Hilbert space 
\[ \cH_\nu(\sqrt{\nu}\bD):=\biggl\{ f(w)\in\cO(\sqrt{\nu}\bD)\biggm| \frac{\nu-1}{\nu\pi}\int_{\sqrt{\nu}\bD}|f(w)|^2\biggl(1-\frac{|w|^2}{\nu}\biggr)^{\nu-2}\,\rd w<\infty\biggr\}, \]
and the dilation map  
\[ \rho_\nu\colon\cH_\nu(\bD)\longrightarrow\cH_\nu(\sqrt{\nu}\bD), \qquad f(z)\longmapsto f(w/\sqrt{\nu}). \]
Then $\rho_\nu$ is a unitary isomorphism, and the composition 
\begin{align*}
&\rho_\nu\circ\widetilde{\rH}_{\nu}^{(g,\Delta,\varepsilon)}\circ\rho_\nu^{-1} \\
&=\cosh(2g)\bI\biggl(2w\frac{\rd}{\rd w}+\nu\biggr)+\sqrt{\nu}\sinh(2g)\bsigma_3\biggl(\biggl(1+\frac{w^2}{\nu}\biggr)\frac{\rd}{\rd w}+w\biggr)+\Delta\bsigma_1+\varepsilon\bsigma_3 
\end{align*}
defines a $\BC^2$-valued differential operator on $\sqrt{\nu}\bD$. 
Then as a formal differential operator, we have the limit 
\begin{align*}
&\lim_{\nu\to\infty}\frac{1}{2}\bigl(\rho_\nu\circ\widetilde{\rH}_{\nu}^{(g/\sqrt{\nu},2\Delta,2\varepsilon)}\circ\rho_\nu^{-1}-\nu\bigr) \\
&=\lim_{\nu\to\infty}\biggl(\bI\biggl(\cosh\biggl(\frac{2g}{\sqrt{\nu}}\biggr)w\frac{\rd}{\rd w}+\biggl(\cosh\biggl(\frac{2g}{\sqrt{\nu}}\biggr)-1\biggr)\frac{\nu}{2}\biggr) \\*
&\eqspace{}+\frac{\sqrt{\nu}}{2}\sinh\biggl(\frac{2g}{\sqrt{\nu}}\biggr)\bsigma_3\biggl(\biggl(1+\frac{w^2}{\nu}\biggr)\frac{\rd}{\rd w}+w\biggr)+\Delta\bsigma_1+\varepsilon\bsigma_3\biggr) \\
&=\bI\biggl(w\frac{\rd}{\rd w}+g^2\biggr)+g\bsigma_3\biggl(\frac{\rd}{\rd w}+w\biggr)+\Delta\bsigma_1+\varepsilon\bsigma_3=\widetilde{\rH}_{\flat}^{(g,\Delta,\varepsilon)}. 
\end{align*}
Hence it is expected that 
\[ 2^s\lim_{\nu\to\infty}\zeta\bigl(\widetilde{\rH}_{\nu}^{(g/\sqrt{\nu},2\Delta,2\varepsilon)};s,2\lambda-\nu\bigr)
=\zeta\bigl(\widetilde{\rH}_{\flat}^{(g,\Delta,\varepsilon)};s,\lambda\bigr) \]
holds. We verify this for $s=n\in\BZ_{\ge 2}$ and for suitable $\lambda, g,\Delta,\varepsilon$ later in Corollary \ref{cor_zeta}\,(3).

\section{Power series expansion of special values of spectral zeta functions}\label{section_zeta_sum}

In this section, we express the spectral zeta functions $\zeta(\widetilde{\rH}_\flat^{(g,\Delta,\varepsilon)};n,\lambda)$, $\zeta(\widetilde{\rH}_\nu^{(g,\Delta,\varepsilon)};n,\lambda)$ 
and $\zeta(\widetilde{\rQ}_\nu^{(\alpha,\beta,\eta)};n,\lambda)$ for $n\in\BZ_{\ge 2}$ of $\widetilde{\rH}_\flat^{(g,\Delta,\varepsilon)}$, 
$\widetilde{\rH}_\nu^{(g,\Delta,\varepsilon)}$ and $\widetilde{\rQ}_\nu^{(\alpha,\beta,\eta)}$ given in (\ref{def_Hflat}), (\ref{def_Hnu}) and (\ref{def_Qnu}) 
as power series of $\Delta$ or $\frac{\alpha-\beta}{\alpha+\beta}$. 
To do this, for $g\in\BR$, $\lambda,\varepsilon\in\BC$, let $\rh_\pm^\flat(\lambda;g,\varepsilon)$ be the operator on $\cH_\flat(\BC)$ given by 
\[ \rh_\pm^\flat(\lambda;g,\varepsilon):=\rd\tau_\flat(X(\pm g))\pm\varepsilon+\lambda
=w\frac{\rd}{\rd w}\pm g\biggl(\frac{\rd}{\rd w}+w\biggr)+g^2\pm\varepsilon+\lambda, \]
where $X(g)\in\fg^\BC_\flat$ is as in Example \ref{ex_Xg}, and if $\lambda\pm\varepsilon\notin -\BZ_{\ge 0}$, then for $m\in\BZ_{\ge 1}$, let 
\begin{align*}
R_m^\flat(\lambda;g,\varepsilon)&:=\Tr\bigl(\bigl(\rh_+^\flat(\lambda;g,\varepsilon)^{-1}\rh_-^\flat(\lambda;g,\varepsilon)^{-1}\bigr)^m\bigr), \\
C_\flat(\lambda,\varepsilon)&:=\max_{\delta\in\{\pm \}}\bigl\Vert\rh_\delta^\flat(\lambda;g,\varepsilon)^{-1}\bigr\Vert_\op
=\max_{\delta\in\{\pm \},\ k\in\BZ_{\ge 0}}\frac{1}{|k+\lambda+\delta\varepsilon|}, \\
C'_\flat(\lambda,\varepsilon)&:=\max_{\delta\in\{\pm \}}\bigl\Vert\rh_\delta^\flat(\lambda;g,\varepsilon)^{-1}\bigr\Vert_\HS
=\max_{\delta\in\{\pm \}}\biggl(\sum_{k=0}^\infty\frac{1}{|k+\lambda+\delta\varepsilon|^2}\biggr)^{1/2}. 
\end{align*}
Similarly, for $\nu\in\BR_{>0}$, $g\in\BR$, $\lambda,\varepsilon\in\BC$, let $\rh_\pm^\nu(\lambda;g,\varepsilon)$ be the operator on $\cH_\nu(\bD)$ given by 
\begin{align*}
\rh_\pm^\nu(\lambda;g,\varepsilon):\hspace{-3pt}&=\rd\tau_\nu(Y(\pm g))\pm\varepsilon+\lambda \\
&=\cosh(2g)\biggl(2z\frac{\rd}{\rd z}+\nu\biggr)\pm\sinh(2g)\biggl((1+z^2)\frac{\rd}{\rd z}+\nu z\biggr)\pm\varepsilon+\lambda, 
\end{align*}
where $Y(g)\in\fg^\BC$ is as in Example \ref{ex_Yg}, and if $\lambda\pm\varepsilon\notin -(2\BZ_{\ge 0}+\nu)$, then for $m\in\BZ_{\ge 1}$, let 
\begin{align*}
R_m^\nu(\lambda;g,\varepsilon)&:=\Tr\bigl(\bigl(\rh_+^\nu(\lambda;g,\varepsilon)^{-1}\rh_-^\nu(\lambda;g,\varepsilon)^{-1}\bigr)^m\bigr), \\
C_\nu(\lambda,\varepsilon)&:=\max_{\delta\in\{\pm \}}\bigl\Vert\rh_\delta^\nu(\lambda;g,\varepsilon)^{-1}\bigr\Vert_\op
=\max_{\delta\in\{\pm \},\ k\in\BZ_{\ge 0}}\frac{1}{|2k+\lambda+\delta\varepsilon+\nu|}, \\
C'_\nu(\lambda,\varepsilon)&:=\max_{\delta\in\{\pm \}}\bigl\Vert\rh_\delta^\nu(\lambda;g,\varepsilon)^{-1}\bigr\Vert_\HS
=\max_{\delta\in\{\pm \}}\biggl(\sum_{k=0}^\infty\frac{1}{|2k+\lambda+\delta\varepsilon+\nu|^2}\biggr)^{1/2}. 
\end{align*}
Then we have 
\[ 2C_\nu(2\lambda-\nu,2\varepsilon)=C_\flat(\lambda,\varepsilon), \qquad 2C_\nu'(2\lambda-\nu,2\varepsilon)=C_\flat'(\lambda,\varepsilon). \]

First, we prove the following. 
\begin{lemma}\label{lem_diff_R}
Let $\nu\in\BR_{>0}\cup\{\flat\}$, $g\in\BR$, $\lambda,\varepsilon\in\BC$, $m\in\BZ_{\ge 1}$. 
We assume $\lambda\pm\varepsilon\notin -\BZ_{\ge 0}$ when $\nu=\flat$, and $\lambda\pm\varepsilon\notin -(2\BZ_{\ge 0}+\nu)$ when $\nu\in\BR_{>0}$. 
Then $R_m^\nu(\lambda;g,\varepsilon)$ is holomorphic with respect to $\lambda$, and for $n\in\BZ_{\ge 0}$, we have 
\begin{align*}
&\frac{1}{n!}\frac{\partial^n R_m^\nu}{\partial\lambda^n}(\lambda;g,\varepsilon)
=(-1)^n\sum_{\substack{\bn\in(\BZ_{\ge 0})^{2m} \\ |\bn|=n}}
\Tr\biggl(\prod_{1\le j\le m}^{\longrightarrow}(\rh_+^\nu(\lambda;g,\varepsilon)^{-n_{2j-1}-1}\rh_-^\nu(\lambda;g,\varepsilon)^{-n_{2j}-1})\biggr) \\
&=\frac{(-1)^n}{2}\sum_{\substack{\bn\in(\BZ_{\ge 0})^{2m} \\ |\bn|=n}}
\Tr\biggl(\prod_{1\le j\le 2m}^{\longrightarrow}\biggl(\begin{bmatrix} \rh_+^\nu(\lambda;g,\varepsilon)^{-1} & 0 \\ 0 & \rh_-^\nu(\lambda;g,\varepsilon)^{-1} \end{bmatrix}^{n_j+1}
\begin{bmatrix}0&1\\1&0\end{bmatrix}\biggr)\biggr). 
\end{align*}
\end{lemma}

\begin{proof}
In the following, we abbreviate $\rh_\pm^\nu(\lambda;g,\varepsilon)=:\rh_\pm$. 
Then for $|\mu|<C_\nu(\lambda,\varepsilon)^{-1}$, we have 
\begin{align*}
R_m^\nu(\lambda+\mu;g,\varepsilon)&=\Tr\bigl(\bigl((\rh_++\mu)^{-1}(\rh_-+\mu)^{-1}\bigr)^m\bigr) \\
&=\Tr\bigl(\bigl(\rh_+^{-1}(1+\mu\rh_+^{-1})^{-1}\rh_-^{-1}(1+\mu\rh_-^{-1})^{-1}\bigr)^m\bigr) \\
&=\Tr\biggl(\biggl(\biggl(\sum_{n=0}^\infty (-\mu)^n\rh_+^{-n-1}\biggr)\biggl(\sum_{n'=0}^\infty (-\mu)^{n'}\rh_-^{-n'-1}\biggr)\biggr)^m\biggr) \\
&=\Tr\biggl(\sum_{n=0}^\infty (-\mu)^n \sum_{\substack{\bn\in(\BZ_{\ge 0})^{2m} \\ |\bn|=n}} \prod_{1\le j\le m}^{\longrightarrow} (\rh_+^{-n_{2j-1}-1}\rh_-^{-n_{2j}-1})\biggr), 
\end{align*}
and since 
\begin{align*}
&\biggl\Vert\sum_{\substack{\bn\in(\BZ_{\ge 0})^{2m} \\ |\bn|=n}} \prod_{1\le j\le m}^{\longrightarrow} (\rh_+^{-n_{2j-1}-1}\rh_-^{-n_{2j}-1})\biggr\Vert_{\Tr} \\
&\le \binom{2m+n-1}{n}\max\{\Vert\rh_+^{-1}\Vert_\HS,\Vert\rh_-^{-1}\Vert_\HS\}^2\max\{\Vert\rh_+^{-1}\Vert_\op,\Vert\rh_-^{-1}\Vert_\op\}^{2m+n-2} \\
&=\frac{(2m)_n}{n!}C_\nu(\lambda,\varepsilon)^{2m+n-2}C'_\nu(\lambda,\varepsilon)^2
\end{align*}
holds, the series 
\[ R_m^\nu(\lambda+\mu;g,\varepsilon)
=\sum_{n=0}^\infty (-\mu)^n \sum_{\substack{\bn\in(\BZ_{\ge 0})^{2m} \\ |\bn|=n}} \Tr\biggl(\prod_{1\le j\le m}^{\longrightarrow} (\rh_+^{-n_{2j-1}-1}\rh_-^{-n_{2j}-1})\biggr) \]
converges absolutely. Hence this is analytic at $\mu=0$, and combining this with the Taylor expansion 
\[ R_m^\nu(\lambda+\mu;g,\varepsilon)=\sum_{n=0}^\infty \frac{1}{n!}\frac{\partial^n R_m^\nu}{\partial\lambda^n}(\lambda;g,\varepsilon)\mu^n, \]
we get the 1st equality. Also, by 
\begin{align*}
&\sum_{\substack{\bn\in(\BZ_{\ge 0})^{2m} \\ |\bn|=n}}
\Tr\biggl(\prod_{1\le j\le 2m}^{\longrightarrow}\biggl(\begin{bmatrix} \rh_+^{-1} & 0 \\ 0 & \rh_-^{-1} \end{bmatrix}^{n_j+1}
\begin{bmatrix}0&1\\1&0\end{bmatrix}\biggr)\biggr) \\
&=\sum_{\substack{\bn\in(\BZ_{\ge 0})^{2m} \\ |\bn|=n}} \Tr\biggl(\prod_{1\le j\le m}^{\longrightarrow} (\rh_+^{-n_{2j-1}-1}\rh_-^{-n_{2j}-1})\biggr)
+\sum_{\substack{\bn\in(\BZ_{\ge 0})^{2m} \\ |\bn|=n}} \Tr\biggl(\prod_{1\le j\le m}^{\longrightarrow} (\rh_-^{-n_{2j-1}-1}\rh_+^{-n_{2j}-1})\biggr)
\end{align*}
and 
\begin{align*}
&\sum_{\substack{\bn\in(\BZ_{\ge 0})^{2m} \\ |\bn|=n}} \Tr\biggl(\prod_{1\le j\le m}^{\longrightarrow} (\rh_-^{-n_{2j-1}-1}\rh_+^{-n_{2j}-1})\biggr) \\
&=\sum_{\substack{\bn\in(\BZ_{\ge 0})^{2m} \\ |\bn|=n}} \Tr\biggl(\rh_+^{-n_{2m}-1}\prod_{1\le j\le m-1}^{\longrightarrow} (\rh_-^{-n_{2j-1}-1}\rh_+^{-n_{2j}-1})\rh_-^{-n_{2m-1}-1}\biggr) \\
&=\sum_{\substack{\bn\in(\BZ_{\ge 0})^{2m} \\ |\bn|=n}} \Tr\biggl(\prod_{1\le j\le m}^{\longrightarrow} (\rh_+^{-n_{2j-1}-1}\rh_-^{-n_{2j}-1})\biggr), 
\end{align*}
we get the 2nd equality. 
\end{proof}

The next theorem gives the expansions of the spectral zeta functions. 
\begin{theorem}\label{thm_zeta_sum}
\begin{enumerate}
\item Let $g,\Delta,\varepsilon\in\BR$, $\lambda\in\BC$. We assume $\lambda\pm\varepsilon\notin -\BZ_{\ge 0}$ and $|\Delta|<C_\flat(\lambda,\varepsilon)^{-1}$. 
Then for $n\in\BZ_{\ge 2}$, $(\widetilde{\rH}_\flat^{(g,\Delta,\varepsilon)}+\lambda)^{-n}$ is of trace class, and we have 
\begin{align*}
\zeta\bigl(\widetilde{\rH}_\flat^{(g,\Delta,\varepsilon)};n,\lambda\bigr)
=\zeta(n,\lambda+\varepsilon)+\zeta(n,\lambda-\varepsilon)
+\frac{(-1)^n}{(n-1)!}\sum_{m=1}^\infty \frac{\Delta^{2m}}{m}\frac{\partial^n R_m^\flat}{\partial\lambda^n}(\lambda;g,\varepsilon). 
\end{align*}
\item Let $\nu\in\BR_{>0}$, $g,\Delta,\varepsilon\in\BR$, $\lambda\in\BC$. 
We assume $\lambda\pm\varepsilon\notin -(2\BZ_{\ge 0}+\nu)$ and $|\Delta|<C_\nu(\lambda,\varepsilon)^{-1}$. 
Then for $n\in\BZ_{\ge 2}$, $(\widetilde{\rH}_\nu^{(g,\Delta,\varepsilon)}+\lambda)^{-n}$ is of trace class, and we have 
\begin{align*}
&\zeta\bigl(\widetilde{\rH}_\nu^{(g,\Delta,\varepsilon)};n,\lambda\bigr) \\
&=2^{-n}\zeta\biggl(n,\frac{\lambda+\varepsilon+\nu}{2}\biggr)+2^{-n}\zeta\biggl(n,\frac{\lambda-\varepsilon+\nu}{2}\biggr)
+\frac{(-1)^n}{(n-1)!}\sum_{m=1}^\infty \frac{\Delta^{2m}}{m}\frac{\partial^n R_m^\nu}{\partial\lambda^n}(\lambda;g,\varepsilon). 
\end{align*}
\item Let $\nu\in\BR_{>0}$, $\alpha,\beta,\eta\in\BR$, $\lambda\in\BC$. We assume $\alpha\beta>1$, $\lambda\pm 2\eta\notin -(2\BZ_{\ge 0}+\nu)$ 
and $\bigl|\lambda\frac{\alpha-\beta}{\alpha+\beta}\bigr|C_\nu(\lambda,2\eta)<1$. 
Then for $n\in\BZ_{\ge 2}$, $(\widetilde{\rQ}_\nu^{(\alpha,\beta,\eta)}+\lambda)^{-n}$ is of trace class, and we have 
\begin{multline*}
\zeta\bigl(\widetilde{\rQ}_\nu^{(\alpha,\beta,\eta)};n,\lambda\bigr)
=2^{-n}\zeta\biggl(n,\frac{\lambda+2\eta+\nu}{2}\biggr)+2^{-n}\zeta\biggl(n,\frac{\lambda-2\eta+\nu}{2}\biggr) \\*
{}+\frac{(-1)^n}{(n-1)!}\sum_{m=1}^\infty \frac{1}{m}\biggl(\frac{\alpha-\beta}{\alpha+\beta}\biggr)^{2m}
\frac{\partial^n}{\partial \lambda^n}\lambda^{2m}R_m^\nu\biggl(\lambda;\frac{1}{2}\artanh\biggl(\frac{1}{\sqrt{\alpha\beta}}\biggr),2\eta\biggr). 
\end{multline*}
\end{enumerate}
\end{theorem}

\begin{proof}
(1), (2) We abbreviate $\rh_\pm^\nu(\lambda;g,\varepsilon)=:\rh_\pm$. For $\nu\in\BR_{>0}\cup\{\flat\}$, since we have 
\[ \widetilde{\rH}_\nu^{(g,\Delta,\varepsilon)}+\lambda=\begin{bmatrix} \rh_+^\nu(\lambda;g,\varepsilon) & \Delta \\ \Delta & \rh_-^\nu(\lambda;g,\varepsilon) \end{bmatrix}, \]
for $|\Delta|<C_\nu(\lambda,\varepsilon)^{-1}$, its inverse is given by 
\begin{align*}
(\widetilde{\rH}_\nu^{(g,\Delta,\varepsilon)}+\lambda)^{-1}&=\begin{bmatrix} \rh_+ & \Delta \\ \Delta & \rh_- \end{bmatrix}^{-1}
=\biggl(\bI+\begin{bmatrix}\rh_+^{-1}&0\\0&\rh_-^{-1}\end{bmatrix}\begin{bmatrix}0&\Delta\\\Delta&0\end{bmatrix}\biggr)^{-1}\begin{bmatrix}\rh_+^{-1}&0\\0&\rh_-^{-1}\end{bmatrix} \\
&=\sum_{m=0}^\infty \biggl(-\Delta\begin{bmatrix}\rh_+^{-1}&0\\0&\rh_-^{-1}\end{bmatrix}\begin{bmatrix}0&1\\1&0\end{bmatrix}\biggr)^m\begin{bmatrix}\rh_+^{-1}&0\\0&\rh_-^{-1}\end{bmatrix}. 
\end{align*}
Then the coefficient of $(-\Delta)^m$ of its $n$-th power is the sum of the products of 
$m$ times of $\Bigl[\begin{smallmatrix}\rh_+^{-1}&0\\0&\rh_-^{-1}\end{smallmatrix}\Bigr]\bigl[\begin{smallmatrix}0&1\\1&0\end{smallmatrix}\bigr]$s 
and $n$ times of $\Bigl[\begin{smallmatrix}\rh_+^{-1}&0\\0&\rh_-^{-1}\end{smallmatrix}\Bigr]$s, 
where the rightmost term is always $\Bigl[\begin{smallmatrix}\rh_+^{-1}&0\\0&\rh_-^{-1}\end{smallmatrix}\Bigr]$. Namely, the $n$-th power is computed as 
\begin{align*}
&(\widetilde{\rH}_\nu^{(g,\Delta,\varepsilon)}+\lambda)^{-n} \\
&=\sum_{m=0}^\infty (-\Delta)^m\sum_{\substack{\bn\in(\BZ_{\ge 0})^{m+1} \\ |\bn|=n-1}}\prod_{1\le j\le m}^{\longrightarrow}\biggl(\begin{bmatrix}\rh_+^{-1}&0\\0&\rh_-^{-1}\end{bmatrix}^{n_j+1} 
\begin{bmatrix}0&1\\1&0\end{bmatrix}\biggr)\begin{bmatrix}\rh_+^{-1}&0\\0&\rh_-^{-1}\end{bmatrix}^{n_{m+1}+1}. 
\end{align*}
Then its trace class norm is bounded as 
\begin{align*}
&\bigl\Vert(\widetilde{\rH}_\nu^{(g,\Delta,\varepsilon)}+\lambda)^{-n}\bigr\Vert_{\Tr} \\
&\le \sum_{m=0}^\infty |\Delta|^m\sum_{\substack{\bn\in(\BZ_{\ge 0})^{m+1} \\ |\bn|=n-1}}
\Biggl\Vert\prod_{1\le j\le m}^{\longrightarrow}\biggl(\begin{bmatrix}\rh_+^{-1}&0\\0&\rh_-^{-1}\end{bmatrix}^{n_j+1} 
\begin{bmatrix}0&1\\1&0\end{bmatrix}\biggr)\begin{bmatrix}\rh_+^{-1}&0\\0&\rh_-^{-1}\end{bmatrix}^{n_{m+1}+1}\Biggr\Vert_{\Tr} \\
&\le \sum_{m=0}^\infty\binom{m+n-1}{m}|\Delta|^m\biggl\Vert\begin{bmatrix}\rh_+^{-1}&0\\0&\rh_-^{-1}\end{bmatrix}\biggr\Vert_\op^{m+n-2}
\biggl\Vert\begin{bmatrix}\rh_+^{-1}&0\\0&\rh_-^{-1}\end{bmatrix}\biggr\Vert_\HS^2 \\
&\le \sum_{m=0}^\infty\frac{(n)_m}{m!}|\Delta|^m C_\nu(\lambda,\varepsilon)^{m+n-2}\cdot 2C'_\nu(\lambda,\varepsilon)^2
=\frac{2C_\nu(\lambda,\varepsilon)^{n-2}C'_\nu(\lambda,\varepsilon)^2}{(1-C_\nu(\lambda,\varepsilon)|\Delta|)^n}<\infty. 
\end{align*}
Next, let $C_m=\BZ/m\BZ$ be the cyclic group, which acts on $(\BZ_{\ge 0})^m$ by the shift of indices. Then by Lemma \ref{lem_diff_R}, for $m\ge 1$ we have 
\begin{align}
&\sum_{\substack{\bn\in(\BZ_{\ge 0})^{m+1} \\ |\bn|=n-1}}\Tr\biggl(\prod_{1\le j\le m}^{\longrightarrow}\biggl(\begin{bmatrix}\rh_+^{-1}&0\\0&\rh_-^{-1}\end{bmatrix}^{n_j+1} 
\begin{bmatrix}0&1\\1&0\end{bmatrix}\biggr)\begin{bmatrix}\rh_+^{-1}&0\\0&\rh_-^{-1}\end{bmatrix}^{n_{m+1}+1}\biggr) \notag\\
&=\sum_{\substack{\bn\in(\BZ_{\ge 0})^{m+1} \\ |\bn|=n-1}}\Tr\biggl(\begin{bmatrix}\rh_+^{-1}&0\\0&\rh_-^{-1}\end{bmatrix}^{n_1+n_{m+1}+2}\begin{bmatrix}0&1\\1&0\end{bmatrix}
\prod_{2\le j\le m}^{\longrightarrow}\biggl(\begin{bmatrix}\rh_+^{-1}&0\\0&\rh_-^{-1}\end{bmatrix}^{n_j+1}\begin{bmatrix}0&1\\1&0\end{bmatrix}\biggr)\biggr) \notag\\
&=\sum_{\substack{\bn\in(\BZ_{\ge 0})^m \\ |\bn|=n}}
n_1\Tr\biggl(\prod_{1\le j\le m}^{\longrightarrow} \biggl(\begin{bmatrix}\rh_+^{-1}&0\\0&\rh_-^{-1}\end{bmatrix}^{n_j+1}\begin{bmatrix}0&1\\1&0\end{bmatrix}\biggr)\biggr) \notag\\
&=\sum_{\substack{\bn\in(\BZ_{\ge 0})^m/C_m \\ |\bn|=n}}
\biggl(\sum_{j=1}^m n_j\biggr)\Tr\biggl(\prod_{1\le j\le m}^{\longrightarrow} \biggl(\begin{bmatrix}\rh_+^{-1}&0\\0&\rh_-^{-1}\end{bmatrix}^{n_j+1}\begin{bmatrix}0&1\\1&0\end{bmatrix}\biggr)\biggr) 
\notag\\
&=n\sum_{\substack{\bn\in(\BZ_{\ge 0})^m/C_m \\ |\bn|=n}}
\Tr\biggl(\prod_{1\le j\le m}^{\longrightarrow} \biggl(\begin{bmatrix}\rh_+^{-1}&0\\0&\rh_-^{-1}\end{bmatrix}^{n_j+1}\begin{bmatrix}0&1\\1&0\end{bmatrix}\biggr)\biggr) \notag\\
&=\frac{n}{m}\sum_{\substack{\bn\in(\BZ_{\ge 0})^m \\ |\bn|=n}}
\Tr\biggl(\prod_{1\le j\le m}^{\longrightarrow} \biggl(\begin{bmatrix}\rh_+^{-1}&0\\0&\rh_-^{-1}\end{bmatrix}^{n_j+1}\begin{bmatrix}0&1\\1&0\end{bmatrix}\biggr)\biggr) \notag\\
&=\begin{cases} \displaystyle \frac{(-1)^n}{(n-1)!m'}\frac{\partial^n R_{m'}^\nu}{\partial\lambda^n}(\lambda;g,\varepsilon) & (m=2m') \\
0 & (m=2m'+1). \end{cases} \label{formula_sumtrace}
\end{align}
Hence we get 
\begin{align*}
&\zeta\bigl(\widetilde{\rH}_\nu^{(g,\Delta,\varepsilon)};n,\lambda\bigr)=\Tr\bigl((\widetilde{\rH}_\nu^{(g,\Delta,\varepsilon)}+\lambda)^{-n}\bigr) \\
&=\sum_{m=0}^\infty (-\Delta)^m\sum_{\substack{\bn\in(\BZ_{\ge 0})^{m+1} \\ |\bn|=n-1}}\Tr\biggl(\prod_{1\le j\le m}^{\longrightarrow}\biggl(\begin{bmatrix}\rh_+^{-1}&0\\0&\rh_-^{-1}\end{bmatrix}^{n_j+1} 
\begin{bmatrix}0&1\\1&0\end{bmatrix}\biggr)\begin{bmatrix}\rh_+^{-1}&0\\0&\rh_-^{-1}\end{bmatrix}^{n_{m+1}+1}\biggr) \\
&=\Tr\biggl(\begin{bmatrix}\rh_+^{-1}&0\\0&\rh_-^{-1}\end{bmatrix}^n\biggr)+\frac{(-1)^n}{(n-1)!}\sum_{m=1}^\infty \frac{\Delta^{2m}}{m}\frac{\partial^n R_m^\nu}{\partial\lambda^n}(\lambda;g,\varepsilon). 
\end{align*}
Then since 
\begin{align*}
\Tr\biggl(\begin{bmatrix}\rh_+^{-1}&0\\0&\rh_-^{-1}\end{bmatrix}^n\biggr)&=\Tr(\rh_+^\nu(\lambda;g,\varepsilon)^{-n})+\Tr(\rh_-^\nu(\lambda;g,\varepsilon)^{-n}) \\
&=\begin{cases} \zeta(n,\lambda+\varepsilon)+\zeta(n,\lambda-\varepsilon) & (\nu=\flat), \\
2^{-n}\zeta\bigl(n,\frac{1}{2}(\lambda+\varepsilon+\nu)\bigr)+2^{-n}\zeta\bigl(n,\frac{1}{2}(\lambda-\varepsilon+\nu)\bigr) & (\nu\in\BR_{>0}) \end{cases}
\end{align*}
by Examples \ref{ex_Xg}, \ref{ex_Yg}, we get the desired formula. 

(3) Let 
\[ B:=\frac{1}{\sqrt{2}}\begin{bmatrix}1&1\\1&-1\end{bmatrix}\begin{bmatrix}\sqrt{\alpha}&0\\0&\sqrt{\beta}\end{bmatrix}. \]
Then $\widetilde{\rQ}_\nu^{(\alpha,\beta,\eta)}+\lambda$ is written as 
\begin{align*}
&\widetilde{\rQ}_\nu^{(\alpha,\beta,\eta)}+\lambda \\
&=\frac{\alpha+\beta}{2\sqrt{\alpha\beta(\alpha\beta-1)}}\biggl(\begin{bmatrix}\alpha&0\\0&\beta\end{bmatrix}\biggl(2z\frac{\rd}{\rd z}+\nu\biggr)
+\begin{bmatrix}0&1\\1&0\end{bmatrix}\biggl((1+z^2)\frac{\rd}{\rd z}+\nu z\biggr) \\*
&\eqspace{}+2\eta\sqrt{\alpha\beta-1}\begin{bmatrix}0&1\\1&0\end{bmatrix}\biggr)+\lambda\bI \\
&=\frac{\alpha+\beta}{2\sqrt{\alpha\beta(\alpha\beta-1)}}B^\dagger\biggl(\bI\biggl(2z\frac{\rd}{\rd z}+\nu\biggr)
+\frac{1}{\sqrt{\alpha\beta}}\begin{bmatrix}1&0\\0&-1\end{bmatrix}\biggl((1+z^2)\frac{\rd}{\rd z}+\nu z\biggr) \\*
&\eqspace{}+2\eta\frac{\sqrt{\alpha\beta-1}}{\sqrt{\alpha\beta}}\begin{bmatrix}1&0\\0&-1\end{bmatrix}\biggr)B
+\frac{\lambda}{2\alpha\beta}B^\dagger\begin{bmatrix}\alpha+\beta&\beta-\alpha\\\beta-\alpha&\alpha+\beta\end{bmatrix}B \\
&=\frac{\alpha+\beta}{2\alpha\beta}B^\dagger\biggl(\frac{\sqrt{\alpha\beta}}{\sqrt{\alpha\beta-1}}\bI\biggl(2z\frac{\rd}{\rd z}+\nu\biggr)
+\frac{1}{\sqrt{\alpha\beta-1}}\begin{bmatrix}1&0\\0&-1\end{bmatrix}\biggl((1+z^2)\frac{\rd}{\rd z}+\nu z\biggr) \\*
&\eqspace{}+2\eta\begin{bmatrix}1&0\\0&-1\end{bmatrix}
+\lambda\begin{bmatrix}1&\frac{\beta-\alpha}{\alpha+\beta}\\\frac{\beta-\alpha}{\alpha+\beta}&1\end{bmatrix}\biggr)B \\
&=\frac{\alpha+\beta}{2\alpha\beta}B^\dagger\begin{bmatrix}\rh_+^\nu\bigl(\lambda;\frac{1}{2}\artanh\bigl(\frac{1}{\sqrt{\alpha\beta}}\bigr),2\eta\bigr)&\Delta\lambda\\
\Delta\lambda&\rh_-^\nu\bigl(\lambda;\frac{1}{2}\artanh\bigl(\frac{1}{\sqrt{\alpha\beta}}\bigr),2\eta\bigr)\end{bmatrix}B \\
&=\frac{\alpha+\beta}{2\alpha\beta}B^{-1}BB^\dagger\begin{bmatrix}\rh_+&\Delta\lambda\\\Delta\lambda&\rh_-\end{bmatrix}B
=\frac{1}{1-\Delta^2}B^{-1}\begin{bmatrix}1&-\Delta\\-\Delta&1\end{bmatrix}\begin{bmatrix}\rh_+&\Delta\lambda\\\Delta\lambda&\rh_-\end{bmatrix}B \\
&=B^{-1}\begin{bmatrix}1&\Delta\\\Delta&1\end{bmatrix}^{-1}\begin{bmatrix}\rh_+&\Delta\lambda\\\Delta\lambda&\rh_-\end{bmatrix}B,
\end{align*}
where we put $\rh_\pm^\nu\bigl(\lambda;\frac{1}{2}\artanh\bigl(\frac{1}{\sqrt{\alpha\beta}}\bigr),2\eta\bigr)=:\rh_\pm$, $\frac{\beta-\alpha}{\alpha+\beta}=:\Delta$. 
Then for $|\Delta\lambda|<C(\lambda,2\eta)^{-1}$, its inverse is given as 
\begin{align*}
&B(\widetilde{\rQ}_\nu^{(\alpha,\beta,\eta)}+\lambda)^{-1}B^{-1}
=\begin{bmatrix}\rh_+&\Delta\lambda\\\Delta\lambda&\rh_-\end{bmatrix}^{-1}\begin{bmatrix}1&\Delta\\\Delta&1\end{bmatrix} \\
&=\biggl(\bI+\Delta\lambda\begin{bmatrix}\rh_+^{-1}&0\\0&\rh_-^{-1}\end{bmatrix}\begin{bmatrix}0&1\\1&0\end{bmatrix}\biggr)^{-1}
\begin{bmatrix}\rh_+^{-1}&0\\0&\rh_-^{-1}\end{bmatrix}\biggl(\bI+\Delta\begin{bmatrix}0&1\\1&0\end{bmatrix}\biggr) \\
&=\sum_{m=0}^\infty \biggl(-\Delta\lambda\begin{bmatrix}\rh_+^{-1}&0\\0&\rh_-^{-1}\end{bmatrix}\begin{bmatrix}0&1\\1&0\end{bmatrix}\biggr)^m
\biggl(\begin{bmatrix}\rh_+^{-1}&0\\0&\rh_-^{-1}\end{bmatrix}+\Delta\begin{bmatrix}\rh_+^{-1}&0\\0&\rh_-^{-1}\end{bmatrix}\begin{bmatrix}0&1\\1&0\end{bmatrix}\biggr). 
\end{align*}
Then for $k\ge 0$, $0\le l\le n$, the coefficient of $(-\Delta\lambda)^k\Delta^l$ of $B(\widetilde{\rQ}_\nu^{(\alpha,\beta,\eta)}+\lambda)^{-n}B^{-1}$ is the sum of the products of 
$k+l$ times of $\Bigl[\begin{smallmatrix}\rh_+^{-1}&0\\0&\rh_-^{-1}\end{smallmatrix}\Bigr]\bigl[\begin{smallmatrix}0&1\\1&0\end{smallmatrix}\bigr]$s 
and $n-l$ times of $\Bigl[\begin{smallmatrix}\rh_+^{-1}&0\\0&\rh_-^{-1}\end{smallmatrix}\Bigr]$s, namely the sum of 
\[ \prod_{1\le j\le k+l}^{\longrightarrow}\biggl(\begin{bmatrix}\rh_+^{-1}&0\\0&\rh_-^{-1}\end{bmatrix}^{n_j+1} 
\begin{bmatrix}0&1\\1&0\end{bmatrix}\biggr)\begin{bmatrix}\rh_+^{-1}&0\\0&\rh_-^{-1}\end{bmatrix}^{n_{k+l+1}} \]
with $\bn\in(\BZ_{\ge 0})^{k+l+1}$, $|\bn|=n-l$. 
For each $\bn$, among $k+l$ times of $\Bigl[\begin{smallmatrix}\rh_+^{-1}&0\\0&\rh_-^{-1}\end{smallmatrix}\Bigr]\bigl[\begin{smallmatrix}0&1\\1&0\end{smallmatrix}\bigr]$s, 
$k$ times come from $\sum_m \Bigl(-\Delta\lambda\Bigl[\begin{smallmatrix}\rh_+^{-1}&0\\0&\rh_-^{-1}\end{smallmatrix}\Bigr]\bigl[\begin{smallmatrix}0&1\\1&0\end{smallmatrix}\bigr]\Bigr)^m$, 
and the other $l$ times come from $\Bigl[\begin{smallmatrix}\rh_+^{-1}&0\\0&\rh_-^{-1}\end{smallmatrix}\Bigr]
\allowbreak+\Delta\Bigl[\begin{smallmatrix}\rh_+^{-1}&0\\0&\rh_-^{-1}\end{smallmatrix}\Bigr]\bigl[\begin{smallmatrix}0&1\\1&0\end{smallmatrix}\bigr]$. 
If $n_{k+l+1}\ge 1$, then $l<n$. Similarly, if $n_{k+l+1}=0$, then $l>0$ and the rightmost 
$\Bigl[\begin{smallmatrix}\rh_+^{-1}&0\\0&\rh_-^{-1}\end{smallmatrix}\Bigr]\bigl[\begin{smallmatrix}0&1\\1&0\end{smallmatrix}\bigr]$ always comes from 
$\Bigl[\begin{smallmatrix}\rh_+^{-1}&0\\0&\rh_-^{-1}\end{smallmatrix}\Bigr]
+\Delta\Bigl[\begin{smallmatrix}\rh_+^{-1}&0\\0&\rh_-^{-1}\end{smallmatrix}\Bigr]\bigl[\begin{smallmatrix}0&1\\1&0\end{smallmatrix}\bigr]$. 
Hence $B(\widetilde{\rQ}_\nu^{(\alpha,\beta,\eta)}+\lambda)^{-n}B^{-1}$ is expanded as 
\begin{align*}
&B(\widetilde{\rQ}_\nu^{(\alpha,\beta,\eta)}+\lambda)^{-n}B^{-1} \\
&=\sum_{k=0}^\infty\sum_{l=0}^{n-1}(-\Delta\lambda)^k\Delta^l \binom{k+l}{k} \\*
&\eqspace{}\times\sum_{\substack{\bn\in(\BZ_{\ge 0})^{k+l+1}\\|\bn|=n-l-1}}
\prod_{1\le j\le k+l}^{\longrightarrow}\biggl(\begin{bmatrix}\rh_+^{-1}&0\\0&\rh_-^{-1}\end{bmatrix}^{n_j+1}\begin{bmatrix}0&1\\1&0\end{bmatrix}\biggr)
\begin{bmatrix}\rh_+^{-1}&0\\0&\rh_-^{-1}\end{bmatrix}^{n_{k+l+1}+1} \\*
&\eqspace{}+\sum_{k=0}^\infty\sum_{l=1}^{n}(-\Delta\lambda)^k\Delta^l \binom{k+l-1}{k}\sum_{\substack{\bn\in(\BZ_{\ge 0})^{k+l}\\|\bn|=n-l}}
\prod_{1\le j\le k+l}^{\longrightarrow}\biggl(\begin{bmatrix}\rh_+^{-1}&0\\0&\rh_-^{-1}\end{bmatrix}^{n_j+1}\begin{bmatrix}0&1\\1&0\end{bmatrix}\biggr). 
\end{align*}
Then its trace class norm is bounded as 
\begin{align*}
&\bigl\Vert B(\widetilde{\rQ}_\nu^{(\alpha,\beta,\eta)}+\lambda)^{-n}B^{-1}\bigr\Vert_{\Tr} \\
&\le\sum_{k=0}^\infty\sum_{l=0}^{n-1}|\Delta\lambda|^k|\Delta|^l \binom{k+l}{k} \\*
&\eqspace{}\times\sum_{\substack{\bn\in(\BZ_{\ge 0})^{k+l+1}\\|\bn|=n-l-1}}
\Biggl\Vert\prod_{1\le j\le k+l}^{\longrightarrow}\biggl(\begin{bmatrix}\rh_+^{-1}&0\\0&\rh_-^{-1}\end{bmatrix}^{n_j+1}\begin{bmatrix}0&1\\1&0\end{bmatrix}\biggr)
\begin{bmatrix}\rh_+^{-1}&0\\0&\rh_-^{-1}\end{bmatrix}^{n_{k+l+1}+1}\Biggr\Vert_{\Tr} \\*
&\eqspace{}+\sum_{k=0}^\infty\sum_{l=1}^{n}|\Delta\lambda|^k|\Delta|^l \binom{k+l-1}{k}\sum_{\substack{\bn\in(\BZ_{\ge 0})^{k+l}\\|\bn|=n-l}}
\Biggl\Vert\prod_{1\le j\le k+l}^{\longrightarrow}\biggl(\begin{bmatrix}\rh_+^{-1}&0\\0&\rh_-^{-1}\end{bmatrix}^{n_j+1}\begin{bmatrix}0&1\\1&0\end{bmatrix}\biggr)\Biggr\Vert_{\Tr} \\
&\le\sum_{k=0}^\infty\sum_{l=0}^{n-1}|\Delta\lambda|^k|\Delta|^l \binom{k+l}{k}\binom{k+n-1}{n-l-1}
\biggl\Vert\begin{bmatrix}\rh_+^{-1}&0\\0&\rh_-^{-1}\end{bmatrix}\biggr\Vert_\op^{k+n-2}\biggl\Vert\begin{bmatrix}\rh_+^{-1}&0\\0&\rh_-^{-1}\end{bmatrix}\biggr\Vert_\HS^2 \\*
&\eqspace{}+\sum_{k=0}^\infty\sum_{l=1}^{n}|\Delta\lambda|^k|\Delta|^l \binom{k+l-1}{k}\binom{k+n-1}{n-l}
\biggl\Vert\begin{bmatrix}\rh_+^{-1}&0\\0&\rh_-^{-1}\end{bmatrix}\biggr\Vert_\op^{k+n-2}\biggl\Vert\begin{bmatrix}\rh_+^{-1}&0\\0&\rh_-^{-1}\end{bmatrix}\biggr\Vert_\HS^2 \\
&\le\sum_{k=0}^\infty\sum_{l=0}^{n}|\Delta\lambda|^k|\Delta|^l \frac{(n)_k}{k!}\binom{n}{l}C_\nu(\lambda,2\eta)^{k+n-2}\cdot 2C'_\nu(\lambda,2\eta)^2 \\
&=\frac{2(1+|\Delta|)^n C_\nu(\lambda,2\eta)^{n-2}C'_\nu(\lambda,2\eta)^2}{(1-C_\nu(\lambda,2\eta)|\Delta\lambda|)^n}<\infty. 
\end{align*}
Also, by Lemma \ref{lem_diff_R} and (\ref{formula_sumtrace}), for $k+l\ge 1$ we have 
\begin{align*}
&\sum_{\substack{\bn\in(\BZ_{\ge 0})^{k+l+1}\\|\bn|=n-l-1}}
\Tr\biggl(\prod_{1\le j\le k+l}^{\longrightarrow}\biggl(\begin{bmatrix}\rh_+^{-1}&0\\0&\rh_-^{-1}\end{bmatrix}^{n_j+1}\begin{bmatrix}0&1\\1&0\end{bmatrix}\biggr)
\begin{bmatrix}\rh_+^{-1}&0\\0&\rh_-^{-1}\end{bmatrix}^{n_{k+l+1}+1}\biggr) \\
&=\begin{cases} \displaystyle \frac{2(-1)^{n-l}}{(n-l-1)!(k+l)}\frac{\partial^{n-l}R_{(k+l)/2}^\nu}{\partial\lambda^{n-l}}(\lambda;g,2\eta) 
& (k+l\colon\text{even}) \\
0 & (k+l\colon\text{odd}), \end{cases} \\
&\sum_{\substack{\bn\in(\BZ_{\ge 0})^{k+l}\\|\bn|=n-l}}
\Tr\biggl(\prod_{1\le j\le k+l}^{\longrightarrow}\biggl(\begin{bmatrix}\rh_+^{-1}&0\\0&\rh_-^{-1}\end{bmatrix}^{n_j+1}\begin{bmatrix}0&1\\1&0\end{bmatrix}\biggr)\biggr) \\
&=\begin{cases} \displaystyle \frac{2(-1)^{n-l}}{(n-l)!}\frac{\partial^{n-l}R_{(k+l)/2}^\nu}{\partial\lambda^{n-l}}(\lambda;g,2\eta) & (k+l\colon\text{even}) \\
0 & (k+l\colon\text{odd}), \end{cases}
\end{align*}
where $g:=\frac{1}{2}\artanh\bigl(\frac{1}{\sqrt{\alpha\beta}}\bigr)$, and hence $\zeta\bigl(\widetilde{\rQ}_\nu^{(\alpha,\beta,\eta)};n,\lambda)$ is computed as 
\begin{align*}
&\zeta\bigl(\widetilde{\rQ}_\nu^{(\alpha,\beta,\eta)};n,\lambda)
=\Tr\bigl(B(\widetilde{\rQ}_\nu^{(\alpha,\beta,\eta)}+\lambda)^{-n}B^{-1}\bigr) \\
&=\sum_{k=0}^\infty\sum_{l=0}^{n-1}(-\Delta\lambda)^k\Delta^l \binom{k+l}{k} \\*
&\eqspace{}\times\sum_{\substack{\bn\in(\BZ_{\ge 0})^{k+l+1}\\|\bn|=n-l-1}}
\Tr\biggl(\prod_{1\le j\le k+l}^{\longrightarrow}\biggl(\begin{bmatrix}\rh_+^{-1}&0\\0&\rh_-^{-1}\end{bmatrix}^{n_j+1}\begin{bmatrix}0&1\\1&0\end{bmatrix}\biggr)
\begin{bmatrix}\rh_+^{-1}&0\\0&\rh_-^{-1}\end{bmatrix}^{n_{k+l+1}+1}\biggr) \\*
&\eqspace{}+\sum_{k=0}^\infty\sum_{l=1}^{n}(-\Delta\lambda)^k\Delta^l \binom{k+l-1}{k}\sum_{\substack{\bn\in(\BZ_{\ge 0})^{k+l}\\|\bn|=n-l}}
\Tr\biggl(\prod_{1\le j\le k+l}^{\longrightarrow}\biggl(\begin{bmatrix}\rh_+^{-1}&0\\0&\rh_-^{-1}\end{bmatrix}^{n_j+1}\begin{bmatrix}0&1\\1&0\end{bmatrix}\biggr)\biggr) \\
&=\Tr\biggl(\begin{bmatrix}\rh_+^{-1}&0\\0&\rh_-^{-1}\end{bmatrix}^n\biggr) \\*
&\eqspace{}+\sum_{k=0}^\infty\sum_{\substack{0\le l\le n-1 \\ k+l\ge 1\colon\text{even}}}(-\Delta\lambda)^k\Delta^l \binom{k+l}{k}
\frac{2(-1)^{n-l}}{(n-l-1)!(k+l)}\frac{\partial^{n-l}R_{(k+l)/2}^\nu}{\partial\lambda^{n-l}}(\lambda;g,2\eta) \\*
&\eqspace{}+\sum_{k=0}^\infty\sum_{\substack{1\le l\le n \\ k+l\colon\text{even}}}(-\Delta\lambda)^k\Delta^l \binom{k+l-1}{k}
\frac{2(-1)^{n-l}}{(n-l)!}\frac{\partial^{n-l}R_{(k+l)/2}^\nu}{\partial\lambda^{n-l}}(\lambda;g,2\eta) \\
&=\Tr\biggl(\begin{bmatrix}\rh_+^{-1}&0\\0&\rh_-^{-1}\end{bmatrix}^n\biggr) \\*
&\eqspace{}+\sum_{k=0}^\infty\sum_{\substack{0\le l\le n \\ k+l\ge 1\colon\text{even}}}\lambda^k\Delta^{k+l} \binom{n}{l}
\frac{2(-1)^{n+k+l}(k+1)_l}{(n-1)!(k+l)}\frac{\partial^{n-l}R_{(k+l)/2}^\nu}{\partial\lambda^{n-l}}(\lambda;g,2\eta) \\
&=\Tr\biggl(\begin{bmatrix}\rh_+^{-1}&0\\0&\rh_-^{-1}\end{bmatrix}^n\biggr) \\*
&\eqspace{}+\sum_{m=1}^\infty\sum_{l=0}^{\min\{2m,n\}}\lambda^{2m-l}\Delta^{2m} \binom{n}{l}
\frac{2(-1)^{n+2m}(2m-l+1)_l}{(n-1)!\cdot 2m}\frac{\partial^{n-l}R_m^\nu}{\partial\lambda^{n-l}}(\lambda;g,2\eta) \\
&=\Tr\biggl(\begin{bmatrix}\rh_+^{-1}&0\\0&\rh_-^{-1}\end{bmatrix}^n\biggr)
+\frac{(-1)^n}{(n-1)!}\sum_{m=1}^\infty \frac{\Delta^{2m}}{m}\sum_{l=0}^{\min\{2m,n\}} \binom{n}{l}\biggl(\frac{\partial^l}{\partial\lambda^l}\lambda^{2m}\biggr)
\frac{\partial^{n-l}R_m^\nu}{\partial\lambda^{n-l}}(\lambda;g,2\eta) \\
&=\Tr\biggl(\begin{bmatrix}\rh_+^{-1}&0\\0&\rh_-^{-1}\end{bmatrix}^n\biggr)
+\frac{(-1)^n}{(n-1)!}\sum_{m=1}^\infty \frac{\Delta^{2m}}{m}\frac{\partial^n}{\partial\lambda^n}\lambda^{2m}R_m^\nu(\lambda;g,2\eta). 
\end{align*}
Then since 
\begin{align*}
\Tr\biggl(\begin{bmatrix}\rh_+^{-1}&0\\0&\rh_-^{-1}\end{bmatrix}^n\biggr)
&=\Tr\bigl(\rh_+^\nu(\lambda;g,2\eta)^{-n}\bigr)+\Tr\bigl(\rh_-^\nu(\lambda;g,2\eta)^{-n}\bigr) \\
&=2^{-n}\zeta\biggl(n,\frac{\lambda+2\eta+\nu}{2}\biggr)+2^{-n}\zeta\biggl(n,\frac{\lambda-2\eta+\nu}{2}\biggr) 
\end{align*}
by Example \ref{ex_Yg}, we get the desired formula. 
\end{proof}

For later use, we prove the following. 
\begin{lemma}\label{lem_limit}
Let $g,\Delta,\varepsilon\in\BR$, $\lambda\in\BC$, $n\in\BZ_{\ge 2}$. We assume $\lambda\pm\varepsilon\notin -\BZ_{\ge 0}$ and $|\Delta|<C_\flat(\lambda,\varepsilon)^{-1}$. 
If $\lim_{\nu\to\infty}R_m^\nu\bigl(2\lambda-\nu;\frac{g}{\sqrt{\nu}},2\varepsilon\bigr)$ converges to a holomorphic function in $\lambda$ around $(\lambda,g,\varepsilon)$, then we have 
\begin{align*}
&2^n\lim_{\nu\to\infty}\zeta\bigl(\widetilde{\rH}_\nu^{(g/\sqrt{\nu},2\Delta,2\varepsilon)};n,2\lambda-\nu\bigr) \\
&=\zeta(n,\lambda+\varepsilon)+\zeta(n,\lambda-\varepsilon)
+\frac{(-1)^n}{(n-1)!}\sum_{m=1}^\infty \frac{(2\Delta)^{2m}}{m}\frac{\partial^n}{\partial\lambda^n}\lim_{\nu\to\infty}R_m^\nu\biggl(2\lambda-\nu;\frac{g}{\sqrt{\nu}},2\varepsilon\biggr). 
\end{align*}
\end{lemma}

\begin{proof}
By Theorem \ref{thm_zeta_sum}\,(2), we have 
\begin{align*}
&2^n\zeta\bigl(\widetilde{\rH}_\nu^{(g/\sqrt{\nu},2\Delta,2\varepsilon)};n,2\lambda-\nu\bigr) \\
&=\zeta(n,\lambda+\varepsilon)+\zeta(n,\lambda-\varepsilon)
+\frac{(-2)^n}{(n-1)!}\sum_{m=1}^\infty \frac{(2\Delta)^{2m}}{m}\frac{\partial^n R_m^\nu}{\partial\lambda^n}\biggl(2\lambda-\nu;\frac{g}{\sqrt{\nu}},2\varepsilon\biggr) \\
&=\zeta(n,\lambda+\varepsilon)+\zeta(n,\lambda-\varepsilon)
+\frac{(-1)^n}{(n-1)!}\sum_{m=1}^\infty \frac{(2\Delta)^{2m}}{m}\frac{\partial^n}{\partial\lambda^n}R_m^\nu\biggl(2\lambda-\nu;\frac{g}{\sqrt{\nu}},2\varepsilon\biggr). 
\end{align*}
We take a sufficiently small $r>0$ such that the ball $B:=\{\mu\in\BC\mid |\mu-\lambda|\le r\}$ does not intersect $\pm\varepsilon-\BZ_{\ge 0}$, 
and $C_\flat(B,\varepsilon):=\max_{\mu\in B}C_\flat(\mu,\varepsilon)$, $C_\flat'(B,\varepsilon):=\max_{\mu\in B}C_\flat'(\mu,\varepsilon)$ satisfy 
$(C_\flat(\lambda,\varepsilon)\le)C_\flat(B,\varepsilon)<|\Delta|^{-1}$. Then we have 
\begin{align*}
&\sum_{m=1}^\infty \frac{|(2\Delta)^{2m}|}{m}\frac{n!}{2\pi}\oint_{\partial B} \biggl|R_m^\nu\biggl(2\mu-\nu;\frac{g}{\sqrt{\nu}},2\varepsilon\biggr)\biggr|\frac{|\rd\mu|}{|(\mu-\lambda)^{n+1}|} \\
&\le \sum_{m=1}^\infty \frac{|2\Delta|^{2m}}{m}\frac{n!}{2\pi}\oint_{\partial B} C_\nu(2\mu-\nu,2\varepsilon)^{2m-2}C_\nu'(2\mu-\nu,2\varepsilon)^2\frac{|\rd\mu|}{r^{n+1}} \\
&=\sum_{m=1}^\infty \frac{|\Delta|^{2m}}{m}\frac{n!}{2\pi}\oint_{\partial B} C_\flat(\mu,\varepsilon)^{2m-2}C_\flat'(\mu,\varepsilon)^2\frac{|\rd\mu|}{r^{n+1}} \\
&\le\frac{n!}{r^n}\sum_{m=1}^\infty \frac{|\Delta|^{2m}}{m}C_\flat(B,\varepsilon)^{2m-2}C_\flat'(B,\varepsilon)^2<\infty 
\end{align*}
uniformly in $\nu>0$. Hence by the dominated convergence theorem, 
\begin{align*}
&\lim_{\nu\to\infty}\sum_{m=1}^\infty \frac{(2\Delta)^{2m}}{m}\frac{\partial^n}{\partial\lambda^n}R_m^\nu\biggl(2\lambda-\nu;\frac{g}{\sqrt{\nu}},2\varepsilon\biggr) \\
&=\lim_{\nu\to\infty}\sum_{m=1}^\infty \frac{(2\Delta)^{2m}}{m}\frac{n!}{2\pi\sqrt{-1}}
\oint_{\partial B} R_m^\nu\biggl(2\mu-\nu;\frac{g}{\sqrt{\nu}},2\varepsilon\biggr)\frac{\rd\mu}{(\mu-\lambda)^{n+1}} \\
&=\sum_{m=1}^\infty \frac{(2\Delta)^{2m}}{m}\frac{n!}{2\pi\sqrt{-1}}
\oint_{\partial B} \lim_{\nu\to\infty} R_m^\nu\biggl(2\mu-\nu;\frac{g}{\sqrt{\nu}},2\varepsilon\biggr)\frac{\rd\mu}{(\mu-\lambda)^{n+1}} \\
&=\sum_{m=1}^\infty \frac{(2\Delta)^{2m}}{m}\frac{\partial^n}{\partial\lambda^n}\lim_{\nu\to\infty} R_m^\nu\biggl(2\lambda-\nu;\frac{g}{\sqrt{\nu}},2\varepsilon\biggr) 
\end{align*}
holds if $\lim_{\nu\to\infty} R_m^\nu\bigl(2\lambda-\nu;\frac{g}{\sqrt{\nu}},2\varepsilon\bigr)$ is holomorphic in $\lambda\in B$. 
\end{proof}

By (\ref{formula_dtauXg}), (\ref{formula_dtauYg}), for $g\in\BR$, $\lambda,\varepsilon\in\BC$ with $\lambda\pm\varepsilon\notin -\BZ_{\ge 0}$, $m\in\BZ_{\ge 1}$, we have 
\[ R_m^\flat(\lambda;g,\varepsilon)=\Tr\bigl(\bigl((\ra^\dagger\ra+g(\ra+\ra^\dagger)+g^2+\varepsilon+\lambda)^{-1}(\ra^\dagger\ra-g(\ra+\ra^\dagger)+g^2-\varepsilon+\lambda)^{-1}\bigr)^m\bigr), \]
and for $g\in\BR$, $\lambda,\varepsilon\in\BC$ with $\lambda\pm\varepsilon\notin -\bigl(\BZ_{\ge 0}+\frac{1}{2}\bigr)$, $m\in\BZ_{\ge 1}$, we have 
\begin{align}
&R_m^+(\lambda;g,\varepsilon):=R_m^{1/2}(\lambda;g,\varepsilon)+R_m^{3/2}(\lambda;g,\varepsilon) \notag\\*
&=\Tr\biggl(\biggl(\biggl(\cosh(2g)\biggl(\ra^\dagger\ra+\frac{1}{2}\biggr)+\frac{\sinh(2g)}{2}(\ra^2+(\ra^\dagger)^2)+\varepsilon+\lambda\biggr)^{-1} \notag\\*
&\hspace{40pt}\times\biggl(\cosh(2g)\biggl(\ra^\dagger\ra+\frac{1}{2}\biggr)-\frac{\sinh(2g)}{2}(\ra^2+(\ra^\dagger)^2)-\varepsilon+\lambda\biggr)^{-1}\biggr)^m\biggr), \notag\\
&R_m^-(\lambda;g,\varepsilon):=R_m^{1/2}(\lambda;g,\varepsilon)-R_m^{3/2}(\lambda;g,\varepsilon) \notag\\*
&=\Tr\biggl(\biggl(\biggl(\cosh(2g)\biggl(\ra^\dagger\ra+\frac{1}{2}\biggr)+\frac{\sinh(2g)}{2}(\ra^2+(\ra^\dagger)^2)+\varepsilon+\lambda\biggr)^{-1} \notag\\*
&\hspace{40pt}\times\biggl(\cosh(2g)\biggl(\ra^\dagger\ra+\frac{1}{2}\biggr)-\frac{\sinh(2g)}{2}(\ra^2+(\ra^\dagger)^2)-\varepsilon+\lambda\biggr)^{-1}\biggr)^m\biggr|_{L^2(\BR)_\even}\biggr) \notag\\*
&\eqspace{}-\Tr\biggl(\biggl(\biggl(\cosh(2g)\biggl(\ra^\dagger\ra+\frac{1}{2}\biggr)+\frac{\sinh(2g)}{2}(\ra^2+(\ra^\dagger)^2)+\varepsilon+\lambda\biggr)^{-1} \notag\\*
&\hspace{40pt}\times\biggl(\cosh(2g)\biggl(\ra^\dagger\ra+\frac{1}{2}\biggr)-\frac{\sinh(2g)}{2}(\ra^2+(\ra^\dagger)^2)-\varepsilon+\lambda\biggr)^{-1}\biggr)^m\biggr|_{L^2(\BR)_\odd}\biggr). 
\label{def_Rmpm}
\end{align}
Also, let 
\[ C(\lambda,\varepsilon):=\max\{C_{1/2}(\lambda,\varepsilon),C_{3/2}(\lambda,\varepsilon)\}=\max_{\delta\in\{\pm \},\ k\in\BZ_{\ge 0}}\frac{1}{\bigl|k+\lambda+\delta\varepsilon+\frac{1}{2}\bigr|}. \]
Then by Theorem \ref{thm_zeta_sum}, we immediately get the following. 

\begin{corollary}\label{cor_zeta_sum}
\begin{enumerate}
\item Let $g,\Delta,\varepsilon\in\BR$, $\lambda\in\BC$. We assume $\lambda\pm\varepsilon\notin -\BZ_{\ge 0}$ and $|\Delta|<C_\flat(\lambda,\varepsilon)^{-1}$. 
Then for $n\in\BZ_{\ge 2}$, $(\widetilde{\rH}_{1\QRM}^{(g,\Delta,\varepsilon)}+\lambda)^{-n}$ is of trace class, and we have 
\begin{align*}
\zeta\bigl(\widetilde{\rH}_{1\QRM}^{(g,\Delta,\varepsilon)};n,\lambda\bigr)
=\zeta(n,\lambda+\varepsilon)+\zeta(n,\lambda-\varepsilon)
+\frac{(-1)^n}{(n-1)!}\sum_{m=1}^\infty \frac{\Delta^{2m}}{m}\frac{\partial^n R_m^\flat}{\partial\lambda^n}(\lambda;g,\varepsilon). 
\end{align*}
\item Let $g,\Delta,\varepsilon\in\BR$, $\lambda\in\BC$. 
We assume $\lambda\pm\varepsilon\notin -\bigl(\BZ_{\ge 0}+\frac{1}{2}\bigr)$ and $|\Delta|<C(\lambda,\varepsilon)^{-1}$. 
Then for $n\in\BZ_{\ge 2}$, $(\widetilde{\rH}_{2\QRM}^{(g,\Delta,\varepsilon)}+\lambda)^{-n}$ is of trace class, and we have 
\begin{align*}
&\zeta\bigl(\widetilde{\rH}_{2\QRM}^{(g,\Delta,\varepsilon)};n,\lambda\bigr) \\
&=\zeta\biggl(n,\lambda+\varepsilon+\frac{1}{2}\biggr)+\zeta\biggl(n,\lambda-\varepsilon+\frac{1}{2}\biggr)
+\frac{(-1)^n}{(n-1)!}\sum_{m=1}^\infty \frac{\Delta^{2m}}{m}\frac{\partial^n R_m^+}{\partial\lambda^n}(\lambda;g,\varepsilon), \\
&\zeta\bigl(\widetilde{\rH}_{2\QRM}^{(g,\Delta,\varepsilon)}\bigr|_{L^2(\BR)_\even};n,\lambda\bigr)-\zeta\bigl(\widetilde{\rH}_{2\QRM}^{(g,\Delta,\varepsilon)}\bigr|_{L^2(\BR)_\odd};n,\lambda\bigr) \\
&=\sum_{k=0}^\infty\frac{(-1)^k}{\bigl(\lambda+\varepsilon+\frac{1}{2}+k\bigr)^n}+\sum_{k=0}^\infty\frac{(-1)^k}{\bigl(\lambda-\varepsilon+\frac{1}{2}+k\bigr)^n}
+\frac{(-1)^n}{(n-1)!}\sum_{m=1}^\infty \frac{\Delta^{2m}}{m}\frac{\partial^n R_m^-}{\partial\lambda^n}(\lambda;g,\varepsilon). 
\end{align*}
\item Let $\alpha,\beta,\eta\in\BR$, $\lambda\in\BC$. We assume $\alpha\beta>1$, $\lambda\pm 2\eta\notin -\bigl(\BZ_{\ge 0}+\frac{1}{2}\bigr)$ 
and $\bigl|\lambda\frac{\alpha-\beta}{\alpha+\beta}\bigr|C(\lambda,2\eta)\allowbreak<1$. 
Then for $n\in\BZ_{\ge 2}$, $(\widetilde{\rH}_\NCHO^{(\alpha,\beta,\eta)}+\lambda)^{-n}$ is of trace class, and we have 
\begin{align*}
&\zeta\bigl(\widetilde{\rH}_\NCHO^{(\alpha,\beta,\eta)};n,\lambda\bigr) \\*
&=\zeta\biggl(n,\lambda+2\eta+\frac{1}{2}\biggr)+\zeta\biggl(n,\lambda-2\eta+\frac{1}{2}\biggr) \\*
&\eqspace{}+\frac{(-1)^n}{(n-1)!}\sum_{m=1}^\infty \frac{1}{m}\biggl(\frac{\alpha-\beta}{\alpha+\beta}\biggr)^{2m}
\frac{\partial^n}{\partial \lambda^n}\lambda^{2m}R_m^+\biggl(\lambda;\frac{1}{2}\artanh\biggl(\frac{1}{\sqrt{\alpha\beta}}\biggr),2\eta\biggr), \\
&\zeta\bigl(\widetilde{\rH}_\NCHO^{(\alpha,\beta,\eta)}\bigr|_{L^2(\BR)_\even};n,\lambda\bigr)-\zeta\bigl(\widetilde{\rH}_\NCHO^{(\alpha,\beta,\eta)}\bigr|_{L^2(\BR)_\odd};n,\lambda\bigr) \\*
&=\sum_{k=0}^\infty\frac{(-1)^k}{\bigl(\lambda+2\eta+\frac{1}{2}+k\bigr)^n}+\sum_{k=0}^\infty\frac{(-1)^k}{\bigl(\lambda-2\eta+\frac{1}{2}+k\bigr)^n} \\*
&\eqspace{}+\frac{(-1)^n}{(n-1)!}\sum_{m=1}^\infty \frac{1}{m}\biggl(\frac{\alpha-\beta}{\alpha+\beta}\biggr)^{2m}
\frac{\partial^n}{\partial \lambda^n}\lambda^{2m}R_m^-\biggl(\lambda;\frac{1}{2}\artanh\biggl(\frac{1}{\sqrt{\alpha\beta}}\biggr),2\eta\biggr). 
\end{align*}
\end{enumerate}
\end{corollary}

\begin{remark}\label{rem_Taylor_R}
Especially, if $\lambda=\eta=0$, then we have 
\begin{align*}
&\zeta\bigl(\widetilde{\rH}_\NCHO^{(\alpha,\beta,0)};n,0\bigr) \\*
&=2\zeta\biggl(n,\frac{1}{2}\biggr)+\frac{(-1)^n}{(n-1)!}\sum_{m=1}^\infty \frac{1}{m}\biggl(\frac{\alpha-\beta}{\alpha+\beta}\biggr)^{2m}
\frac{\partial^n}{\partial \lambda^n}\lambda^{2m}R_m^+\biggl(\lambda;\frac{1}{2}\artanh\biggl(\frac{1}{\sqrt{\alpha\beta}}\biggr),0\biggr)\biggr|_{\lambda=0} \\
&=2\zeta\biggl(n,\frac{1}{2}\biggr)+\frac{(-1)^n}{(n-1)!}\sum_{m=1}^{\lfloor n/2\rfloor} \binom{n}{2m}\frac{(2m)!}{m}\biggl(\frac{\alpha-\beta}{\alpha+\beta}\biggr)^{2m}
\frac{\partial^{n-2m}R_m^+}{\partial \lambda^{n-2m}}\biggl(0;\frac{1}{2}\artanh\biggl(\frac{1}{\sqrt{\alpha\beta}}\biggr),0\biggr) \\
&=2\zeta\biggl(n,\frac{1}{2}\biggr)+\sum_{m=1}^{\lfloor n/2\rfloor}\biggl(\frac{\alpha-\beta}{\alpha+\beta}\biggr)^{2m}\frac{(-1)^n n}{m(n-2m)!}
\frac{\partial^{n-2m}R_m^+}{\partial\lambda^{n-2m}}\biggl(0;\frac{1}{2}\artanh\biggl(\frac{1}{\sqrt{\alpha\beta}}\biggr),0\biggr). 
\end{align*}
Comparing this with Kimoto--Wakayama's formula \cite[Theorem 2.6]{KW4}, 
\begin{align*}
\zeta\bigl(\widetilde{\rH}_\NCHO^{(\alpha,\beta,0)};n,0\bigr)
=2\zeta\biggl(n,\frac{1}{2}\biggr)+2\sum_{m=1}^{\lfloor n/2\rfloor}\biggl(\frac{\alpha-\beta}{\alpha+\beta}\biggr)^{2m}R_{n,m}\biggl(\frac{1}{\sqrt{\alpha\beta-1}}\biggr), 
\end{align*}
our $R_m^+(\lambda;g,\varepsilon)$ and $R_{n,m}(\kappa)$ in \cite{KW4} are related as 
\[ R_{n,m}(\sinh(2g))=\frac{(-1)^n n}{2m(n-2m)!}\frac{\partial^{n-2m}R_m^+}{\partial\lambda^{n-2m}}(0;g,0), \]
namely, we have 
\begin{equation}\label{formula_Taylor_R}
R_m^+(\lambda;g,0)=\sum_{n=0}^\infty \frac{(-1)^n 2m}{n+2m}R_{n+2m,m}(\sinh(2g))\lambda^n. 
\end{equation}
\end{remark}

\section{Integral expression of each term}\label{section_int_exp}

In this section, we give integral expressions of 
\begin{align*}
R_m^\flat(\lambda;g,\varepsilon):=\Tr\bigl(\bigl((\rd\tau_\flat(X(g))+\varepsilon+\lambda)^{-1}(\rd\tau_\flat(X(-g))-\varepsilon+\lambda)^{-1}\bigr)^m\bigr), \\
R_m^\nu(\lambda;g,\varepsilon):=\Tr\bigl(\bigl((\rd\tau_\nu(Y(g))+\varepsilon+\lambda)^{-1}(\rd\tau_\nu(Y(-g))-\varepsilon+\lambda)^{-1}\bigr)^m\bigr), 
\end{align*}
where, for $g\in\BR$, 
\begin{align*}
X(g):\hspace{-3pt}&=\begin{bmatrix}0&-g&-g^2\\0&1&g\\0&0&0\end{bmatrix}
=\begin{bmatrix}1&-g&g^2/2\\0&1&-g\\0&0&1\end{bmatrix}\begin{bmatrix}0&0&0\\0&1&0\\0&0&0\end{bmatrix}\begin{bmatrix}1&g&g^2/2\\0&1&g\\0&0&1\end{bmatrix}\in\fg^\BC_\flat, \\ 
Y(g):\hspace{-3pt}&=\begin{bmatrix} -\cosh(2g)&-\sinh(2g) \\ \sinh(2g)&\cosh(2g) \end{bmatrix} \\
&=\begin{bmatrix} \cosh(g)&-\sinh(g) \\ -\sinh(g)&\cosh(g) \end{bmatrix}
\begin{bmatrix}-1&0\\0&1\end{bmatrix}\begin{bmatrix} \cosh(g)&\sinh(g) \\ \sinh(g)&\cosh(g) \end{bmatrix}\in\fg^\BC
\end{align*}
are as in Examples \ref{ex_Xg}, \ref{ex_Yg}. 

\begin{theorem}\label{thm_int_exp}
\begin{enumerate}
\item Let $g\in\BR$, $\lambda,\varepsilon\in\BC$, $m\in\BZ_{\ge 1}$. If $\Re\lambda-|\Re\varepsilon|>0$, then we have 
\begin{align*}
R_m^\flat(\lambda;g,\varepsilon)
&=\int_{[0,1]^{2m}} \exp\biggl(-4g^2\frac{\Phi_m(u_1,\ldots,u_{2m})}{1-\prod_{j=1}^{2m}u_j}\biggr)\frac{\prod_{j=1}^{2m} u_j^{\lambda-(-1)^j\varepsilon-1}}{1-\prod_{j=1}^{2m}u_j}
\,\prod_{j=1}^{2m}\rd u_j,
\end{align*}
where 
\begin{align}
\Phi_m(u_1,\ldots,u_{2m}):=m\biggl(1+\prod_{j=1}^{2m}u_j\biggr)+\sum_{l=1}^{2m-1}(-1)^l\sum_{k=1}^{2m}\prod_{j=k}^{k+l-1}u_j, \label{formula_Phigu}
\end{align}
with $u_{2m+j}:=u_j$ for $j=1,\ldots,2m$. 
\item Let $\nu>0$, $g\in\BR$, $\lambda,\varepsilon\in\BC$, $m\in\BZ_{\ge 1}$. If $\Re\lambda-|\Re\varepsilon|+\nu>0$, then we have 
\begin{multline*}
R^\nu_m(\lambda;g,\varepsilon)
=\int_{[0,1]^{2m}} \frac{\bigl(\frac{1}{2}(\sqrt{\Psi^g_m(u_1,\ldots,u_{2m})+2}-\sqrt{\Psi^g_m(u_1,\ldots,u_{2m})-2})\bigr)^{2\nu-2}}
{\sqrt{\Psi^g_m(u_1,\ldots,u_{2m})+2}\sqrt{\Psi^g_m(u_1,\ldots,u_{2m})-2}} \\*
{}\times\prod_{j=1}^{2m} u_j^{\lambda-(-1)^j\varepsilon-1}\,\rd u_j,
\end{multline*}
where 
\begin{align}
\Psi^g_m(u_1,\ldots,u_{2m}):=\tr\biggl(\prod_{1\le j\le 2m}^{\longrightarrow}\biggl(\begin{bmatrix} \cosh(2g)&(-1)^j\sinh(2g) \\ (-1)^j\sinh(2g)&\cosh(2g) \end{bmatrix}
\begin{bmatrix}u_j^{-1}&0\\0&u_j\end{bmatrix}\biggr)\biggr). \label{formula_Psigu}
\end{align}
\item Let $g\in\BR$, $\lambda,\varepsilon\in\BC$, $m\in\BZ_{\ge 1}$. If $\Re\lambda-|\Re\varepsilon|>0$, then we have 
\[ \lim_{\nu\to\infty}2^{2m}R_m^\nu\biggl(2\lambda-\nu;\frac{g}{\sqrt{\nu}},2\varepsilon\biggr)=R_m^\flat(\lambda;g,\varepsilon). \]
\end{enumerate}
\end{theorem}

For example, we have 
\begin{gather*}
\Phi_1(u_1,u_2)=1+u_1u_2-u_1-u_2=(1-u_1)(1-u_2), \\
\begin{split}
\Phi_2(u_1,u_2,u_3,u_4)&=2(1+u_1u_2u_3u_4)-u_1-u_2-u_3-u_4+u_1u_2+u_2u_3+u_3u_4+u_4u_1 \\*
&\eqspace{}-u_1u_2u_3-u_2u_3u_4-u_3u_4u_1-u_4u_1u_2, 
\end{split}\\
\begin{split}
\Psi^g_1(u_1,u_2)&=\tr\biggl(\begin{bmatrix} \cosh(2g)&-\sinh(2g) \\ -\sinh(2g)&\cosh(2g) \end{bmatrix}\begin{bmatrix}u_1^{-1}&0\\0&u_1\end{bmatrix}
\begin{bmatrix} \cosh(2g)&\sinh(2g) \\ \sinh(2g)&\cosh(2g) \end{bmatrix}\begin{bmatrix}u_2^{-1}&0\\0&u_2\end{bmatrix}\biggr) \\
&=(u_1^{-1}u_2^{-1}+u_1u_2)\cosh^2(2g)-(u_1u_2^{-1}+u_1^{-1}u_2)\sinh^2(2g).
\end{split}
\end{gather*}

To prove the theorem, we prepare the following lemma. For $u\in(0,1)$, let 
\[ u^{X(g)}:=\exp(X(g)\log u)\in\Gamma_\flat^\circ, \qquad u^{Y(g)}:=\exp(Y(g)\log u)\in\Gamma^\circ \]
be as in Examples \ref{ex_Xg}, \ref{ex_Yg}, so that we have 
\begin{equation}\label{formula_Trace}
\Tr(\tau_\flat(u^{X(g)}))=\Tr(\tau_\flat(u^{X(0)}))=\frac{1}{1-u}, \quad 
\Tr(\tau_\nu(u^{Y(g)}))=\Tr(\tau_\nu(u^{Y(0)}))=\frac{u^\nu}{1-u^2}.
\end{equation}

\begin{lemma}\label{lem_Trtau}
\begin{enumerate}
\item Let $g\in\BR$, $m\in\BZ_{\ge 1}$, $u_1,\ldots,u_{2m}\in (0,1)$. Then we have 
\begin{gather*}
\biggl\Vert\tau_\flat\biggl(\prod_{1\le j\le 2m}^{\longrightarrow}u_j^{X(-(-1)^jg)}\biggr)\biggr\Vert_{\Tr}
\le \frac{1}{1-\prod_{j=1}^{2m}u_j}, \\
\Tr\biggl(\tau_\flat\biggl(\prod_{1\le j\le 2m}^{\longrightarrow}u_j^{X(-(-1)^jg)}\biggr)\biggr)
=\exp\biggl(-4g^2\frac{\Phi_m(u_1,\ldots,u_{2m})}{1-\prod_{j=1}^{2m}u_j}\biggr)\frac{1}{1-\prod_{j=1}^{2m}u_j},
\end{gather*}
where $\Phi_m(u_1,\ldots,u_{2m})$ is as in (\ref{formula_Phigu}). 
\item Let $\nu>0$, $g\in\BR$, $m\in\BZ_{\ge 1}$, $u_1,\ldots,u_{2m}\in (0,1)$. Then we have 
\begin{gather*}
\biggl\Vert\tau_\nu\biggl(\prod_{1\le j\le 2m}^{\longrightarrow}u_j^{Y(-(-1)^jg)}\biggr)\biggr\Vert_{\Tr}
\le \frac{\prod_{j=1}^{2m}u_j^\nu}{1-\prod_{j=1}^{2m}u_j^2}, \\
\begin{split}
&\Tr\biggl(\tau_\nu\biggl(\prod_{1\le j\le 2m}^{\longrightarrow}u_j^{Y(-(-1)^jg)}\biggr)\biggr) \\*
&=\frac{\bigl(\frac{1}{2}(\sqrt{\Psi^g_m(u_1,\ldots,u_{2m})+2}-\sqrt{\Psi^g_m(u_1,\ldots,u_{2m})-2})\bigr)^{2\nu-2}}
{\sqrt{\Psi^g_m(u_1,\ldots,u_{2m})+2}\sqrt{\Psi^g_m(u_1,\ldots,u_{2m})-2}}, 
\end{split}
\end{gather*}
where $\Psi^g_m(u_1,\ldots,u_{2m})$ is as in (\ref{formula_Psigu}). 
\item Let $g\in\BR$, $m\in\BZ_{\ge 1}$, $u_1,\ldots,u_{2m}\in (0,1)$. Then we have 
\[ \lim_{\nu\to\infty}\Tr\biggl(\tau_\nu\biggl(\prod_{1\le j\le 2m}^{\longrightarrow}u_j^{Y(-(-1)^jg/\sqrt{\nu})}\biggr)\biggr)\prod_{j=1}^{2m} u_j^{-\nu}
=\Tr\biggl(\tau_\flat\biggl(\prod_{1\le j\le 2m}^{\longrightarrow}u_j^{2X(-(-1)^jg)}\biggr)\biggr). \]
\end{enumerate}
\end{lemma}

\begin{proof}
(1) First, let $u_j=e^{-x_j}\in (0,1)$, and put $p_j:=(\sum_{k=1}^{2m}x_k)/x_j$, so that $\sum_{j=1}^{2m}p_j^{-1}=1$. 
Then by H\"older's inequality (\ref{Holder_ineq}) and (\ref{formula_Trace}), we have 
\begin{align*}
&\biggl\Vert\tau_\flat\biggl(\prod_{1\le j\le 2m}^{\longrightarrow}u_j^{X(-(-1)^jg)}\biggr)\biggr\Vert_{\Tr}
\le \prod_{j=1}^{2m}\Vert\tau_\flat(e^{-x_jX(-(-1)^jg)})\Vert_{p_j} \\
&=\prod_{j=1}^{2m}\bigl(\Tr\bigl(\tau_\flat(e^{-\sum_{k=1}^{2m}x_kX(-(-1)^jg)}) \bigr)\bigr)^{1/p_j}
=\Tr\bigl(\tau_\flat(e^{-\sum_{k=1}^{2m}x_kX(0)}) \bigr) \\
&=\frac{1}{1-e^{-\sum_{k=1}^{2m}x_k}}
=\frac{1}{1-\prod_{j=1}^{2m}u_j}<\infty. 
\end{align*}
Next, we put 
\begin{align*}
\gamma=\begin{bmatrix}1&a&b\\0&c&d\\0&0&1\end{bmatrix}
:\hspace{-3pt}&=\begin{bmatrix}1&-g&g^2/2\\0&1&-g\\0&0&1\end{bmatrix}\biggl(\prod_{1\le j\le 2m}^{\longrightarrow}u_j^{X(-(-1)^jg)}\biggr) \begin{bmatrix}1&g&g^2/2\\0&1&g\\0&0&1\end{bmatrix} \\
&=\prod_{1\le j\le 2m}^{\longrightarrow} \left(\begin{bmatrix} 1&2(-1)^jg&2g^2 \\ 0&1&2(-1)^jg \\ 0&0&1 \end{bmatrix}\begin{bmatrix}1&0&0\\0&u_j&0\\0&0&1\end{bmatrix}\right), 
\end{align*}
so that 
\[ \Tr\biggl(\tau_\flat\biggl(\prod_{1\le j\le 2m}^{\longrightarrow}u_j^{X(-(-1)^jg)}\biggr)\biggr)=\Tr(\tau_\flat(\gamma))
=\frac{\exp\bigl(-\frac{ad+b(1-c)}{1-c}\bigr)}{1-c} \]
by Theorem \ref{thm_Tr_tauflat}. To compute $\gamma$, since we have 
\[ \gamma\in\begin{bmatrix} 1&\BC g&\BC g^2 \\ 0&\BC&\BC g \\ 0&0&1 \end{bmatrix} \]
by the induction on $m$, it suffices to compute the coefficients of $g^n$ for $n=0,1,2$. Then the coefficient of $g^0$ is computed as 
\begin{align*}
\prod_{1\le j\le 2m}^{\longrightarrow} \begin{bmatrix}1&0&0\\0&u_j&0\\0&0&1\end{bmatrix}=\begin{bmatrix}1&0&0\\0&\prod_{j=1}^{2m}u_j&0\\0&0&1\end{bmatrix}, 
\end{align*}
that of $g^1$ is computed as 
\begin{align*}
&\sum_{k=1}^{2m}\prod_{1\le j\le k-1}^{\longrightarrow} \begin{bmatrix}1&0&0\\0&u_j&0\\0&0&1\end{bmatrix} \begin{bmatrix}0&2(-1)^k&0\\0&0&2(-1)^k\\0&0&0\end{bmatrix}
\prod_{k\le j\le 2m}^{\longrightarrow} \begin{bmatrix}1&0&0\\0&u_j&0\\0&0&1\end{bmatrix} \\
&=\sum_{k=1}^{2m}\begin{bmatrix} 0&2(-1)^k\prod_{j=k}^{2m}u_j&0 \\ 0&0&2(-1)^k\prod_{j=1}^{k-1}u_j \\ 0&0&0 \end{bmatrix}, 
\end{align*}
and that of $g^2$ is computed as 
\begin{align*}
&\sum_{1\le k<l\le 2m}\prod_{1\le j\le k-1}^{\longrightarrow} \begin{bmatrix}1&0&0\\0&u_j&0\\0&0&1\end{bmatrix} \begin{bmatrix}0&2(-1)^k&0\\0&0&2(-1)^k\\0&0&0\end{bmatrix} \\*
&\eqspace{}\times\prod_{k\le j\le l-1}^{\longrightarrow} \begin{bmatrix}1&0&0\\0&u_j&0\\0&0&1\end{bmatrix} \begin{bmatrix}0&2(-1)^l&0\\0&0&2(-1)^l\\0&0&0\end{bmatrix}
\prod_{l\le j\le 2m}^{\longrightarrow} \begin{bmatrix}1&0&0\\0&u_j&0\\0&0&1\end{bmatrix} \\*
&\eqspace{}+\sum_{k=1}^{2m}\prod_{1\le j\le k-1}^{\longrightarrow} \begin{bmatrix}1&0&0\\0&u_j&0\\0&0&1\end{bmatrix} \begin{bmatrix}0&0&2\\0&0&0\\0&0&0\end{bmatrix}
\prod_{k\le j\le 2m}^{\longrightarrow} \begin{bmatrix}1&0&0\\0&u_j&0\\0&0&1\end{bmatrix} \\
&=\sum_{1\le k<l\le 2m}\begin{bmatrix} 0&0&4(-1)^{k+l}\prod_{j=k}^{l-1}u_j \\ 0&0&0\\0&0&0 \end{bmatrix} +\sum_{k=1}^{2m}\begin{bmatrix}0&0&2\\0&0&0\\0&0&0\end{bmatrix} \\
&=\sum_{1\le k<l\le 2m}\begin{bmatrix} 0&0&4(-1)^{k+l}\prod_{j=k}^{l-1}u_j \\ 0&0&0\\0&0&0 \end{bmatrix} +\begin{bmatrix}0&0&4m\\0&0&0\\0&0&0\end{bmatrix}. 
\end{align*}
Therefore we have 
\begin{align*}
\gamma=\begin{bmatrix} 1&2g\sum_{k=1}^{2m}(-1)^k\prod_{j=k}^{2m}u_j&4g^2\sum_{1\le k<l\le 2m}(-1)^{k+l}\prod_{j=k}^{l-1}u_j+4mg^2 \\
0&\prod_{j=1}^{2m}u_j&2g\sum_{k=1}^{2m}(-1)^k\prod_{j=1}^{k-1}u_j \\ 0&0&1 \end{bmatrix}. 
\end{align*}
Then $ad+b(1-c)$ is computed as 
\begin{align*}
&ad+b(1-c)=\biggl(2g\sum_{k=1}^{2m}(-1)^k\prod_{j=k}^{2m}u_j\biggr)\biggl(2g\sum_{l=1}^{2m}(-1)^l\prod_{j=1}^{l-1}u_j\biggr) \\*
&\hphantom{ad+b(1-c)}\eqspace{}+\biggl(4g^2\sum_{1\le k<l\le 2m}(-1)^{k+l}\prod_{j=k}^{l-1}u_j+4mg^2\biggr)\biggl(1-\prod_{j=1}^{2m}u_j\biggr) \\
&=4g^2\biggl(\sum_{k=1}^{2m}\prod_{j=1}^{2m}u_j+\sum_{1\le k<l\le 2m}(-1)^{k+l}\prod_{j=1}^{2m}u_j\prod_{j=k}^{l-1}u_j+\sum_{1\le l<k\le 2m}(-1)^{k+l}\prod_{j=1}^{l-1}u_j\prod_{j=k}^{2m}u_j\biggr) \\*
&\eqspace{}+4g^2\biggl(\sum_{1\le k<l\le 2m}(-1)^{k+l}\prod_{j=k}^{l-1}u_j-\sum_{1\le k<l\le 2m}(-1)^{k+l}\prod_{j=1}^{2m}u_j\prod_{j=k}^{l-1}u_j+m\biggl(1-\prod_{j=1}^{2m}u_j\biggr)\biggr) \\
&=4g^2\biggl(2m\prod_{j=1}^{2m}u_j+\!\!\sum_{1\le l<k\le 2m}\!\!(-1)^{k+l}\prod_{j=k}^{2m+l-1}u_j+\!\!\sum_{1\le k<l\le 2m}\!\!(-1)^{k+l}\prod_{j=k}^{l-1}u_j+m\biggl(1-\prod_{j=1}^{2m}u_j\biggr)\biggr) \\
&=4g^2\biggl(m\biggl(1+\prod_{j=1}^{2m}u_j\biggr)+\sum_{l'=1}^{2m-1}(-1)^{l'}\sum_{k=2m-l'+1}^{2m}\prod_{j=k}^{k+l'-1}u_j+\sum_{l'=1}^{2m-1}(-1)^{l'}\sum_{k=1}^{2m-l'}\prod_{j=k}^{k+l'-1}u_j\biggr) \\
&=4g^2\biggl(m\biggl(1+\prod_{j=1}^{2m}u_j\biggr)+\sum_{l=1}^{2m-1}(-1)^l\sum_{k=1}^{2m}\prod_{j=k}^{k+l-1}u_j\biggr)=4g^2\Phi_m(u_1,\ldots,u_{2m}), 
\end{align*}
where we put $2m+l-k=:l'$ or $l-k=:l'$ at the 4th equality. Hence by Theorem \ref{thm_Tr_tauflat}, we get 
\[ \Tr(\tau_\flat(\gamma))=\frac{\exp\bigl(-\frac{ad+b(1-c)}{1-c}\bigr)}{1-c}
=\exp\biggl(-4g^2\frac{\Phi_m(u_1,\ldots,u_{2m})}{1-\prod_{j=1}^{2m}u_j}\biggr)\frac{1}{1-\prod_{j=1}^{2m}u_j}. \]

(2) Again, let $u_j=e^{-x_j}\in (0,1)$, and put $p_j:=(\sum_{k=1}^{2m}x_k)/x_j$, so that $\sum_{j=1}^{2m}p_j^{-1}=1$. 
Then by H\"older's inequality (\ref{Holder_ineq}) and (\ref{formula_Trace}), we have 
\begin{align*}
&\biggl\Vert\tau_\nu\biggl(\prod_{1\le j\le 2m}^{\longrightarrow}u_j^{Y(-(-1)^jg)}\biggr)\biggr\Vert_{\Tr}
\le \prod_{j=1}^{2m}\Vert\tau_\nu(e^{-x_jY(-(-1)^jg)})\Vert_{p_j} \\
&=\prod_{j=1}^{2m}\bigl(\Tr\bigl(\tau_\nu(e^{-\sum_{k=1}^{2m}x_kY(-(-1)^jg)}) \bigr)\bigr)^{1/p_j}
=\Tr\bigl(\tau_\nu(e^{-\sum_{k=1}^{2m}x_kY(0)}) \bigr) \\
&=\frac{e^{-\nu\sum_{k=1}^{2m}x_k}}{1-e^{-2\sum_{k=1}^{2m}x_k}}
=\frac{\prod_{j=1}^{2m}u_j^\nu}{1-\prod_{j=1}^{2m}u_j^2}<\infty. 
\end{align*}
Next, we put 
\begin{align*}
\gamma:\hspace{-3pt}&=\begin{bmatrix} \cosh(g)&-\sinh(g) \\ -\sinh(g)&\cosh(g) \end{bmatrix}\biggl(\prod_{1\le j\le 2m}^{\longrightarrow} u_j^{Y(-(-1)^jg)}\biggr)
\begin{bmatrix} \cosh(g)&\sinh(g) \\ \sinh(g)&\cosh(g) \end{bmatrix} \\
&=\prod_{1\le j\le 2m}^{\longrightarrow}\biggl(\begin{bmatrix} \cosh(2g)&(-1)^j\sinh(2g) \\ (-1)^j\sinh(2g)&\cosh(2g) \end{bmatrix}
\begin{bmatrix}u_j^{-1}&0\\0&u_j\end{bmatrix}\biggr), 
\end{align*}
so that $\tr(\gamma)=\Psi^g_m(u_1,\ldots,u_{2m})$, and by Theorem \ref{thm_Tr_taunu}, 
\[ \Tr\biggl(\tau_\nu\biggl(\prod_{1\le j\le 2m}^{\longrightarrow}u_j^{Y(-(-1)^jg)}\biggr)\biggr)=\Tr(\tau_\nu(\gamma))=\frac{\mu(\gamma)^\nu}{1-\mu(\gamma)^2}, \]
where $\mu(\gamma)$ is the eigenvalue of $\gamma\in M(2,\BC)$ such that $|\mu|<1$. 
Then for $u_j\in(0,1)$, since $\gamma\in \Gamma^\circ\cap SL(2,\BR)$, $\tr(\gamma)\in\BR$ holds, and by Theorem \ref{thm_semigroup}\,(2), $|\tr(\gamma)|>2$ holds. 
Also, since $\gamma$ and $\bI$ belong to the same path-connected component of $\Gamma\cap SL(2,\BR)$, $\tr(\gamma)>2$ holds. Hence we have 
\begin{align*}
\mu(\gamma)^{\pm 1}=\frac{1}{2}\bigl(\tr(\gamma)\mp\sqrt{\tr(\gamma)^2-4}\bigr)=\Bigl(\frac{1}{2}\bigl(\sqrt{\tr(\gamma)+2}\mp\sqrt{\tr(\gamma)-2}\bigr)\Bigr)^2, 
\end{align*}
and 
\begin{align*}
\frac{\mu(\gamma)^\nu}{1-\mu(\gamma)^2}&=\frac{\mu(\gamma)^{\nu-1}}{\mu(\gamma)^{-1}-\mu(\gamma)}
=\frac{\bigl(\frac{1}{2}(\sqrt{\tr(\gamma)+2}-\sqrt{\tr(\gamma)-2})\bigr)^{2\nu-2}}{\sqrt{\tr(\gamma)+2}\sqrt{\tr(\gamma)-2}} \\
&=\frac{\bigl(\frac{1}{2}(\sqrt{\Psi^g_m(u_1,\ldots,u_{2m})+2}-\sqrt{\Psi^g_m(u_1,\ldots,u_{2m})-2})\bigr)^{2\nu-2}}{\sqrt{\Psi^g_m(u_1,\ldots,u_{2m})+2}\sqrt{\Psi^g_m(u_1,\ldots,u_{2m})-2}}. 
\end{align*}
Therefore we get the desired formula. 

(3) Since 
\begin{align*}
&\Psi_m^{g/\sqrt{\nu}}(u_1,\ldots,u_{2m})\prod_{j=1}^{2m}u_j \\
&=\cosh^{2m}\biggl(\frac{2g}{\sqrt{\nu}}\biggr)\tr\biggl(\prod_{1\le j\le 2m}^{\longrightarrow}
\biggl(\begin{bmatrix} 1&(-1)^j\tanh(2g/\sqrt{\nu}) \\ (-1)^j\tanh(2g/\sqrt{\nu})&1 \end{bmatrix}
\begin{bmatrix}1&0\\0&u_j^2\end{bmatrix}\biggr)\biggr) \\
&=\Bigl(1+\frac{4mg^2}{\nu}\Bigr)\tr\biggl(\prod_{1\le j\le 2m}^{\longrightarrow}\begin{bmatrix}1&0\\0&u_j^2\end{bmatrix}\biggr)
+\frac{4g^2}{\nu}\sum_{1\le k<l\le 2m}(-1)^{k+l} \\*
&\eqspace{}\times\tr\biggl(\prod_{1\le j\le k-1}^{\longrightarrow} \begin{bmatrix}1&0\\0&u_j^2\end{bmatrix}\begin{bmatrix}0&1\\1&0\end{bmatrix}
\prod_{k\le j\le l-1}^{\longrightarrow} \begin{bmatrix}1&0\\0&u_j^2\end{bmatrix}\begin{bmatrix}0&1\\1&0\end{bmatrix}
\prod_{l\le j\le 2m}^{\longrightarrow}\begin{bmatrix}1&0\\0&u_j^2\end{bmatrix}\biggr)+O(\nu^{-2}) \\
&=\Bigl(1+\frac{4mg^2}{\nu}\Bigr)\biggl(1+\prod_{j=1}^{2m}u_j^2\biggr)+\frac{4g^2}{\nu}\sum_{1\le k<l\le 2m}(-1)^{k+l}\biggl(\prod_{j=k}^{l-1}u_j^2+\prod_{j=l}^{2m+k-1}u_j^2\biggr)+O(\nu^{-2}) \\
&=\Bigl(1+\frac{4mg^2}{\nu}\Bigr)\biggl(1+\prod_{j=1}^{2m}u_j^2\biggr) \\*
&\eqspace{}+\frac{4g^2}{\nu}\biggl(\sum_{l'=1}^{2m-1}(-1)^{l'}\sum_{k=1}^{2m-l'}\prod_{j=k}^{k+l'-1}u_j^2+\sum_{l'=1}^{2m-1}(-1)^{l'}\sum_{k'=2m-l'+1}^{2m}\prod_{j=k'}^{k'+l'-1}u_j^2\biggr)+O(\nu^{-2}) \\
&=\Bigl(1+\frac{4mg^2}{\nu}\Bigr)\biggl(1+\prod_{j=1}^{2m}u_j^2\biggr)+\frac{4g^2}{\nu}\sum_{l=1}^{2m-1}(-1)^l\sum_{k=1}^{2m}\prod_{j=k}^{k+l-1}u_j^2+O(\nu^{-2}) \\
&=1+\prod_{j=1}^{2m}u_j^2+\frac{4g^2}{\nu}\Phi_m(u_1^2,\ldots,u_{2m}^2)+O(\nu^{-2}) 
\end{align*}
holds as $\nu\to\infty$, where we put $(k,l-k)=:(k,l')$ or $(l,2m+k-l)=:(k',l')$ at the 4th equality, we have 
\begin{align*}
&\sqrt{\Psi_m^{g/\sqrt{\nu}}(u_1,\ldots,u_{2m})\pm 2} \\
&=\biggl(\prod_{j=1}^{2m}u_j^{-1}\biggl(1\pm 2\prod_{j=1}^{2m}u_j+\prod_{j=1}^{2m}u_j^2+\frac{4g^2}{\nu}\Phi_m(u_1^2,\ldots,u_{2m}^2)+O(\nu^{-2})\biggr)\biggr)^{1/2} \\
&=\prod_{j=1}^{2m}u_j^{-1/2}\biggl(1\pm\prod_{j=1}^{2m}u_j\biggr)
\biggl(1+\frac{4g^2}{\nu}\biggl(1\pm \prod_{j=1}^{2m}u_j\biggr)^{-2}\Phi_m(u_1^2,\ldots,u_{2m}^2)+O(\nu^{-2})\biggr)^{1/2} \\
&=\prod_{j=1}^{2m}u_j^{-1/2}\biggl(1\pm\prod_{j=1}^{2m}u_j\biggr)
\biggl(1+\frac{2g^2}{\nu}\biggl(1\pm \prod_{j=1}^{2m}u_j\biggr)^{-2}\Phi_m(u_1^2,\ldots,u_{2m}^2)+O(\nu^{-2})\biggr) \\
&=\prod_{j=1}^{2m}u_j^{-1/2}\biggl(1\pm\prod_{j=1}^{2m}u_j+\frac{2g^2}{\nu}\biggl(1\pm \prod_{j=1}^{2m}u_j\biggr)^{-1}\Phi_m(u_1^2,\ldots,u_{2m}^2)+O(\nu^{-2})\biggr), 
\end{align*}
and hence we get 
\begin{align*}
&\Tr\biggl(\tau_\nu\biggl(\prod_{1\le j\le 2m}^{\longrightarrow}u_j^{Y(-(-1)^jg/\sqrt{\nu})}\biggr)\biggr)\prod_{j=1}^{2m} u_j^{-\nu} \\
&=\frac{\Bigl(\frac{1}{2}\Bigl(\sqrt{\Psi_m^{g/\sqrt{\nu}}(u_1,\ldots,u_{2m})+2}-\sqrt{\Psi_m^{g/\sqrt{\nu}}(u_1,\ldots,u_{2m})-2}\Bigr)\Bigr)^{2\nu-2}}
{\sqrt{\Psi_m^{g/\sqrt{\nu}}(u_1,\ldots,u_{2m})+2}\sqrt{\Psi_m^{g/\sqrt{\nu}}(u_1,\ldots,u_{2m})-2}}\prod_{j=1}^{2m} u_j^{-\nu} \\
&=\biggl(\prod_{j=1}^{2m}u_j^{-1/2}\biggl(\prod_{j=1}^{2m}u_j+\frac{g^2}{\nu}\biggl(\frac{\Phi_m(u_1^2,\ldots,u_{2m}^2)}{1+\prod_{j=1}^{2m}u_j}
-\frac{\Phi_m(u_1^2,\ldots,u_{2m}^2)}{1-\prod_{j=1}^{2m}u_j}\biggr)+O(\nu^{-2})\biggr)\biggr)^{2\nu-2} \\*
&\eqspace{}\times\biggl(\prod_{j=1}^{2m}u_j^{-1}\biggl(1+\prod_{j=1}^{2m}u_j+O(\nu^{-1})\biggr)\biggl(1-\prod_{j=1}^{2m}u_j+O(\nu^{-1})\biggr)\biggr)^{-1}\prod_{j=1}^{2m}u_j^{-\nu} \\
&=\biggl(\prod_{j=1}^{2m}u_j^{1/2}\biggl(1\hspace{-1pt}-\hspace{-1pt}\frac{2g^2}{\nu}\frac{\Phi_m(u_1^2,\ldots,u_{2m}^2)}{1-\prod_{j=1}^{2m}u_j^2}\hspace{-1pt}+\hspace{-1pt}O(\nu^{-2})\biggr)\biggr)^{2\nu-2}
\frac{\prod_{j=1}^{2m}u_j^{-\nu}}{\prod_{j=1}^{2m}u_j^{-1}\bigl(1\hspace{-1pt}-\hspace{-1pt}\prod_{j=1}^{2m}u_j^2\hspace{-1pt}+\hspace{-1pt}O(\nu^{-1})\bigr)} \\
&=\biggl(1-\frac{2g^2}{\nu}\frac{\Phi_m(u_1^2,\ldots,u_{2m}^2)}{1-\prod_{j=1}^{2m}u_j^2}+O(\nu^{-2})\biggr)^{2\nu-2}\frac{1}{1-\prod_{j=1}^{2m}u_j^2+O(\nu^{-1})} \\
&\longrightarrow \exp\biggl(-4g^2\frac{\Phi_m(u_1^2,\ldots,u_{2m}^2)}{1-\prod_{j=1}^{2m}u_j^2}\biggr)\frac{1}{1-\prod_{j=1}^{2m}u_j^2} 
=\Tr\biggl(\tau_\flat\biggl(\prod_{1\le j\le 2m}^{\longrightarrow}u_j^{2X(-(-1)^jg)}\biggr)\biggr)
\end{align*}
as $\nu\to\infty$. 
\end{proof}

\begin{proof}[Proof of Theorem \ref{thm_int_exp}]
(1), (2) For $\Re(\lambda\pm\varepsilon)>0$, $(\rd\tau_\flat(X(\pm g))+\lambda\pm\varepsilon)^{-1}$ is expressed as 
\begin{align*}
(\rd\tau_\flat(X(\pm g))+\lambda\pm\varepsilon)^{-1}&=\int_0^\infty e^{-x\rd\tau_\flat(X(\pm g))}e^{-x(\lambda\pm\varepsilon)}\,\rd x
=\int_0^\infty \tau_\flat(e^{-xX(\pm g)})e^{-x(\lambda\pm\varepsilon)}\,\rd x \\
&=\int_0^1 \tau_\flat(u^{X(\pm g)})u^{\lambda\pm\varepsilon-1}\,\rd u, 
\end{align*}
and for $\Re(\lambda\pm\varepsilon)+\nu>0$, $(\rd\tau_\nu(Y(\pm g))+\lambda\pm\varepsilon)^{-1}$ is expressed as 
\begin{align*}
(\rd\tau_\nu(Y(\pm g))+\lambda\pm\varepsilon)^{-1}&=\int_0^\infty e^{-x\rd\tau_\nu(Y(\pm g))}e^{-x(\lambda\pm\varepsilon)}\,\rd x
=\int_0^\infty \tau_\nu(e^{-xY(\pm g)})e^{-x(\lambda\pm\varepsilon)}\,\rd x \\
&=\int_0^1 \tau_\nu(u^{Y(\pm g)})u^{\lambda\pm\varepsilon-1}\,\rd u. 
\end{align*}
Hence we have 
\begin{align*}
&\bigl((\rd\tau_\flat(X(+g))+\lambda+\varepsilon)^{-1}(\rd\tau_\flat(X(-g))+\lambda-\varepsilon)^{-1}\bigr)^m \\
&=\int_{[0,1]^{2m}}\tau_\flat(u_1^{X(+g)}u_2^{X(-g)}\cdots u_{2m-1}^{X(+g)}u_{2m}^{X(-g)})
u_1^{\lambda+\varepsilon-1}u_2^{\lambda-\varepsilon-1}\cdots u_{2m-1}^{\lambda+\varepsilon-1}u_{2m}^{\lambda-\varepsilon-1}\,\rd u_1\cdots \rd u_{2m}, \\
&\bigl((\rd\tau_\nu(Y(+g))+\lambda+\varepsilon)^{-1}(\rd\tau_\nu(Y(-g))+\lambda-\varepsilon)^{-1}\bigr)^m \\
&=\int_{[0,1]^{2m}}\tau_\nu(u_1^{Y(+g)}u_2^{Y(-g)}\cdots u_{2m-1}^{Y(+g)}u_{2m}^{Y(-g)})
u_1^{\lambda+\varepsilon-1}u_2^{\lambda-\varepsilon-1}\cdots u_{2m-1}^{\lambda+\varepsilon-1}u_{2m}^{\lambda-\varepsilon-1}\,\rd u_1\cdots \rd u_{2m}. 
\end{align*}
These converge as Bochner integrals with respect to the trace class norms since 
\begin{align*}
&\bigl\Vert\bigl((\rd\tau_\flat(X(+g))+\lambda+\varepsilon)^{-1}(\rd\tau_\flat(X(-g))+\lambda-\varepsilon)^{-1}\bigr)^m\bigr\Vert_{\Tr} \\
&\le \int_{[0,1]^{2m}}\biggl\Vert\tau_\flat\biggl(\prod_{1\le j\le 2m}^{\longrightarrow} u_j^{X(-(-1)^jg)}\biggr)\biggr\Vert_{\Tr}
\prod_{j=1}^{2m} \bigl|u_j^{\lambda-(-1)^j\varepsilon-1}\bigr|\prod_{j=1}^{2m}\rd u_j \\
&\le\int_{[0,1]^{2m}} \frac{1}{1-\prod_{j=1}^{2m}u_j} \prod_{j=1}^{2m}u_j^{\Re(\lambda-(-1)^j\varepsilon)-1}\,\rd u_j
=\sum_{k=0}^\infty\int_{[0,1]^{2m}} \prod_{j=1}^{2m}u_j^{\Re(\lambda-(-1)^j\varepsilon)+k-1}\,\rd u_j \\
&=\sum_{k=0}^\infty \frac{1}{(\Re(\lambda+\varepsilon)+k)^m(\Re(\lambda-\varepsilon)+k)^m}<\infty, 
\end{align*}
and 
\begin{align*}
&\bigl\Vert\bigl((\rd\tau_\nu(Y(+g))+\lambda+\varepsilon)^{-1}(\rd\tau_\nu(Y(-g))+\lambda-\varepsilon)^{-1}\bigr)^m\bigr\Vert_{\Tr} \\
&\le \int_{[0,1]^{2m}}\biggl\Vert\tau_\nu\biggl(\prod_{1\le j\le 2m}^{\longrightarrow} u_j^{Y(-(-1)^jg)}\biggr)\biggr\Vert_{\Tr}
\prod_{j=1}^{2m} \bigl|u_j^{\lambda-(-1)^j\varepsilon-1}\bigr|\prod_{j=1}^{2m}\rd u_j \\
&\le\int_{[0,1]^{2m}} \frac{\prod_{j=1}^{2m}u_j^\nu}{1-\prod_{j=1}^{2m}u_j^2} \prod_{j=1}^{2m}u_j^{\Re(\lambda-(-1)^j\varepsilon)-1}\,\rd u_j
=\sum_{k=0}^\infty\int_{[0,1]^{2m}} \prod_{j=1}^{2m}u_j^{\Re(\lambda-(-1)^j\varepsilon)+\nu+2k-1}\,\rd u_j \\
&=\sum_{k=0}^\infty \frac{1}{(\Re(\lambda+\varepsilon)+\nu+2k)^m(\Re(\lambda-\varepsilon)+\nu+2k)^m}<\infty 
\end{align*}
hold. Then their traces are computed as 
\begin{align*}
R_m^\flat(\lambda;g,\varepsilon)&=\Tr\bigl(\bigl((\rd\tau_\flat(X(+g))+\lambda+\varepsilon)^{-1}(\rd\tau_\flat(X(-g))+\lambda-\varepsilon)^{-1}\bigr)^m\bigr) \\
&=\int_{[0,1]^{2m}}\Tr\biggl(\tau_\flat\biggl(\prod_{1\le j\le 2m}^{\longrightarrow} u_j^{X(-(-1)^jg)}\biggr)\biggr)
\prod_{j=1}^{2m} u_j^{\lambda-(-1)^j\varepsilon-1}\prod_{j=1}^{2m}\rd u_j, \\
R_m^\nu(\lambda;g,\varepsilon)&=\Tr\bigl(\bigl((\rd\tau_\nu(Y(+g))+\lambda+\varepsilon)^{-1}(\rd\tau_\nu(Y(-g))+\lambda-\varepsilon)^{-1}\bigr)^m\bigr) \\
&=\int_{[0,1]^{2m}}\Tr\biggl(\tau_\nu\biggl(\prod_{1\le j\le 2m}^{\longrightarrow} u_j^{Y(-(-1)^jg)}\biggr)\biggr)
\prod_{j=1}^{2m} u_j^{\lambda-(-1)^j\varepsilon-1}\prod_{j=1}^{2m}\rd u_j. 
\end{align*}
Hence by Lemma \ref{lem_Trtau}, we get the desired formulas. 

(3) Since 
\begin{align*}
&\int_{[0,1]^{2m}}\biggl\Vert\tau_\nu\biggl(\prod_{1\le j\le 2m}^{\longrightarrow} u_j^{Y(-(-1)^jg/\sqrt{\nu})}\biggr)\biggr\Vert_{\Tr}
\prod_{j=1}^{2m} \bigl|u_j^{2\lambda-\nu-2(-1)^j\varepsilon-1}\bigr|\prod_{j=1}^{2m}\rd u_j \\
&\le\int_{[0,1]^{2m}} \frac{\prod_{j=1}^{2m}u_j^\nu}{1-\prod_{j=1}^{2m}u_j^2} \prod_{j=1}^{2m}u_j^{\Re(2\lambda-2(-1)^j\varepsilon)-\nu-1}\,\rd u_j \\
&=\sum_{k=0}^\infty \frac{1}{(2\Re(\lambda+\varepsilon)+2k)^m(2\Re(\lambda-\varepsilon)+2k)^m}<\infty 
\end{align*}
holds uniformly in $\nu>0$, by the dominated convergence theorem, we get 
\begin{align*}
&\lim_{\nu\to\infty}2^{2m}R_m^\nu\biggl(2\lambda-\nu;\frac{g}{\sqrt{\nu}},2\varepsilon\biggr) \\
&=\lim_{\nu\to\infty}2^{2m}\int_{[0,1]^{2m}}\Tr\biggl(\tau_\nu\biggl(\prod_{1\le j\le 2m}^{\longrightarrow} u_j^{Y(-(-1)^jg/\sqrt{\nu})}\biggr)\biggr)
\prod_{j=1}^{2m} u_j^{2\lambda-\nu-2(-1)^j\varepsilon-1}\prod_{j=1}^{2m}\rd u_j \\
&=2^{2m}\int_{[0,1]^{2m}}\lim_{\nu\to\infty}\Tr\biggl(\tau_\nu\biggl(\prod_{1\le j\le 2m}^{\longrightarrow} u_j^{Y(-(-1)^jg/\sqrt{\nu})}\biggr)\biggr)
\prod_{j=1}^{2m} u_j^{2\lambda-\nu-2(-1)^j\varepsilon-1}\prod_{j=1}^{2m}\rd u_j \\
&=2^{2m}\int_{[0,1]^{2m}}\Tr\biggl(\tau_\flat\biggl(\prod_{1\le j\le 2m}^{\longrightarrow} u_j^{2X(-(-1)^jg)}\biggr)\biggr)
\prod_{j=1}^{2m} u_j^{2\lambda-2(-1)^j\varepsilon-1}\prod_{j=1}^{2m}\rd u_j \\
&=\int_{[0,1]^{2m}}\Tr\biggl(\tau_\flat\biggl(\prod_{1\le j\le 2m}^{\longrightarrow} u_j^{X(-(-1)^jg)}\biggr)\biggr)
\prod_{j=1}^{2m} u_j^{\lambda-(-1)^j\varepsilon-1}\prod_{j=1}^{2m}\rd u_j \\
&=R_m^\flat(\lambda;g,\varepsilon). \qedhere
\end{align*}
\end{proof}

By Lemma \ref{lem_limit} and Theorems \ref{thm_zeta_sum}, \ref{thm_int_exp}, we immediately get the following. 
\begin{corollary}\label{cor_zeta}
\begin{enumerate}
\item Let $g,\Delta,\varepsilon\in\BR$, $\lambda\in\BC$. We assume $\Re\lambda-|\varepsilon|>0$ and $|\Delta|<|\lambda-|\varepsilon||$. 
Then for $n\in\BZ_{\ge 2}$, we have 
\begin{multline*}
\zeta\bigl(\widetilde{\rH}_\flat^{(g,\Delta,\varepsilon)};n,\lambda\bigr) 
=\zeta(n,\lambda+\varepsilon)+\zeta(n,\lambda-\varepsilon)+\frac{1}{\Gamma(n)}\sum_{m=1}^\infty \frac{\Delta^{2m}}{m} \\
{}\times\int_{[0,1]^{2m}} \exp\biggl(-4g^2\frac{\Phi_m(u_1,\ldots,u_{2m})}{1-\prod_{j=1}^{2m}u_j}\biggr)
\frac{\bigl(-\log\prod_{j=1}^{2m}u_j\bigr)^n}{1-\prod_{j=1}^{2m}u_j}\prod_{j=1}^{2m} u_j^{\lambda-(-1)^j\varepsilon-1}\,\rd u_j, 
\end{multline*}
where $\Phi_m(u_1,\ldots,u_{2m})$ is as in (\ref{formula_Phigu}). 
\item Let $\nu>0$, $g,\Delta,\varepsilon\in\BR$, $\lambda\in\BC$. 
We assume $\Re\lambda-|\varepsilon|+\nu>0$ and $|\Delta|<|\lambda-|\varepsilon|+\nu|$. 
Then for $n\in\BZ_{\ge 2}$, we have 
\begin{align*}
\zeta\bigl(\widetilde{\rH}_\nu^{(g,\Delta,\varepsilon)};n,\lambda\bigr)
=2^{-n}\zeta\biggl(n,\frac{\lambda+\varepsilon+\nu}{2}\biggr)+2^{-n}\zeta\biggl(n,\frac{\lambda-\varepsilon+\nu}{2}\biggr)+\frac{1}{\Gamma(n)}\sum_{m=1}^\infty \frac{\Delta^{2m}}{m} \\*
{}\times\int_{[0,1]^{2m}} \frac{\bigl(\frac{1}{2}(\sqrt{\Psi^g_m(u_1,\ldots,u_{2m})+2}-\sqrt{\Psi^g_m(u_1,\ldots,u_{2m})-2})\bigr)^{2\nu-2}}
{\sqrt{\Psi^g_m(u_1,\ldots,u_{2m})+2}\sqrt{\Psi^g_m(u_1,\ldots,u_{2m})-2}} \\*
{}\times\biggl(-\log\prod_{j=1}^{2m}u_j\biggr)^n\prod_{j=1}^{2m} u_j^{\lambda-(-1)^j\varepsilon-1}\,\rd u_j, 
\end{align*}
where $\Psi^g_m(u_1,\ldots,u_{2m})$ is as in (\ref{formula_Psigu}). 
\item Let $g,\Delta,\varepsilon\in\BR$, $\lambda\in\BC$. We assume $\Re\lambda-|\varepsilon|>0$ and $|\Delta|<|\lambda-|\varepsilon||$. 
Then for $n\in\BZ_{\ge 2}$, we have 
\[ 2^n\lim_{\nu\to\infty}\zeta\bigl(\widetilde{\rH}_\nu^{(g/\sqrt{\nu},2\Delta,2\varepsilon)};n,2\lambda-\nu\bigr)
=\zeta\bigl(\widetilde{\rH}_\flat^{(g,\Delta,\varepsilon)};n,\lambda\bigr). \]
\end{enumerate}
\end{corollary}

We recall 
\begin{align*}
&R_m^+(\lambda;g,\varepsilon):=R_m^{1/2}(\lambda;g,\varepsilon)+R_m^{3/2}(\lambda;g,\varepsilon), \\
&R_m^-(\lambda;g,\varepsilon):=R_m^{1/2}(\lambda;g,\varepsilon)-R_m^{3/2}(\lambda;g,\varepsilon) 
\end{align*}
from (\ref{def_Rmpm}). Then their integral expressions are given as follows. 

\begin{corollary}\label{cor_int_exp}
Let $g\in\BR$, $\lambda,\varepsilon\in\BC$, $m\in\BZ_{\ge 1}$. If $\Re\lambda-|\Re\varepsilon|+\frac{1}{2}>0$, then we have 
\begin{align*}
R_m^+(\lambda;g,\varepsilon)&=\int_{[0,1]^{2m}} \frac{1}{\sqrt{\Psi^g_m(u_1,\ldots,u_{2m})-2}}\prod_{j=1}^{2m} u_j^{\lambda-(-1)^j\varepsilon-1}\,\rd u_j, \\
R_m^-(\lambda;g,\varepsilon)&=\int_{[0,1]^{2m}} \frac{1}{\sqrt{\Psi^g_m(u_1,\ldots,u_{2m})+2}}\prod_{j=1}^{2m} u_j^{\lambda-(-1)^j\varepsilon-1}\,\rd u_j, 
\end{align*}
where $\Psi_m^g(u_1,\ldots,u_{2m})$ is as in (\ref{formula_Psigu}). 
\end{corollary}


\begin{proof}
For $\nu=1/2, 3/2$, since 
\begin{align*}
\frac{\bigl(\frac{1}{2}(\sqrt{\Psi+2}-\sqrt{\Psi-2})\bigr)^{2\cdot\frac{1}{2}-2}}{\sqrt{\Psi+2}\sqrt{\Psi-2}}
&=\frac{\frac{1}{2}(\sqrt{\Psi+2}+\sqrt{\Psi-2})}{\sqrt{\Psi+2}\sqrt{\Psi-2}}=\frac{1}{2}\biggl(\frac{1}{\sqrt{\Psi-2}}+\frac{1}{\sqrt{\Psi+2}}\biggr), \\
\frac{\bigl(\frac{1}{2}(\sqrt{\Psi+2}-\sqrt{\Psi-2})\bigr)^{2\cdot\frac{3}{2}-2}}{\sqrt{\Psi+2}\sqrt{\Psi-2}}
&=\frac{\frac{1}{2}(\sqrt{\Psi+2}-\sqrt{\Psi-2})}{\sqrt{\Psi+2}\sqrt{\Psi-2}}=\frac{1}{2}\biggl(\frac{1}{\sqrt{\Psi-2}}-\frac{1}{\sqrt{\Psi+2}}\biggr) 
\end{align*}
hold, we have 
\begin{align*}
&R_m^{1/2}(\lambda;g,\varepsilon) \\
&=\frac{1}{2}\int_{[0,1]^{2m}} \biggl(\frac{1}{\sqrt{\Psi^g_m(u_1,\ldots,u_{2m})\hspace{-1pt}-\hspace{-1pt}2}}+\frac{1}{\sqrt{\Psi^g_m(u_1,\ldots,u_{2m})\hspace{-1pt}+\hspace{-1pt}2}}\biggr)
\prod_{j=1}^{2m} u_j^{\lambda-(-1)^j\varepsilon-1}\,\rd u_j, \\
&R_m^{3/2}(\lambda;g,\varepsilon) \\
&=\frac{1}{2}\int_{[0,1]^{2m}} \biggl(\frac{1}{\sqrt{\Psi^g_m(u_1,\ldots,u_{2m})\hspace{-1pt}-\hspace{-1pt}2}}-\frac{1}{\sqrt{\Psi^g_m(u_1,\ldots,u_{2m})\hspace{-1pt}+\hspace{-1pt}2}}\biggr)
\prod_{j=1}^{2m} u_j^{\lambda-(-1)^j\varepsilon-1}\,\rd u_j.
\end{align*}
Hence we get the desired formulas. 
\end{proof}

\begin{example}\label{ex_R1}
Let $m=1$. Then we have 
\begin{align*}
&\Phi_1(u,v)=1+uv-u-v=(1-u)(1-v), \\
&\Psi^g_1(u,v)\pm 2 \\
&=\tr\biggl(\begin{bmatrix} \cosh(2g)&-\sinh(2g) \\ -\sinh(2g)&\cosh(2g) \end{bmatrix}\begin{bmatrix}u^{-1}&0\\0&u\end{bmatrix}
\begin{bmatrix} \cosh(2g)&\sinh(2g) \\ \sinh(2g)&\cosh(2g) \end{bmatrix}\begin{bmatrix}v^{-1}&0\\0&v\end{bmatrix}\biggr)\pm 2 \\
&=(u^{-1}v^{-1}+uv\pm 2)\cosh^2(2g)-(uv^{-1}+u^{-1}v\pm 2)\sinh^2(2g) \\
&=\frac{\cosh^2(2g)}{uv}\bigl((1\pm uv)^2-\tanh^2(2g)(u\pm v)^2\bigr) \\
&=\frac{1}{uv}\bigl((1\pm uv)^2+\sinh^2(2g)(1-u^2)(1-v^2)), 
\end{align*}
and hence, for $g\in\BR$, $\lambda,\varepsilon\in\BC$ with $\Re\lambda-|\Re\varepsilon|>0$ resp. $\Re\lambda-|\Re\varepsilon|+\frac{1}{2}>0$, we have 
\begin{align*}
R_1^\flat(\lambda;g,\varepsilon)&=\int_{[0,1]^2}\exp\biggl(-4g^2\frac{(1-u)(1-v)}{1-uv}\biggr)\frac{u^{\lambda+\varepsilon-1}v^{\lambda-\varepsilon-1}}{1-uv}\,\rd u\rd v, \\
R_1^\pm(\lambda;g,\varepsilon)&=\sech(2g)\int_{[0,1]^2} \frac{u^{\lambda+\varepsilon-\frac{1}{2}}v^{\lambda-\varepsilon-\frac{1}{2}}}{\sqrt{(1\mp uv)^2-\tanh^2(2g)(u\mp v)^2}}\,\rd u\rd v \\
&=\int_{[0,1]^2} \frac{u^{\lambda+\varepsilon-\frac{1}{2}}v^{\lambda-\varepsilon-\frac{1}{2}}}{\sqrt{(1\mp uv)^2+\sinh^2(2g)(1-u^2)(1-v^2)}}\,\rd u\rd v. 
\end{align*}
\end{example}

\section{Computation of the first term of 1pQRM}\label{section_1pQRM}

In this section, we seek another expression of 
\begin{align*}
R_1^\flat(\lambda;g,\varepsilon)
&=\Tr\bigl((\ra^\dagger\ra+g(\ra+\ra^\dagger)+g^2+\varepsilon+\lambda)^{-1}(\ra^\dagger\ra-g(\ra+\ra^\dagger)+g^2-\varepsilon+\lambda)^{-1}\bigr).
\end{align*}
Then by Example \ref{ex_R1}, for $g\in\BR$, $\lambda,\varepsilon\in\BC$ with $\Re\lambda-|\Re\varepsilon|>0$, we have 
\begin{align*}
R_1^\flat(\lambda;g,\varepsilon)&=\int_{[0,1]^2}\exp\biggl(-4g^2\frac{(1-u)(1-v)}{1-uv}\biggr)\frac{u^{\lambda+\varepsilon-1}v^{\lambda-\varepsilon-1}}{1-uv}\,\rd u\rd v \\
&=\sum_{n=0}^\infty \frac{(-4g^2)^n}{n!}\int_{[0,1]^2} \frac{(1-u)^n(1-v)^n}{(1-uv)^{n+1}}u^{\lambda+\varepsilon-1}v^{\lambda-\varepsilon-1}\,\rd u\rd v. 
\end{align*}
Now, for $n\in\BZ_{\ge 0}$, $\lambda,\varepsilon\in\BC$ with $\Re\lambda-|\Re\varepsilon|>0$, we put 
\[ J_n^\flat(\lambda,\varepsilon):=\int_{[0,1]^2} \frac{(1-u)^n(1-v)^n}{(1-uv)^{n+1}}u^{\lambda+\varepsilon-1}v^{\lambda-\varepsilon-1}\,\rd u\rd v, \]
so that 
\[ R_1^\flat(\lambda;g,\varepsilon)=\sum_{n=0}^\infty \frac{(-4g^2)^n}{n!}J_n^\flat(\lambda,\varepsilon) \]
holds. 

\begin{remark}
The integral $J_n^\flat(\lambda,\varepsilon)$ is regarded as a generalization of Beukers' integral \cite{Be1} used for the proof of the irrationality of $\zeta(2)$ after Ap\'ery's work \cite{Ap, vdP}, 
\[ J_n^\flat(n+1,0)=\int_{[0,1]^2} \frac{(1-u)^n(1-v)^nu^nv^n}{(1-uv)^{n+1}}\,\rd u\rd v. \]
The original proof of the irrationality by Ap\'ery uses the sequences $\{A_n\}$, $\{B_n\}$ given by 
\begin{align*}
A_n&:=\sum_{k=0}^n\binom{n}{k}^2\binom{n+k}{k}, \\
B_n&:=\sum_{k=0}^n\binom{n}{k}^2\binom{n+k}{k}\biggl(2\sum_{m=1}^n\frac{(-1)^{m-1}}{m^2}+\sum_{m=1}^k\frac{(-1)^{n+m-1}}{m^2\binom{n}{m}\binom{n+m}{m}}\biggr). 
\end{align*}
These are characterized by the common recurrence relation with the distinct initial conditions, 
\begin{align*}
n^2A_n-(11n^2-11n+3)A_{n-1}-(n-1)^2A_{n-2}&=0, & A_0&=1, & A_1&=3, \\
n^2B_n-(11n^2-11n+3)B_{n-1}-(n-1)^2B_{n-2}&=0, & B_0&=0, & B_1&=5. 
\end{align*}
Then Beukers' integral and Ap\'ery's numbers are related as 
\begin{equation}\label{formula_Apery_Beukers}
(-1)^nJ_n^\flat(n+1,0)=A_n\zeta(2)-B_n. 
\end{equation}
For a proof of this equality, see, e.g., \cite{JT}. For more general identities, see also \cite{Ne}. 
By \cite{Be2}, the generating function $\sum_{n=0}^\infty A_nt^n$ of $\{A_n\}$ becomes a modular form of weight 1 if we take $t$ to be an appropriate modular function with respect to $\Gamma_1(5)$. 
\end{remark}

First we prove the following. 
\begin{proposition}\label{prop_Jflat}
For $n\in\BZ_{\ge 0}$, $\lambda,\varepsilon\in\BC$ with $\lambda\pm\varepsilon\notin-\BZ_{\ge 0}$, we have 
\[ J_n^\flat(\lambda,\varepsilon)
=\sum_{k=0}^\infty \frac{n!(k+1)_n}{(\lambda+\varepsilon+k)_{n+1}(\lambda-\varepsilon+k)_{n+1}}. \]
\end{proposition}

\begin{proof}
Since $(1-uv)^{-(n+1)}=\sum_{k=0}^\infty \frac{(n+1)_k}{k!}(uv)^k$ and $\frac{(n+1)_k}{k!}=\frac{(k+n)!}{k!n!}=\frac{(k+1)_n}{n!}$ hold, we have 
\begin{align*}
J_n^\flat(\lambda,\varepsilon)&=\sum_{k=0}^\infty \frac{(n+1)_k}{k!}\int_{[0,1]^2}(uv)^k(1-u)^n(1-v)^n u^{\lambda+\varepsilon-1}v^{\lambda-\varepsilon-1}\,\rd u\rd v \\
&=\sum_{k=0}^\infty \frac{(k+1)_n}{n!}\int_0^1 (1-u)^n u^{\lambda+\varepsilon+k-1}\,\rd u \int_0^1 (1-v)^n v^{\lambda-\varepsilon+k-1}\,\rd v \\
&=\sum_{k=0}^\infty \frac{(k+1)_n}{n!}\frac{\Gamma(n+1)\Gamma(\lambda+\varepsilon+k)}{\Gamma(\lambda+\varepsilon+k+n+1)}
\frac{\Gamma(n+1)\Gamma(\lambda-\varepsilon+k)}{\Gamma(\lambda-\varepsilon+k+n+1)} \\
&=\sum_{k=0}^\infty \frac{n!(k+1)_n}{(\lambda+\varepsilon+k)_{n+1}(\lambda-\varepsilon+k)_{n+1}}. \qedhere
\end{align*}
\end{proof}

Then each summand has the following partial fraction decomposition. 
\begin{lemma}\label{lem_partfrac}
\begin{align*}
\frac{n!(x+1)_n}{(\lambda+\varepsilon+x)_{n+1}(\lambda-\varepsilon+x)_{n+1}}
=\frac{(-1)^n}{2\varepsilon}\sum_{l=0}^n\binom{n}{l}\biggl(-\frac{(\lambda+\varepsilon-n+l)_n}{(1+2\varepsilon)_l(1-2\varepsilon)_{n-l}}\frac{1}{\lambda+\varepsilon+x+l} \\*
{}+\frac{(\lambda-\varepsilon-n+l)_n}{(1+2\varepsilon)_{n-l}(1-2\varepsilon)_l}\frac{1}{\lambda-\varepsilon+x+l}\biggr). 
\end{align*}
\end{lemma}

\begin{proof}
Put 
\[ f(x):=\frac{n!(x+1)_n}{(\lambda+\varepsilon+x)_{n+1}(\lambda-\varepsilon+x)_{n+1}}, \]
and let 
\[ f(x)=\sum_{l=0}^n\biggl(\frac{a_l}{\lambda+\varepsilon+x+l}+\frac{b_l}{\lambda-\varepsilon+x+l}\biggr). \]
be its partial fraction decomposition. Then we have 
\begin{align*}
a_l&=\lim_{x\to -\lambda-\varepsilon-l}(\lambda+\varepsilon+x+l)f(x) \\
&=\lim_{x\to -\lambda-\varepsilon-l}\frac{n!(x+1)_n}{(\lambda+\varepsilon+x)_l(\lambda+\varepsilon+x+l+1)_{n-l}(\lambda-\varepsilon+x)_{n+1}} \\
&=\frac{n!(-\lambda-\varepsilon-l+1)_n}{(-l)_l(1)_{n-l}(-2\varepsilon-l)_{n+1}}
=\frac{(-1)^nn!(\lambda+\varepsilon-n+l)_n}{(-1)^l l!(n-l)!(-2\varepsilon-l)_l(-2\varepsilon)(-2\varepsilon+1)_{n-l}} \\
&=-\frac{(-1)^n}{2\varepsilon}\binom{n}{l}\frac{(\lambda+\varepsilon-n+l)_n}{(1+2\varepsilon)_l(1-2\varepsilon)_{n-l}}, 
\end{align*}
and similarly, 
\[ b_l=\frac{(-1)^n}{2\varepsilon}\binom{n}{l}\frac{(\lambda-\varepsilon-n+l)_n}{(1+2\varepsilon)_{n-l}(1-2\varepsilon)_l}. \qedhere \]
\end{proof}

We also have the following. 
\begin{lemma}\label{lem_partfrac2}
\[ \sum_{l=0}^n\binom{n}{l}\frac{(\lambda+\varepsilon-n+l)_n}{(1+2\varepsilon)_l(1-2\varepsilon)_{n-l}}
=\sum_{l=0}^n\binom{n}{l}\frac{(\lambda-\varepsilon-n+l)_n}{(1+2\varepsilon)_{n-l}(1-2\varepsilon)_l}. \]
\end{lemma}

\begin{proof}
By Lemma \ref{lem_partfrac}, we have 
\begin{align*}
&\frac{(-1)^n2\varepsilon n!(x+1)_nx}{(\lambda+\varepsilon+x)_{n+1}(\lambda-\varepsilon+x)_{n+1}} \\
&=\sum_{l=0}^n\binom{n}{l}\biggl(-\frac{(\lambda+\varepsilon-n+l)_n}{(1+2\varepsilon)_l(1-2\varepsilon)_{n-l}}\frac{x}{\lambda+\varepsilon+x+l}
+\frac{(\lambda-\varepsilon-n+l)_n}{(1+2\varepsilon)_{n-l}(1-2\varepsilon)_l}\frac{x}{\lambda-\varepsilon+x+l}\biggr). 
\end{align*}
Then by taking the limit $x\to\infty$, we get 
\[ \sum_{l=0}^n\binom{n}{l}\biggl(-\frac{(\lambda+\varepsilon-n+l)_n}{(1+2\varepsilon)_l(1-2\varepsilon)_{n-l}}
+\frac{(\lambda-\varepsilon-n+l)_n}{(1+2\varepsilon)_{n-l}(1-2\varepsilon)_l}\biggr)=0. \qedhere \]
\end{proof}

Now we prove the following. 
\begin{theorem}\label{thm_Jflat}
For $n\in\BZ_{\ge 0}$, $\lambda,\varepsilon\in\BC$ with $\lambda\pm\varepsilon\notin-\BZ_{\ge 0}$, $\varepsilon\notin\bigl\{\pm\frac{1}{2},\pm 1,\ldots,\pm\frac{n}{2}\bigr\}$, we have 
\[ J_n^\flat(\lambda,\varepsilon)
=(-1)^nA_n^\flat(\lambda,\varepsilon)\sum_{k=0}^\infty \frac{1}{(\lambda+\varepsilon+k)(\lambda-\varepsilon+k)}+(-1)^nB_n^\flat(\lambda,\varepsilon), \]
where $A_0^\flat(\lambda,\varepsilon):=1$, $B_0^\flat(\lambda,\varepsilon):=0$, and for $n\ge 1$, 
\begin{align*}
&A_n^\flat(\lambda,\varepsilon):=\sum_{l=0}^n\binom{n}{l}\frac{(\lambda+\varepsilon-n+l)_n}{(1+2\varepsilon)_l(1-2\varepsilon)_{n-l}}
=\sum_{l=0}^n\binom{n}{l}\frac{(\lambda-\varepsilon-n+l)_n}{(1+2\varepsilon)_{n-l}(1-2\varepsilon)_l} \\*
&=\sum_{m=1}^n\frac{(-1)^mm}{2\cdot n!}\sum_{l=m}^{n}\binom{n}{l}\binom{n}{l-m}\biggl(\lambda-\frac{m}{2}-n+l\biggr)_n
\biggl(\frac{1}{\varepsilon-\frac{m}{2}}-\frac{1}{\varepsilon+\frac{m}{2}}\biggr), \\
&B_n^\flat(\lambda,\varepsilon) \\*
:\hspace{-3pt}&=\frac{1}{2\varepsilon}\sum_{l=1}^n\sum_{k=0}^{l-1}\binom{n}{l}\biggl(\frac{(\lambda+\varepsilon-n+l)_n}{(1+2\varepsilon)_l(1-2\varepsilon)_{n-l}}\frac{1}{\lambda+\varepsilon+k}
-\frac{(\lambda-\varepsilon-n+l)_n}{(1+2\varepsilon)_{n-l}(1-2\varepsilon)_l}\frac{1}{\lambda-\varepsilon+k}\biggr) \\
&=\sum_{m=1}^n\frac{(-1)^m}{2\cdot n!}\sum_{l=m}^n\sum_{k=0}^{m-1}\binom{n}{l}\binom{n}{l-m}\frac{\bigl(\lambda-\frac{m}{2}-n+l\bigr)_n}{\lambda-\frac{m}{2}+k}
\biggl(-\frac{1}{\varepsilon-\frac{m}{2}}+\frac{1}{\varepsilon+\frac{m}{2}}\biggr). 
\end{align*}
\end{theorem}

\begin{proof}
When $n=0$, this is clear from Proposition \ref{prop_Jflat}. 
When $n\ge 1$, for any $N\in\BZ_{>0}$, by Lemma \ref{lem_partfrac}, we have 
\begin{align*}
&(-1)^n\sum_{k=0}^N \frac{n!(k+1)_n}{(\lambda+\varepsilon+k)_{n+1}(\lambda-\varepsilon+k)_{n+1}} \\
&=\frac{1}{2\varepsilon}\sum_{k=0}^N \sum_{l=0}^n\binom{n}{l}\biggl(-\frac{(\lambda+\varepsilon-n+l)_n}{(1\hspace{-1pt}+\hspace{-1pt}2\varepsilon)_l(1\hspace{-1pt}-\hspace{-1pt}2\varepsilon)_{n-l}}
\frac{1}{\lambda\hspace{-1pt}+\hspace{-1pt}\varepsilon\hspace{-1pt}+\hspace{-1pt}k\hspace{-1pt}+\hspace{-1pt}l}
+\frac{(\lambda-\varepsilon-n+l)_n}{(1\hspace{-1pt}+\hspace{-1pt}2\varepsilon)_{n-l}(1\hspace{-1pt}-\hspace{-1pt}2\varepsilon)_l}
\frac{1}{\lambda\hspace{-1pt}-\hspace{-1pt}\varepsilon\hspace{-1pt}+\hspace{-1pt}k\hspace{-1pt}+\hspace{-1pt}l}\biggr) \\
&=\frac{1}{2\varepsilon}\sum_{l=0}^n\sum_{k=l}^{N+l} \binom{n}{l}\biggl(-\frac{(\lambda+\varepsilon-n+l)_n}{(1+2\varepsilon)_l(1-2\varepsilon)_{n-l}}\frac{1}{\lambda+\varepsilon+k}
+\frac{(\lambda-\varepsilon-n+l)_n}{(1+2\varepsilon)_{n-l}(1-2\varepsilon)_l}\frac{1}{\lambda-\varepsilon+k}\biggr) \\
&=\frac{1}{2\varepsilon}\sum_{l=0}^n\sum_{k=0}^N \binom{n}{l}\biggl(-\frac{(\lambda+\varepsilon-n+l)_n}{(1+2\varepsilon)_l(1-2\varepsilon)_{n-l}}\frac{1}{\lambda+\varepsilon+k}
+\frac{(\lambda-\varepsilon-n+l)_n}{(1+2\varepsilon)_{n-l}(1-2\varepsilon)_l}\frac{1}{\lambda-\varepsilon+k}\biggr) \\*
&\eqspace{}+\frac{1}{2\varepsilon}\sum_{l=1}^n\sum_{k=0}^{l-1}\binom{n}{l}\biggl(\frac{(\lambda+\varepsilon-n+l)_n}{(1+2\varepsilon)_l(1-2\varepsilon)_{n-l}}\frac{1}{\lambda+\varepsilon+k}
-\frac{(\lambda-\varepsilon-n+l)_n}{(1+2\varepsilon)_{n-l}(1-2\varepsilon)_l}\frac{1}{\lambda-\varepsilon+k}\biggr) \\*
&\eqspace{}+\frac{1}{2\varepsilon}\sum_{l=1}^n\sum_{k=N+1}^{N+l}\binom{n}{l}\biggl(-\frac{(\lambda+\varepsilon-n+l)_n}{(1\hspace{-1pt}+\hspace{-1pt}2\varepsilon)_l(1\hspace{-1pt}-\hspace{-1pt}2\varepsilon)_{n-l}}
\frac{1}{\lambda\hspace{-1pt}+\hspace{-1pt}\varepsilon\hspace{-1pt}+\hspace{-1pt}k}
+\frac{(\lambda-\varepsilon-n+l)_n}{(1\hspace{-1pt}+\hspace{-1pt}2\varepsilon)_{n-l}(1\hspace{-1pt}-\hspace{-1pt}2\varepsilon)_l}
\frac{1}{\lambda\hspace{-1pt}-\hspace{-1pt}\varepsilon\hspace{-1pt}+\hspace{-1pt}k}\biggr) \\
&=\frac{1}{2\varepsilon}\sum_{l=0}^n\binom{n}{l}\frac{(\lambda+\varepsilon-n+l)_n}{(1+2\varepsilon)_l(1-2\varepsilon)_{n-l}}
\sum_{k=0}^N \biggl(-\frac{1}{\lambda+\varepsilon+k}+\frac{1}{\lambda-\varepsilon+k}\biggr) \\*
&\eqspace{}+\frac{1}{2\varepsilon}\sum_{l=1}^n\sum_{k=0}^{l-1}\binom{n}{l}\biggl(\frac{(\lambda+\varepsilon-n+l)_n}{(1+2\varepsilon)_l(1-2\varepsilon)_{n-l}}\frac{1}{\lambda+\varepsilon+k}
-\frac{(\lambda-\varepsilon-n+l)_n}{(1+2\varepsilon)_{n-l}(1-2\varepsilon)_l}\frac{1}{\lambda-\varepsilon+k}\biggr) \\*
&\eqspace{}+O(N^{-1}), 
\end{align*}
where we have used Lemma \ref{lem_partfrac2} at the 4th equality. Then by taking the limit $N\to\infty$, we get the former equalities of $A_n^\flat(\lambda,\varepsilon), B_n^\flat(\lambda,\varepsilon)$. 
To prove the latter equalities, we consider the partial fraction decompositions of $A_n^\flat(\lambda,\varepsilon), B_n^\flat(\lambda,\varepsilon)$ with respect to $\varepsilon$. 
For $A_n^\flat(\lambda,\varepsilon)$, since the degrees of numerators and denominators are equal, it suffices to compute the limit $\varepsilon\to\infty$ and the residues at all poles. 
First, the limit $\varepsilon\to\infty$ is computed as 
\[ \lim_{\varepsilon\to\infty} A_n^\flat(\lambda,\varepsilon)=\sum_{l=0}^n\binom{n}{l}\frac{1}{2^l(-2)^{n-l}}=\biggl(\frac{1}{2}-\frac{1}{2}\biggr)^n=0. \]
Next, for $m=1,2,\ldots,n$, the residue at the pole $\varepsilon=\frac{m}{2}$ is computed as 
\begin{align*}
\lim_{\varepsilon\to\frac{m}{2}}\biggl(\varepsilon-\frac{m}{2}\biggr) A_n^\flat(\lambda,\varepsilon)
&=\lim_{\varepsilon\to\frac{m}{2}}2\varepsilon\biggl(\varepsilon-\frac{m}{2}\biggr) \sum_{l=0}^n\binom{n}{l}\frac{(-1)^{n-l}(\lambda+\varepsilon-n+l)_n}{(2\varepsilon-n+l)_{n-l}(2\varepsilon)(2\varepsilon+1)_l}\\
&=\lim_{\varepsilon\to\frac{m}{2}}\varepsilon(2\varepsilon-m)\sum_{l=0}^n\binom{n}{l}(-1)^{n-l}\frac{(\lambda+\varepsilon-n+l)_n}{(2\varepsilon-n+l)_{n+1}} \\
&=\lim_{\varepsilon\to\frac{m}{2}}\varepsilon\sum_{l=0}^{n-m}\binom{n}{l}(-1)^{n-l}\frac{(\lambda+\varepsilon-n+l)_n}{(2\varepsilon-n+l)_{n-l-m}(2\varepsilon-m+1)_{l+m}} \\
&=\frac{m}{2}\sum_{l=0}^{n-m}\binom{n}{l}(-1)^{n-l}\frac{\bigl(\lambda+\frac{m}{2}-n+l\bigr)_n}{(m-n+l)_{n-l-m}(1)_{l+m}} \\
&=\frac{m}{2}\sum_{l=0}^{n-m}\binom{n}{l}(-1)^{m}\frac{\bigl(\lambda+\frac{m}{2}-n+l\bigr)_n}{(n-l-m)!(l+m)!} \\
&=\frac{(-1)^mm}{2\cdot n!}\sum_{l=0}^{n-m}\binom{n}{l}\binom{n}{l+m}\biggl(\lambda+\frac{m}{2}-n+l\biggr)_n \\
&=\frac{(-1)^mm}{2\cdot n!}\sum_{l=m}^{n}\binom{n}{l-m}\binom{n}{l}\biggl(\lambda-\frac{m}{2}-n+l\biggr)_n. 
\end{align*}
Similarly, we have 
\[ \lim_{\varepsilon\to-\frac{m}{2}}\biggl(\varepsilon+\frac{m}{2}\biggr)A_n^\flat(\lambda,\varepsilon)
=-\frac{(-1)^mm}{2\cdot n!}\sum_{l=m}^{n}\binom{n}{l}\binom{n}{l-m}\biggl(\lambda-\frac{m}{2}-n+l\biggr)_n. \]
Hence we get the latter equality of $A_n^\flat(\lambda,\varepsilon)$. Next, for $B_n^\flat(\lambda,\varepsilon)$, 
since the degrees of numerators are smaller than those of denominators, it suffices to compute the residues at all poles. 
For $m=1,2,\ldots,n$, the residue at the pole $\varepsilon=\frac{m}{2}$ is computed as 
\begin{align*}
&\lim_{\varepsilon\to\frac{m}{2}}\biggl(\varepsilon-\frac{m}{2}\biggr)B_n^\flat(\lambda,\varepsilon) \\
&=\lim_{\varepsilon\to\frac{m}{2}}\biggl(\varepsilon\hspace{-1pt}-\hspace{-1pt}\frac{m}{2}\biggr)\sum_{l=1}^n\sum_{k=0}^{l-1}\binom{n}{l}
\biggl(\frac{(\lambda\hspace{-1pt}+\hspace{-1pt}\varepsilon\hspace{-1pt}-\hspace{-1pt}n\hspace{-1pt}+\hspace{-1pt}l)_n}{(2\varepsilon-n+l)_{n+1}}
\frac{(-1)^{n-l}}{\lambda\hspace{-1pt}+\hspace{-1pt}\varepsilon\hspace{-1pt}+\hspace{-1pt}k}
+\frac{(\lambda-\varepsilon-n+l)_n}{(-2\varepsilon\hspace{-1pt}-\hspace{-1pt}n\hspace{-1pt}+\hspace{-1pt}l)_{n+1}}
\frac{(-1)^{n-l}}{\lambda\hspace{-1pt}-\hspace{-1pt}\varepsilon\hspace{-1pt}+\hspace{-1pt}k}\biggr) \\
&=\frac{1}{2}\lim_{\varepsilon\to\frac{m}{2}}\biggl(\sum_{l=1}^{n-m}
\sum_{k=0}^{l-1}\binom{n}{l}\frac{(\lambda+\varepsilon-n+l)_n}{(2\varepsilon-n+l)_{n-l-m}(2\varepsilon-m+1)_{l+m}}\frac{(-1)^{n-l}}{\lambda+\varepsilon+k} \\*
&\hspace{100pt}-\sum_{l=m}^n\sum_{k=0}^{l-1}\binom{n}{l}\frac{(\lambda-\varepsilon-n+l)_n}{(-2\varepsilon-n+l)_{n-l+m}(-2\varepsilon+m+1)_{l-m}}\frac{(-1)^{n-l}}{\lambda-\varepsilon+k}\biggr) \\
&=\frac{1}{2}\biggl(\sum_{l=1}^{n-m}\sum_{k=0}^{l-1}\binom{n}{l}\frac{\bigl(\lambda+\frac{m}{2}-n+l\bigr)_n}{(m-n+l)_{n-l-m}(1)_{l+m}}\frac{(-1)^{n-l}}{\lambda+\frac{m}{2}+k} \\*
&\hspace{152pt}-\sum_{l=m}^n\sum_{k=0}^{l-1}\binom{n}{l}\frac{\bigl(\lambda-\frac{m}{2}-n+l\bigr)_n}{(-m-n+l)_{n-l+m}(1)_{l-m}}\frac{(-1)^{n-l}}{\lambda-\frac{m}{2}+k}\biggr) \\
&=\frac{(-1)^m}{2\cdot n!}\biggl(\sum_{l=1}^{n-m}\sum_{k=0}^{l-1}\binom{n}{l}\binom{n}{l+m}\frac{\bigl(\lambda+\frac{m}{2}-n+l\bigr)_n}{\lambda+\frac{m}{2}+k} \\*
&\hspace{202pt}-\sum_{l=m}^n\sum_{k=0}^{l-1}\binom{n}{l}\binom{n}{l-m}\frac{\bigl(\lambda-\frac{m}{2}-n+l\bigr)_n}{\lambda-\frac{m}{2}+k}\biggr) \\
&=\frac{(-1)^m}{2\cdot n!}\biggl(\sum_{l=m+1}^{n}\sum_{k=0}^{l-m-1}\binom{n}{l-m}\binom{n}{l}\frac{\bigl(\lambda-\frac{m}{2}-n+l\bigr)_n}{\lambda+\frac{m}{2}+k} \\*
&\hspace{202pt}-\sum_{l=m}^n\sum_{k=0}^{l-1}\binom{n}{l}\binom{n}{l-m}\frac{\bigl(\lambda-\frac{m}{2}-n+l\bigr)_n}{\lambda-\frac{m}{2}+k}\biggr) \\
&=\frac{(-1)^m}{2\cdot n!}\biggl(\sum_{l=m+1}^{n}\sum_{k=m}^{l-1}\binom{n}{l-m}\binom{n}{l}\frac{\bigl(\lambda-\frac{m}{2}-n+l\bigr)_n}{\lambda-\frac{m}{2}+k} \\*
&\hspace{202pt}-\sum_{l=m}^n\sum_{k=0}^{l-1}\binom{n}{l}\binom{n}{l-m}\frac{\bigl(\lambda-\frac{m}{2}-n+l\bigr)_n}{\lambda-\frac{m}{2}+k}\biggr) \\
&=-\frac{(-1)^m}{2\cdot n!}\sum_{l=m}^n\sum_{k=0}^{m-1}\binom{n}{l}\binom{n}{l-m}\frac{\bigl(\lambda-\frac{m}{2}-n+l\bigr)_n}{\lambda-\frac{m}{2}+k}. 
\end{align*}
Also, since $B_n^\flat(\lambda,\varepsilon)$ is even with respect to $\varepsilon$, we have 
\begin{align*}
&\lim_{\varepsilon\to-\frac{m}{2}}\biggl(\varepsilon+\frac{m}{2}\biggr)B_n^\flat(\lambda,\varepsilon)
=\frac{(-1)^m}{2\cdot n!}\sum_{l=m}^n\sum_{k=0}^{m-1}\binom{n}{l}\binom{n}{l-m}\frac{\bigl(\lambda-\frac{m}{2}-n+l\bigr)_n}{\lambda-\frac{m}{2}+k}. 
\end{align*}
Hence we get the latter equality of $B_n^\flat(\lambda,\varepsilon)$. 
\end{proof}

\begin{remark}
Although $A_n^\flat(\lambda,\varepsilon)$ and $B_n^\flat(\lambda,\varepsilon)$ have poles at $\varepsilon\in\bigl\{\pm\frac{1}{2},\pm 1,\ldots,\pm\frac{n}{2}\bigr\}$, since 
\begin{align*}
&\sum_{k=0}^\infty\frac{1}{(\lambda+\varepsilon+k)(\lambda-\varepsilon+k)}
=\frac{1}{2\varepsilon}\sum_{k=0}^\infty\biggl(\frac{1}{\lambda-\varepsilon+k}-\frac{1}{\lambda+\varepsilon+k}\biggr) \\
&=\frac{1}{2\varepsilon}\biggl(\sum_{k=0}^{m-1}\frac{1}{\lambda-\varepsilon+k}+\sum_{k=0}^\infty \biggl(\frac{1}{\lambda-\varepsilon+m+k}-\frac{1}{\lambda+\varepsilon+k}\biggr)\biggr) \\
&=\frac{1}{2\varepsilon}\sum_{k=0}^{m-1}\frac{1}{\lambda-\varepsilon+k}+\frac{\varepsilon-\frac{m}{2}}{\varepsilon}\sum_{k=0}^\infty \frac{1}{(\lambda+\varepsilon+k)(\lambda-\varepsilon+m+k)} 
\end{align*}
holds for $m=1,2,\ldots,n$, we have 
\begin{align*}
&\lim_{\varepsilon\to\frac{m}{2}}\biggl(\varepsilon-\frac{m}{2}\biggr)A_n^\flat(\lambda,\varepsilon)\sum_{k=0}^\infty\frac{1}{(\lambda+\varepsilon+k)(\lambda-\varepsilon+k)}
=-\lim_{\varepsilon\to\frac{m}{2}}\biggl(\varepsilon-\frac{m}{2}\biggr)B_n^\flat(\lambda,\varepsilon) \\
&=\frac{(-1)^m}{2\cdot n!}\sum_{l=m}^n\binom{n}{l}\binom{n}{l-m}\biggl(\lambda-\frac{m}{2}-n+l\biggr)_n\sum_{k=0}^{m-1}\frac{1}{\lambda-\frac{m}{2}+k}. 
\end{align*}
Hence the residue of $J_n^\flat(\lambda,\varepsilon)$ at $\varepsilon=\frac{m}{2}$ vanishes. 
Since $J_n^\flat(\lambda,-\varepsilon)=J_n^\flat(\lambda,\varepsilon)$ holds, the residue at $\varepsilon\hspace{-1pt}=\hspace{-1pt}-\frac{m}{2}$ also vanishes, 
namely, $J_n^\flat(\lambda,\varepsilon)$ is holomorphic at $\varepsilon\hspace{-1pt}\in\hspace{-1pt}\bigl\{\pm\frac{1}{2},\pm 1,\ldots,\pm\frac{n}{2}\bigr\}$. 
\end{remark}

\begin{remark}
By (\ref{formula_Apery_Beukers}), $A_n^\flat(\lambda,\varepsilon),B_n^\flat(\lambda,\varepsilon)$ and Ap\'ery's numbers $A_n,B_n$ are related as 
\begin{equation}\label{formula_rem_Apery}
A_n^\flat(n+1,0)=A_n, \qquad B_n^\flat(n+1,0)=A_n\sum_{k=1}^n\frac{1}{k^2}-B_n. 
\end{equation}
Indeed, the former formula of (\ref{formula_rem_Apery}) is directly verified as 
\[ A_n^\flat(n+1,0)=\sum_{l=0}^n\binom{n}{l}\frac{(l+1)_n}{(1)_l(1)_{n-l}}=\sum_{l=0}^n\binom{n}{l}^2\frac{(l+1)_n}{n!}=\sum_{l=0}^n\binom{n}{l}^2\binom{n+l}{l}=A_n. \]
On the other hand, the latter formula of (\ref{formula_rem_Apery}) is non-trivial. 
\end{remark}

From Theorem \ref{thm_Jflat}, we immediately get the following. 
\begin{corollary}
For $g\in\BR$, $\lambda,\varepsilon\in\BC$ with $\lambda\pm\varepsilon\notin-\BZ_{\ge 0}$, $\varepsilon\notin\pm\frac{1}{2}\BZ_{\ge 1}$, we have 
\[ R_1^\flat(\lambda;g,\varepsilon)=\sum_{n=0}^\infty \frac{(4g^2)^n}{n!}
\biggl(A_n^\flat(\lambda,\varepsilon)\sum_{k=0}^\infty \frac{1}{(\lambda+\varepsilon+k)(\lambda-\varepsilon+k)}+B_n^\flat(\lambda,\varepsilon)\biggr) \]
where $A_n^\flat(\lambda,\varepsilon),B_n^\flat(\lambda,\varepsilon)$ are as in Theorem \ref{thm_Jflat}. 
\end{corollary}

\section{Computation of the first term of 2pQRM}\label{section_2pQRM}

In this section, we seek another expression of 
\begin{align*}
&R_1^+(\lambda;g,\varepsilon):=R_1^{1/2}(\lambda;g,\varepsilon)+R_1^{3/2}(\lambda;g,\varepsilon) \\*
&=\Tr\biggl(\biggl(\cosh(2g)\biggl(\ra^\dagger\ra+\frac{1}{2}\biggr)+\frac{\sinh(2g)}{2}(\ra^2+(\ra^\dagger)^2)+\varepsilon+\lambda\biggr)^{-1} \\*
&\hspace{50pt}\times\biggl(\cosh(2g)\biggl(\ra^\dagger\ra+\frac{1}{2}\biggr)-\frac{\sinh(2g)}{2}(\ra^2+(\ra^\dagger)^2)-\varepsilon+\lambda\biggr)^{-1}\biggr), \\
&R_1^-(\lambda;g,\varepsilon):=R_1^{1/2}(\lambda;g,\varepsilon)-R_1^{3/2}(\lambda;g,\varepsilon) \\*
&=\Tr\biggl(\biggl(\cosh(2g)\biggl(\ra^\dagger\ra+\frac{1}{2}\biggr)+\frac{\sinh(2g)}{2}(\ra^2+(\ra^\dagger)^2)+\varepsilon+\lambda\biggr)^{-1} \\*
&\hspace{50pt}\times\biggl(\cosh(2g)\biggl(\ra^\dagger\ra+\frac{1}{2}\biggr)-\frac{\sinh(2g)}{2}(\ra^2+(\ra^\dagger)^2)-\varepsilon+\lambda\biggr)^{-1}\biggr|_{L^2(\BR)_\even}\biggr) \\*
&\eqspace{}-\Tr\biggl(\biggl(\cosh(2g)\biggl(\ra^\dagger\ra+\frac{1}{2}\biggr)+\frac{\sinh(2g)}{2}(\ra^2+(\ra^\dagger)^2)+\varepsilon+\lambda\biggr)^{-1} \\*
&\hspace{50pt}\times\biggl(\cosh(2g)\biggl(\ra^\dagger\ra+\frac{1}{2}\biggr)-\frac{\sinh(2g)}{2}(\ra^2+(\ra^\dagger)^2)-\varepsilon+\lambda\biggr)^{-1}\biggr|_{L^2(\BR)_\odd}\biggr). 
\end{align*}
Then by Example \ref{ex_R1}, for $\delta\in\{\pm\}$, $g\in\BR$, $\lambda,\varepsilon\in\BC$ with $\Re\lambda-|\Re\varepsilon|+\frac{1}{2}>0$, we have 
\begin{align*}
R_1^\delta(\lambda;g,\varepsilon)&=\sech(2g)\int_{[0,1]^2} \frac{u^{\lambda+\varepsilon-\frac{1}{2}}v^{\lambda-\varepsilon-\frac{1}{2}}}{\sqrt{(1-\delta uv)^2-\tanh^2(2g)(u-\delta v)^2}}\,\rd u\rd v \\
&=\int_{[0,1]^2} \frac{u^{\lambda+\varepsilon-\frac{1}{2}}v^{\lambda-\varepsilon-\frac{1}{2}}}{\sqrt{(1-\delta uv)^2+\sinh^2(2g)(1-u^2)(1-v^2)}}\,\rd u\rd v. 
\end{align*}
Since $\bigl|\frac{u-\delta v}{1-\delta uv}\bigr|<1$ holds for $u,v\in (0,1)$, we have 
\begin{align*}
R_1^\delta(\lambda;g,\varepsilon)
&=\sech(2g)\int_{[0,1]^2} \frac{u^{\lambda+\varepsilon-\frac{1}{2}}v^{\lambda-\varepsilon-\frac{1}{2}}}{1-\delta uv}
\biggl(1-\tanh^2(2g)\biggl(\frac{u-\delta v}{1-\delta uv}\biggr)^2\biggr)^{-1/2}\,\rd u\rd v \\
&=\sech(2g)\sum_{n=0}^\infty \frac{\bigl(\frac{1}{2}\bigr)_n}{n!}\tanh^{2n}(2g)\int_{[0,1]^2} \frac{(u-\delta v)^{2n}}{(1-\delta uv)^{2n+1}}
u^{\lambda+\varepsilon-\frac{1}{2}}v^{\lambda-\varepsilon-\frac{1}{2}}\,\rd u\rd v. 
\end{align*}
Now, for $\delta\in\{\pm\}$, $n\in\BZ_{\ge 0}$, $\lambda,\varepsilon\in\BC$ with $\Re\lambda-|\Re\varepsilon|+\frac{1}{2}>0$, we put 
\[ J_n^\delta(\lambda,\varepsilon):=\int_{[0,1]^2} \frac{(u-\delta v)^n}{(1-\delta uv)^{n+1}}u^{\lambda+\varepsilon-\frac{1}{2}}v^{\lambda-\varepsilon-\frac{1}{2}}\,\rd u\rd v, \]
so that 
\[ R_1^\delta(\lambda;g,\varepsilon)=\sech(2g)\sum_{n=0}^\infty \frac{\bigl(\frac{1}{2}\bigr)_n}{n!}\tanh^{2n}(2g)J_{2n}^\delta(\lambda,\varepsilon) \]
holds. 

\begin{remark}
\begin{enumerate}
\item $\{J_{n}^\delta(\lambda,\varepsilon)\}_n$ has other generating functions 
\begin{align}
U_\delta(t;\lambda,\varepsilon):\hspace{-3pt}&=\int_{[0,1]^2} \frac{u^{\lambda+\varepsilon-\frac{1}{2}}v^{\lambda-\varepsilon-\frac{1}{2}}}{1-\delta uv-t(u-\delta v)}\,\rd u\rd v 
=\sum_{n=0}^\infty J_n^\delta(\lambda,\varepsilon)t^n, \notag\\
V_\delta(t;\lambda,\varepsilon):\hspace{-3pt}&=\frac{1}{2}\bigl(U_\delta(\sqrt{t};\lambda,\varepsilon)+U_\delta(-\sqrt{t};\lambda,\varepsilon)\bigr) \notag\\
&=\int_{[0,1]^2} \frac{(1-\delta uv)u^{\lambda+\varepsilon-\frac{1}{2}}v^{\lambda-\varepsilon-\frac{1}{2}}}{(1-\delta uv)^2-t(u-\delta v)^2}\,\rd u\rd v
=\sum_{n=0}^\infty J_{2n}^\delta(\lambda,\varepsilon)t^n. \label{def_Vdelta}
\end{align}
We have $U_+(t;\lambda,0)=U_+(-t;\lambda,0)=V_+(t^2;\lambda,0)$, since $J_{2n+1}^+(\lambda,0)=0$ holds. 
\item In \cite{KW4}, ``Ap\'ery-like numbers'' $J_k(n)$ $(k\in\BZ_{\ge 2}, n\in\BZ_{\ge 0})$ are introduced by using the generating function 
\[ R_{k,1}(\kappa)=\frac{k}{2}\sum_{n=0}^\infty \frac{\bigl(\frac{1}{2}\bigr)_n}{n!}J_k(n)(-\kappa^2)^n. \]
Comparing this with (\ref{formula_Taylor_R}) in Remark \ref{rem_Taylor_R}, we have 
\begin{align*}
R_1^+(\lambda;g,0)&=\sum_{k=0}^\infty \frac{2(-1)^k}{k+2}R_{k+2,1}(\sinh(2g))\lambda^k \\
&=\sum_{k=0}^\infty\sum_{n=0}^\infty (-1)^{k+n}\frac{\bigl(\frac{1}{2}\bigr)_n}{n!}J_{k+2}(n)\sinh^{2n}(2g)\lambda^k, 
\end{align*}
namely, we have 
\[ \int_{[0,1]^2} \frac{(1-u^2)^n(1-v^2)^n}{(1-uv)^{2n+1}}u^{\lambda-\frac{1}{2}}v^{\lambda-\frac{1}{2}}\,\rd u\rd v=\sum_{k=0}^\infty (-1)^kJ_{k+2}(n)\lambda^k. \]
This has another generating function 
\begin{align*}
\widetilde{V}_+(s;\lambda,0):=\frac{1}{1-s}V_+\biggl(\frac{s}{s-1};\lambda,0\biggr)
&=\int_{[0,1]^2} \frac{(1-uv)u^{\lambda-\frac{1}{2}}v^{\lambda-\frac{1}{2}}}{(1-uv)^2-s(1-u^2)(1-v^2)}\,\rd u\rd v \\
&=\sum_{n=0}^\infty\sum_{k=0}^\infty (-1)^kJ_{k+2}(n)\lambda^ks^n. 
\end{align*}
\end{enumerate}
\end{remark}

First we prove the following. 
\begin{proposition}\label{prop_Jpm}
For $n\in\BZ_{\ge 0}$, $\delta\in\{\pm\}$, $\lambda,\varepsilon\in\BC$ with $\lambda\pm\varepsilon\notin-\bigl(\BZ_{\ge 0}+\frac{1}{2}\bigr)$, we have 
\[ J_n^\delta(\lambda,\varepsilon)
=\frac{1}{2}\sum_{k=0}^\infty\frac{\bigl(\lambda-\frac{n-1}{2}\bigr)_n}{\bigl(\lambda+k-\frac{n-1}{2}\bigr)_{n+1}}
\biggl(\frac{1}{\lambda-\varepsilon+k+\frac{1}{2}}+\frac{(-\delta)^n}{\lambda+\varepsilon+k+\frac{1}{2}}\biggr)\delta^k. \]
\end{proposition}

\begin{proof}
By separating the domain of integration into $\{0\le v\le u\le 1\}$ and $\{0\le u\le v\le 1\}$, exchanging $u$ and $v$ on the 2nd domain, and putting $v=uw$, we have 
\begin{align*}
&J_n^\delta(\lambda,\varepsilon):=\int_{[0,1]^2}\frac{(u-\delta v)^n}{(1-\delta uv)^{n+1}}u^{\lambda+\varepsilon-\frac{1}{2}}v^{\lambda-\varepsilon-\frac{1}{2}}\,\rd u\rd v \\
&=\int_{0\le v\le u\le 1}\frac{(u-\delta v)^n}{(1-\delta uv)^{n+1}}(u^{\lambda+\varepsilon-\frac{1}{2}}v^{\lambda-\varepsilon-\frac{1}{2}}
+(-\delta)^n v^{\lambda+\varepsilon-\frac{1}{2}}u^{\lambda-\varepsilon-\frac{1}{2}})\,\rd u\rd v \\
&=\int_{[0,1]^2}\frac{(u-\delta uw)^n}{(1-\delta u^2w)^{n+1}}(u^{\lambda+\varepsilon-\frac{1}{2}}(uw)^{\lambda-\varepsilon-\frac{1}{2}}
+(-\delta)^n(uw)^{\lambda+\varepsilon-\frac{1}{2}}u^{\lambda-\varepsilon-\frac{1}{2}})u\,\rd u\rd w \\
&=\int_0^1\biggl(\int_0^1 \frac{(1-\delta w)^n}{(1-\delta u^2w)^{n+1}}u^{2\lambda+n}\,\rd u\biggr)(w^{\lambda-\varepsilon-\frac{1}{2}}+(-\delta)^nw^{\lambda+\varepsilon-\frac{1}{2}})\,\rd w. 
\end{align*}
Then the inner integral is computed as 
\begin{align*}
&\int_0^1 \frac{(1-\delta w)^n}{(1-\delta u^2w)^{n+1}}u^{2\lambda+n}\,\rd u \\
&=\int_0^1 \sum_{k=0}^\infty \frac{(n+1)_k}{k!}(\delta u^2w)^k \sum_{l=0}^n\binom{n}{l}(-\delta w)^{n-l}u^{2\lambda+n}\,\rd u \\
&=\int_0^1 \sum_{l=0}^n\sum_{k=0}^\infty \frac{(k+1)_n}{n!}\binom{n}{l} (-1)^{n-l}(\delta w)^{k+n-l}u^{2\lambda+2k+n}\,\rd u \\
&=\int_0^1 \sum_{l=0}^n\sum_{k=n-l}^\infty \frac{(k+l-n+1)_n}{n!}\binom{n}{l} (-1)^{n+l}(\delta w)^k u^{2\lambda+2k+2l-n}\,\rd u \\
&=\sum_{k=0}^\infty\sum_{l=\max\{n-k,0\}}^n \frac{(k+l-n+1)_n}{n!}\binom{n}{l} \frac{(-1)^{n+l}(\delta w)^k}{2\lambda+2k+2l-n+1} \\
&=\frac{(-1)^n}{2\cdot n!}\sum_{k=0}^\infty\sum_{l=0}^n \binom{n}{l} (-1)^l\frac{(k-n+1+l)_n}{\lambda+k-\frac{n-1}{2}+l}(\delta w)^k, 
\end{align*}
where we have used $\frac{(n+1)_k}{k!}=\frac{(n+k)!}{k!n!}=\frac{(k+1)_n}{n!}$ at the 2nd equality, and replaced $k$ with $k+l-n$ at the 3rd equality. We note that $0\le l\le n-k-1$ implies $(k+l-n+1)_n=0$, 
and hence we can extend the range of $l$ from $\max\{n-k,0\}\le l\le n$ to $0\le l\le n$ at the last equality. 
Next, for general $\alpha,\beta\in\BC$, we consider the following equality, 
\begin{align*}
&\sum_{l=0}^n \binom{n}{l} (-1)^l\frac{(\alpha+l)_n}{\beta+l}
=\frac{(\alpha)_n}{\beta}\sum_{l=0}^n \frac{(-n)_l(\alpha+n)_l(\beta)_l}{l!(\alpha)_l(\beta+1)_l}
=\frac{(\alpha)_n}{\beta}{}_3F_2\biggl(\begin{matrix} -n,\alpha+n,\beta \\ \alpha,\beta+1 \end{matrix};1\biggr) \\
&=\frac{(\alpha)_n}{\beta}\frac{(\beta-\alpha-n+1)_n(1)_n}{(\beta+1)_n(1-\alpha-n)_n}
=(-1)^n\frac{(\beta-\alpha-n+1)_n n!}{(\beta)_{n+1}}, 
\end{align*}
which follows from the Pfaff--Saalsch\"utz formula \cite[Theorem 2.2.6]{AAR}, 
\[ {}_3F_2\biggl(\begin{matrix} -n,a,b \\ 1+a+b-c-n,c \end{matrix};1\biggr)=\frac{(c-a)_n(c-b)_n}{(c)_n(c-a-b)_n}. \] 
Applying this formula with $\alpha=k-n+1$, $\beta=\lambda+k-\frac{n-1}{2}$, we get 
\begin{align*}
\int_0^1 \frac{(1-\delta w)^n}{(1-\delta u^2w)^{n+1}}u^{2\lambda+n}\,\rd u 
=\frac{1}{2}\sum_{k=0}^\infty\frac{\bigl(\lambda-\frac{n-1}{2}\bigr)_n}{\bigl(\lambda+k-\frac{n-1}{2}\bigr)_{n+1}}(\delta w)^k. 
\end{align*}
Therefore we get 
\begin{align*}
&J_n^\delta(\lambda,\varepsilon)=\int_0^1\biggl(\int_0^1 \frac{(1-\delta w)^n}{(1-\delta u^2w)^{n+1}}u^{2\lambda+n}\,\rd u\biggr)
(w^{\lambda-\varepsilon-\frac{1}{2}}+(-\delta)^n w^{\lambda+\varepsilon-\frac{1}{2}})\,\rd w \\
&=\frac{1}{2}\int_0^1\sum_{k=0}^\infty\frac{\bigl(\lambda-\frac{n-1}{2}\bigr)_n}{\bigl(\lambda+k-\frac{n-1}{2}\bigr)_{n+1}}(\delta w)^k
(w^{\lambda-\varepsilon-\frac{1}{2}}+(-\delta)^n w^{\lambda+\varepsilon-\frac{1}{2}})\,\rd w \\
&=\frac{1}{2}\sum_{k=0}^\infty\frac{\bigl(\lambda-\frac{n-1}{2}\bigr)_n}{\bigl(\lambda+k-\frac{n-1}{2}\bigr)_{n+1}}
\biggl(\frac{1}{\lambda-\varepsilon+k+\frac{1}{2}}+\frac{(-\delta)^n}{\lambda+\varepsilon+k+\frac{1}{2}}\biggr)\delta^k. \qedhere 
\end{align*}
\end{proof}

$J_n^\delta(\lambda,\varepsilon)$ satisfies the following recurrence relation. 
\begin{proposition}\label{prop_recurrence_Jpm}
For $n\in\BZ_{\ge 2}$, $\delta\in\{\pm\}$, $\lambda,\varepsilon\in\BC$ with $\lambda\pm\varepsilon\notin-\bigl(\BZ_{\ge 0}+\frac{1}{2}\bigr)$, $\varepsilon\ne\pm\frac{n}{2}$, we have 
\begin{align*}
J_n^\delta(\lambda,\varepsilon)
=\frac{\bigl(\lambda+\frac{n-1}{2}\bigr)\bigl(\lambda-\frac{n-1}{2}\bigr)}{\bigl(\varepsilon+\frac{n}{2}\bigr)\bigl(\varepsilon-\frac{n}{2}\bigr)}J_{n-2}^\delta(\lambda,\varepsilon)
-\begin{cases} \displaystyle \frac{\lambda}{(n-1)\bigl(\varepsilon+\frac{n}{2}\bigr)\bigl(\varepsilon-\frac{n}{2}\bigr)} & (\delta=+,\; n\colon\text{even}), \\
\displaystyle \frac{\varepsilon}{n\bigl(\varepsilon+\frac{n}{2}\bigr)\bigl(\varepsilon-\frac{n}{2}\bigr)} & (\delta=+,\; n\colon\text{odd}), \\
\displaystyle \frac{1}{2\bigl(\varepsilon+\frac{n}{2}\bigr)\bigl(\varepsilon-\frac{n}{2}\bigr)} & (\delta=-). \end{cases}
\end{align*}
\end{proposition}

\begin{proof}
By Proposition \ref{prop_Jpm}, we have 
\begin{align*}
&J_n^\delta(\lambda,\varepsilon)-\frac{\bigl(\lambda+\frac{n-1}{2}\bigr)\bigl(\lambda-\frac{n-1}{2}\bigr)}{\bigl(\varepsilon+\frac{n}{2}\bigr)\bigl(\varepsilon-\frac{n}{2}\bigr)}
J_{n-2}^\delta(\lambda,\varepsilon) \\
&=\sum_{k=0}^\infty\Biggl(\frac{\bigl(\lambda-\frac{n-1}{2}\bigr)_n}{\bigl(\lambda+k-\frac{n-1}{2}\bigr)_{n+1}}
-\frac{\bigl(\lambda+\frac{n-1}{2}\bigr)\bigl(\lambda-\frac{n-1}{2}\bigr)}{\bigl(\varepsilon+\frac{n}{2}\bigr)\bigl(\varepsilon-\frac{n}{2}\bigr)}
\frac{\bigl(\lambda-\frac{n-3}{2}\bigr)_{n-2}}{\bigl(\lambda+k-\frac{n-3}{2}\bigr)_{n-1}}\Biggr) \\*
&\hspace{228pt}\times\frac{1}{2}\biggl(\frac{1}{\lambda-\varepsilon+k+\frac{1}{2}}+\frac{(-\delta)^n}{\lambda+\varepsilon+k+\frac{1}{2}}\biggr)\delta^k \\
&=\sum_{k=0}^\infty \frac{\bigl(\lambda-\frac{n-1}{2}\bigr)_n
\bigl(\bigl(\varepsilon+\frac{n}{2}\bigr)\bigl(\varepsilon-\frac{n}{2}\bigr)-\bigl(\lambda+k-\frac{n-1}{2}\bigr)\bigl(\lambda+k+\frac{n+1}{2}\bigr)\bigr)}
{\bigl(\lambda+k-\frac{n-1}{2}\bigr)_{n+1}\bigl(\varepsilon+\frac{n}{2}\bigr)\bigl(\varepsilon-\frac{n}{2}\bigr)} \\*
&\hspace{228pt}\times\frac{1}{2}\biggl(\frac{1}{\lambda-\varepsilon+k+\frac{1}{2}}+\frac{(-\delta)^n}{\lambda+\varepsilon+k+\frac{1}{2}}\biggr)\delta^k \\
&=-\sum_{k=0}^\infty \frac{\bigl(\lambda\hspace{-1pt}-\hspace{-1pt}\frac{n-1}{2}\bigr)_n
\bigl(\lambda\hspace{-1pt}+\hspace{-1pt}\varepsilon\hspace{-1pt}+\hspace{-1pt}k\hspace{-1pt}+\hspace{-1pt}\frac{1}{2}\bigr)
\bigl(\lambda\hspace{-1pt}-\hspace{-1pt}\varepsilon\hspace{-1pt}+\hspace{-1pt}k\hspace{-1pt}+\hspace{-1pt}\frac{1}{2}\bigr)}
{\bigl(\lambda+k-\frac{n-1}{2}\bigr)_{n+1}\bigl(\varepsilon+\frac{n}{2}\bigr)\bigl(\varepsilon-\frac{n}{2}\bigr)}
\frac{\bigl(\lambda\hspace{-1pt}+\hspace{-1pt}\varepsilon\hspace{-1pt}+\hspace{-1pt}k\hspace{-1pt}+\hspace{-1pt}\frac{1}{2}\bigr)
\hspace{-1pt}+\hspace{-1pt}(-\delta)^n\bigl(\lambda\hspace{-1pt}-\hspace{-1pt}\varepsilon\hspace{-1pt}+\hspace{-1pt}k\hspace{-1pt}+\hspace{-1pt}\frac{1}{2}\bigr)}
{2\bigl(\lambda+\varepsilon+k+\frac{1}{2}\bigr)\bigl(\lambda-\varepsilon+k+\frac{1}{2}\bigr)}\delta^k \\
&=-\frac{\bigl(\lambda-\frac{n-1}{2}\bigr)_n}{\bigl(\varepsilon+\frac{n}{2}\bigr)\bigl(\varepsilon-\frac{n}{2}\bigr)}\sum_{k=0}^\infty \frac{\delta^k}{\bigl(\lambda+k-\frac{n-1}{2}\bigr)_{n+1}} \\*
&\eqspace{}\times \begin{cases} \displaystyle \frac{\bigl(\lambda+k+\frac{n+1}{2}\bigr)\bigl(\lambda+k+\frac{n-1}{2}\bigr)
-\bigl(\lambda+k-\frac{n-1}{2}\bigr)\bigl(\lambda+k-\frac{n-3}{2}\bigr)}{2(n-1)} & ((-\delta)^n=+1) \\ 
\displaystyle \frac{\varepsilon\bigl(\bigl(\lambda+k+\frac{n+1}{2}\bigr)-\bigl(\lambda+k-\frac{n-1}{2}\bigr)\bigr)}{n} & ((-\delta)^n=-1) \end{cases} \\
&=-\frac{\bigl(\lambda-\frac{n-1}{2}\bigr)_n}{\bigl(\varepsilon+\frac{n}{2}\bigr)\bigl(\varepsilon-\frac{n}{2}\bigr)} \\*
&\eqspace{}\times \begin{cases} \displaystyle 
\frac{1}{2(n-1)}\sum_{k=0}^\infty \Biggl(\frac{\delta^k}{\bigl(\lambda+k-\frac{n-1}{2}\bigr)_{n-1}}-\frac{\delta^k}{\bigl(\lambda+k-\frac{n-5}{2}\bigr)_{n-1}}\Biggr) & ((-\delta)^n=+1) \\ 
\displaystyle \frac{\varepsilon}{n}\sum_{k=0}^\infty \Biggl(\frac{1}{\bigl(\lambda+k-\frac{n-1}{2}\bigr)_n}-\frac{1}{\bigl(\lambda+k-\frac{n-3}{2}\bigr)_n}\Biggr) 
& ((-\delta)^n=-1) \end{cases} \\
&=\begin{cases} \displaystyle \frac{-\bigl(\lambda-\frac{n-1}{2}\bigr)_n}{2(n-1)\bigl(\varepsilon+\frac{n}{2}\bigr)\bigl(\varepsilon-\frac{n}{2}\bigr)}
\Biggl(\frac{1}{\bigl(\lambda-\frac{n-1}{2}\bigr)_{n-1}}+\frac{\delta}{\bigl(\lambda-\frac{n-3}{2}\bigr)_{n-1}}\Biggr) & ((-\delta)^n=+1) \\
\displaystyle \frac{-\varepsilon\bigl(\lambda-\frac{n-1}{2}\bigr)_n}{n\bigl(\varepsilon+\frac{n}{2}\bigr)\bigl(\varepsilon-\frac{n}{2}\bigr)}
\frac{1}{\bigl(\lambda-\frac{n-1}{2}\bigr)_n} & ((-\delta)^n=+1) \end{cases} \\
&=\begin{cases} \displaystyle -\frac{\bigl(\lambda+\frac{n-1}{2}\bigr)+\delta\bigl(\lambda-\frac{n-1}{2}\bigr)}
{2(n-1)\bigl(\varepsilon+\frac{n}{2}\bigr)\bigl(\varepsilon-\frac{n}{2}\bigr)} & ((-\delta)^n=+1) \\
\displaystyle -\frac{\varepsilon}{n\bigl(\varepsilon+\frac{n}{2}\bigr)\bigl(\varepsilon+\frac{n}{2}\bigr)} & ((-\delta)^n=-1) \end{cases} \\
&=\begin{cases} \displaystyle -\frac{\lambda}{(n-1)\bigl(\varepsilon+\frac{n}{2}\bigr)\bigl(\varepsilon-\frac{n}{2}\bigr)} & (\delta=+,\; n\colon\text{even}) \\
\displaystyle -\frac{\varepsilon}{n\bigl(\varepsilon+\frac{n}{2}\bigr)\bigl(\varepsilon-\frac{n}{2}\bigr)} & (\delta=+,\; n\colon\text{odd}) \\
\displaystyle -\frac{1}{2\bigl(\varepsilon+\frac{n}{2}\bigr)\bigl(\varepsilon-\frac{n}{2}\bigr)} & (\delta=-). \end{cases} \qedhere
\end{align*}
\end{proof}

Now we prove the following. 
\begin{theorem}\label{thm_Jpm}
For $n\in\BZ_{\ge 0}$, $\delta\in\{\pm\}$, $\lambda,\varepsilon\in\BC$ with $\lambda\pm\varepsilon\notin-\bigl(\BZ_{\ge 0}+\frac{1}{2}\bigr)$, 
$\varepsilon\notin\bigl\{-\frac{n}{2},-\frac{n}{2}+1,\ldots,\frac{n}{2}\bigr\}$, we have 
\[ J_n^\delta(\lambda,\varepsilon)
=\frac{A_n^\delta(\lambda,\varepsilon)}{2}\sum_{k=0}^\infty\biggl(\frac{1}{\lambda-\varepsilon+k+\frac{1}{2}}
-\frac{\delta^n}{\lambda+\varepsilon+k+\frac{1}{2}}\biggr)\delta^k+B_n^\delta(\lambda,\varepsilon), \]
where 
\[ A_n^\delta(\lambda,\varepsilon):=\frac{\bigl(\lambda-\frac{n-1}{2}\bigr)_n}{\bigl(\varepsilon-\frac{n}{2}\bigr)_{n+1}}
=\begin{cases}\displaystyle \frac{\bigl(\frac{1}{2}+\lambda\bigr)_{n/2}\bigl(\frac{1}{2}-\lambda\bigr)_{n/2}}{\varepsilon(1+\varepsilon)_{n/2}(1-\varepsilon)_{n/2}} & (n\colon\text{even}), \\
\displaystyle -\frac{\lambda(1+\lambda)_{\lfloor n/2\rfloor}(1-\lambda)_{\lfloor n/2\rfloor}}{\bigl(\frac{1}{2}+\varepsilon\bigr)_{\lceil n/2\rceil}\bigl(\frac{1}{2}-\varepsilon\bigr)_{\lceil n/2\rceil}} 
& (n\colon\text{odd}), \end{cases} \]
$B_0^\delta(\lambda,\varepsilon)=0$, and for $n\ge 1$, 
\begin{align*}
&B_n^\delta(\lambda,\varepsilon) \\
:\hspace{-3pt}&=\frac{1}{2}
\sum_{m=0}^{\lceil n/2\rceil-1}\frac{(-1)^{m+1}}{m!(n-m)!}\biggl(\frac{1}{\varepsilon-\frac{n}{2}+m}-\frac{(-\delta)^n}{\varepsilon+\frac{n}{2}-m}\biggr)
\sum_{k=0}^{n-2m-1}\frac{\bigl(\lambda-\frac{n-1}{2}\bigr)_n}{\lambda+k-\frac{n-1}{2}+m}\delta^k \\
&=\begin{cases} \displaystyle \sum_{l=0}^{n/2-1}\frac{\lambda\bigl(\lambda-\frac{n-1}{2}\bigr)_l\bigl(-\lambda-\frac{n-1}{2}\bigr)_l}
{(n-2l-1)\bigl(\varepsilon-\frac{n}{2}\bigr)_{l+1}\bigl(-\varepsilon-\frac{n}{2}\bigr)_{l+1}} & (\delta=+,\; n\colon\text{even}), \\ 
\displaystyle \sum_{l=0}^{\lceil n/2\rceil-1}\frac{\varepsilon\bigl(\lambda-\frac{n-1}{2}\bigr)_l\bigl(-\lambda-\frac{n-1}{2}\bigr)_l}
{(n-2l)\bigl(\varepsilon-\frac{n}{2}\bigr)_{l+1}\bigl(-\varepsilon-\frac{n}{2}\bigr)_{l+1}} & (\delta=+,\; n\colon\text{odd}), \\ 
\displaystyle \sum_{l=0}^{\lceil n/2\rceil-1}\frac{\bigl(\lambda-\frac{n-1}{2}\bigr)_l\bigl(-\lambda-\frac{n-1}{2}\bigr)_l}
{2\bigl(\varepsilon-\frac{n}{2}\bigr)_{l+1}\bigl(-\varepsilon-\frac{n}{2}\bigr)_{l+1}} & (\delta=-). 
\end{cases}
\end{align*}
\end{theorem}

\begin{proof}
First, we find the partial fraction decomposition of each term of 
\[ J_n^\delta(\lambda,\varepsilon)
=\frac{1}{2}\sum_{k=0}^\infty\frac{\bigl(\lambda-\frac{n-1}{2}\bigr)_n}{\bigl(\lambda+k-\frac{n-1}{2}\bigr)_{n+1}}
\biggl(\frac{1}{\lambda-\varepsilon+k+\frac{1}{2}}+\frac{(-\delta)^n}{\lambda+\varepsilon+k+\frac{1}{2}}\biggr)\delta^k. \]
We put 
\[ f(x)=\frac{1}{\bigl(\lambda+x-\frac{n-1}{2}\bigr)_{n+1}}\frac{1}{\lambda-\varepsilon+x+\frac{1}{2}}. \]
Then the residue of $f(x)$ at the pole $x=-\lambda+\frac{n-1}{2}-m$ $(0\le m\le n)$ is given by 
\begin{align*}
&\lim_{x\to-\lambda+\frac{n-1}{2}-m}\biggl(\lambda+x-\frac{n-1}{2}+m\biggr)f(x) \\
&=\lim_{x\to-\lambda+\frac{n-1}{2}-m}\frac{1}{\bigl(\lambda+x-\frac{n-1}{2}\bigr)_m\bigl(\lambda+x-\frac{n-1}{2}+m+1\bigr)_{n-m}}\frac{1}{\lambda-\varepsilon+x+\frac{1}{2}} \\
&=\frac{1}{(-m)_m(1)_{n-m}}\frac{1}{-\varepsilon+\frac{n}{2}-m}=\frac{(-1)^{m+1}}{m!(n-m)!}\frac{1}{\varepsilon-\frac{n}{2}+m}, 
\end{align*}
and the residue at the pole $x=-\lambda+\varepsilon-\frac{1}{2}$ is given by 
\[ \lim_{x\to-\lambda+\varepsilon-\frac{1}{2}}\biggl(\lambda-\varepsilon+x+\frac{1}{2}\biggr)f(x)
=\frac{1}{\bigl(\varepsilon-\frac{n}{2}\bigr)_{n+1}}. \]
Hence we have 
\begin{align*}
&\frac{1}{\bigl(\lambda+x-\frac{n-1}{2}\bigr)_{n+1}}\frac{1}{\lambda-\varepsilon+x+\frac{1}{2}} \\
&=\sum_{m=0}^n\frac{(-1)^{m+1}}{m!(n-m)!}\frac{1}{\varepsilon-\frac{n}{2}+m}\frac{1}{\lambda+x-\frac{n-1}{2}+m}
+\frac{1}{\bigl(\varepsilon-\frac{n}{2}\bigr)_{n+1}}\frac{1}{\lambda-\varepsilon+x+\frac{1}{2}}, 
\end{align*}
and by changing $\varepsilon$ to $-\varepsilon$, $m$ to $n-m$, we have 
\begin{align*}
&\frac{1}{\bigl(\lambda+x-\frac{n-1}{2}\bigr)_{n+1}}\frac{1}{\lambda+\varepsilon+x+\frac{1}{2}} \\
&=\sum_{m=0}^n\frac{(-1)^{m}}{m!(n-m)!}\frac{1}{\varepsilon+\frac{n}{2}-m}\frac{1}{\lambda+x-\frac{n-1}{2}+m}
+\frac{1}{\bigl(-\varepsilon-\frac{n}{2}\bigr)_{n+1}}\frac{1}{\lambda+\varepsilon+x+\frac{1}{2}} \\
&=\sum_{m=0}^n\frac{(-1)^{n-m}}{(n-m)!m!}\frac{1}{\varepsilon-\frac{n}{2}+m}\frac{1}{\lambda+x+\frac{n+1}{2}-m}
+\frac{(-1)^{n+1}}{\bigl(\varepsilon-\frac{n}{2}\bigr)_{n+1}}\frac{1}{\lambda+\varepsilon+x+\frac{1}{2}}. 
\end{align*}
Therefore we have 
\begin{align*}
&\frac{2}{\bigl(\lambda-\frac{n-1}{2}\bigr)_n}J_n^\delta(\lambda,\varepsilon) \\
&=\sum_{k=0}^\infty\biggl(\sum_{m=0}^n\frac{(-1)^{m+1}}{m!(n-m)!}\frac{1}{\varepsilon-\frac{n}{2}+m}\frac{1}{\lambda+k-\frac{n-1}{2}+m}
+\frac{1}{\bigl(\varepsilon-\frac{n}{2}\bigr)_{n+1}}\frac{1}{\lambda-\varepsilon+k+\frac{1}{2}} \\*
&\eqspace{}+\sum_{m=0}^n\frac{(-1)^m\delta^n}{m!(n-m)!}\frac{1}{\varepsilon-\frac{n}{2}+m}\frac{1}{\lambda+k+\frac{n+1}{2}-m}
-\frac{\delta^n}{\bigl(\varepsilon-\frac{n}{2}\bigr)_{n+1}}\frac{1}{\lambda+\varepsilon+k+\frac{1}{2}}\biggr)\delta^k \\
&=\sum_{k=0}^\infty
\biggl(\frac{1}{\bigl(\varepsilon-\frac{n}{2}\bigr)_{n+1}}\biggl(\frac{1}{\lambda-\varepsilon+k+\frac{1}{2}}-\frac{\delta^n}{\lambda+\varepsilon+k+\frac{1}{2}}\biggr) \\*
&\eqspace{}+\sum_{m=0}^n\frac{(-1)^{m+1}}{m!(n-m)!}\frac{1}{\varepsilon-\frac{n}{2}+m}\biggl(\frac{1}{\lambda+k-\frac{n-1}{2}+m}-\frac{\delta^n}{\lambda+k+\frac{n+1}{2}-m}\biggr)\biggr)\delta^k. 
\end{align*}
Then for $n\ge 1$, the 2nd term is computed as the limit $N\to\infty$ of the following. 
\begin{align*}
&\sum_{k=0}^N\sum_{m=0}^n\frac{(-1)^{m+1}}{m!(n-m)!}\frac{1}{\varepsilon-\frac{n}{2}+m}\biggl(\frac{1}{\lambda+k-\frac{n-1}{2}+m}-\frac{\delta^n}{\lambda+k+\frac{n+1}{2}-m}\biggr)\delta^k \\
&=\sum_{k=0}^N\sum_{m=0}^{\lceil n/2\rceil-1}\frac{(-1)^{m+1}}{m!(n-m)!}
\frac{1}{\varepsilon-\frac{n}{2}+m}\biggl(\frac{1}{\lambda+k-\frac{n-1}{2}+m}-\frac{\delta^n}{\lambda+k+\frac{n+1}{2}-m}\biggr)\delta^k \\*
&\eqspace{}+\sum_{k=0}^N\sum_{m=0}^{\lceil n/2\rceil-1}\frac{(-1)^{n-m+1}}{m!(n-m)!}
\frac{1}{\varepsilon+\frac{n}{2}-m}\biggl(\frac{1}{\lambda+k+\frac{n+1}{2}-m}-\frac{\delta^n}{\lambda+k-\frac{n-1}{2}+m}\biggr)\delta^k \\
&=\sum_{m=0}^{\lceil n/2\rceil-1}\frac{(-1)^{m+1}}{m!(n-m)!}
\frac{1}{\varepsilon-\frac{n}{2}+m}\biggl(\sum_{k=0}^N\frac{\delta^k}{\lambda+k-\frac{n-1}{2}+m}-\sum_{k=n-2m}^{N+n-2m}\frac{\delta^k}{\lambda+k-\frac{n-1}{2}+m}\biggr) \\*
&\eqspace{}+\sum_{m=0}^{\lceil n/2\rceil-1}\frac{(-1)^{n-m+1}}{m!(n-m)!}
\frac{1}{\varepsilon+\frac{n}{2}-m}\biggl(\sum_{k=n-2m}^{N+n-2m}\frac{\delta^{k+n}}{\lambda+k-\frac{n-1}{2}+m}-\sum_{k=0}^N\frac{\delta^{k+n}}{\lambda+k-\frac{n-1}{2}+m}\biggr) \\
&=\sum_{m=0}^{\lceil n/2\rceil-1}\frac{(-1)^{m+1}}{m!(n\hspace{-1pt}-\hspace{-1pt}m)!}
\frac{1}{\varepsilon\hspace{-1pt}-\hspace{-1pt}\frac{n}{2}\hspace{-1pt}+\hspace{-1pt}m}
\biggl(\sum_{k=0}^{n-2m-1}\frac{\delta^k}{\lambda\hspace{-1pt}+\hspace{-1pt}k\hspace{-1pt}-\hspace{-1pt}\frac{n-1}{2}\hspace{-1pt}+\hspace{-1pt}m}
-\sum_{k=N+1}^{N+n-2m}\frac{\delta^k}{\lambda\hspace{-1pt}+\hspace{-1pt}k\hspace{-1pt}-\hspace{-1pt}\frac{n-1}{2}\hspace{-1pt}+\hspace{-1pt}m}\biggr) \\*
&\eqspace{}+\sum_{m=0}^{\lceil n/2\rceil-1}\frac{(-1)^{n-m+1}}{m!(n\hspace{-1pt}-\hspace{-1pt}m)!}
\frac{1}{\varepsilon\hspace{-1pt}+\hspace{-1pt}\frac{n}{2}\hspace{-1pt}-\hspace{-1pt}m}
\biggl(\sum_{k=N+1}^{N+n-2m}\frac{\delta^{k+n}}{\lambda\hspace{-1pt}+\hspace{-1pt}k\hspace{-1pt}-\hspace{-1pt}\frac{n-1}{2}\hspace{-1pt}+\hspace{-1pt}m}
-\sum_{k=0}^{n-2m-1}\frac{\delta^{k+n}}{\lambda\hspace{-1pt}+\hspace{-1pt}k\hspace{-1pt}-\hspace{-1pt}\frac{n-1}{2}\hspace{-1pt}+\hspace{-1pt}m}\biggr) \\
&=\sum_{m=0}^{\lceil n/2\rceil-1}\frac{(-1)^{m+1}}{m!(n-m)!}\biggl(\frac{1}{\varepsilon-\frac{n}{2}+m}-\frac{(-\delta)^n}{\varepsilon+\frac{n}{2}-m}\biggr)
\sum_{k=0}^{n-2m-1}\frac{\delta^k}{\lambda+k-\frac{n-1}{2}+m}+O(N^{-1}). 
\end{align*}
Hence we get the 1st equalities of $A_n^\delta(\lambda,\varepsilon)$, $B_n^\delta(\lambda,\varepsilon)$. Especially, for $n=0,1$, we have 
\begin{align*}
J_0^\delta(\lambda,\varepsilon)&=\frac{1}{2\varepsilon}\sum_{k=0}^\infty\biggl(\frac{1}{\lambda-\varepsilon+k+\frac{1}{2}}-\frac{1}{\lambda+\varepsilon+k+\frac{1}{2}}\biggr)\delta^k, \\
J_1^\delta(\lambda,\varepsilon)&=\frac{\lambda}{2\bigl(\varepsilon\hspace{-1pt}-\hspace{-1pt}\frac{1}{2}\bigr)\bigl(\varepsilon\hspace{-1pt}+\hspace{-1pt}\frac{1}{2}\bigr)}
\sum_{k=0}^\infty\biggl(\frac{1}{\lambda\hspace{-1pt}-\hspace{-1pt}\varepsilon\hspace{-1pt}+\hspace{-1pt}k\hspace{-1pt}+\hspace{-1pt}\frac{1}{2}}
-\frac{\delta}{\lambda\hspace{-1pt}+\hspace{-1pt}\varepsilon\hspace{-1pt}+\hspace{-1pt}k\hspace{-1pt}+\hspace{-1pt}\frac{1}{2}}\biggr)\delta^k
-\frac{1}{2}\biggl(\frac{1}{\varepsilon\hspace{-1pt}-\hspace{-1pt}\frac{1}{2}}+\frac{\delta}{\varepsilon\hspace{-1pt}+\hspace{-1pt}\frac{1}{2}}\biggr).
\end{align*}
Then the 2nd equalities of $A_n^\delta(\lambda,\varepsilon)$, $B_n^\delta(\lambda,\varepsilon)$ follow from the recurrence relation in Proposition \ref{prop_recurrence_Jpm}. 
\end{proof}

\begin{remark}
Although $A_n^\delta(\lambda,\varepsilon)$ and $B_n^\delta(\lambda,\varepsilon)$ have poles at $\varepsilon\in\bigl\{-\frac{n}{2},-\frac{n}{2}+1,\ldots,\frac{n}{2}\bigr\}$ and 
$\varepsilon\in\bigl\{-\frac{n}{2},-\frac{n}{2}+1,\ldots,\frac{n}{2}\bigr\}\setminus\{0\}$ respectively, 
since 
\begin{align*}
&\sum_{k=0}^\infty\biggl(\frac{1}{\lambda-\varepsilon+k+\frac{1}{2}}-\frac{\delta^n}{\lambda+\varepsilon+k+\frac{1}{2}}\biggr)\delta^k \\
&=\sum_{k=0}^{n-2m-1}\frac{\delta^k}{\lambda-\varepsilon+k+\frac{1}{2}}
+\sum_{k=0}^\infty\biggl(\frac{\delta^{k+n-2m}}{\lambda-\varepsilon+k+n-2m+\frac{1}{2}}-\frac{\delta^{k+n}}{\lambda+\varepsilon+k+\frac{1}{2}}\biggr) \\
&=\sum_{k=0}^{n-2m-1}\frac{\delta^k}{\lambda-\varepsilon+k+\frac{1}{2}}
+\sum_{k=0}^\infty\frac{(2\varepsilon-n+2m)\delta^{k+n}}{\bigl(\lambda-\varepsilon+k+n-2m+\frac{1}{2}\bigr)\bigl(\lambda+\varepsilon+k+\frac{1}{2}\bigr)}
\end{align*}
holds for $m=0,1,\ldots,\bigl\lfloor\frac{n}{2}\bigr\rfloor$, we have 
\begin{align*}
&\lim_{\varepsilon\to\frac{n}{2}-m}\biggl(\varepsilon-\frac{n}{2}+m\biggr)\frac{A_n^\delta(\lambda,\varepsilon)}{2}
\sum_{k=0}^\infty\biggl(\frac{1}{\lambda-\varepsilon+k+\frac{1}{2}}-\frac{\delta^n}{\lambda+\varepsilon+k+\frac{1}{2}}\biggr)\delta^k \\
&=\lim_{\varepsilon\to\frac{n}{2}-m}\frac{\bigl(\lambda-\frac{n-1}{2}\bigr)_n}{2\bigl(\varepsilon-\frac{n}{2}\bigr)_m\bigl(\varepsilon-\frac{n}{2}+m+1\bigr)_{n-m}}
\sum_{k=0}^{n-2m-1}\frac{\delta^k}{\lambda-\varepsilon+k+\frac{1}{2}} \\
&=\frac{\bigl(\lambda-\frac{n-1}{2}\bigr)_n}{2(-m)_m(1)_{n-m}}\sum_{k=0}^{n-2m-1}\frac{\delta^k}{\lambda+k-\frac{n-1}{2}+m} \\
&=\frac{(-1)^m\bigl(\lambda-\frac{n-1}{2}\bigr)_n}{2\cdot m!(n-m)!}\sum_{k=0}^{n-2m-1}\frac{\delta^k}{\lambda+k-\frac{n-1}{2}+m} \\
&=-\lim_{\varepsilon\to\frac{n}{2}-m}\biggl(\varepsilon-\frac{n}{2}+m\biggr)B_n^\delta(\lambda,\varepsilon). 
\end{align*}
Hence the residue of $J_n^\delta(\lambda,\varepsilon)$ at $\varepsilon=\frac{n}{2}-m$ vanishes. 
Since $J_n^\delta(\lambda,-\varepsilon)=(-\delta)^nJ_n^\delta(\lambda,\varepsilon)$ holds, the residue at $\varepsilon=-\frac{n}{2}+m$ also vanishes, 
namely, $J_n^\delta(\lambda,\varepsilon)$ is holomorphic at $\varepsilon\in\bigl\{-\frac{n}{2},-\frac{n}{2}+1,\ldots,\frac{n}{2}\bigr\}$. 
\end{remark}

From Theorem \ref{thm_Jpm}, we immediately get the following. 
\begin{corollary}
For $g\in\BR$, $\lambda,\varepsilon\in\BC$ with $\lambda\pm\varepsilon\notin-\bigl(\BZ_{\ge 0}+\frac{1}{2}\bigr)$, $\varepsilon\notin\pm\BZ_{\ge 1}$, we have 
\begin{align*}
R_1^+(\lambda;g,\varepsilon)
&=\Biggl({}_3F_2\biggl(\begin{matrix}\frac{1}{2},\frac{1}{2}+\lambda,\frac{1}{2}-\lambda, \\ 1+\varepsilon,1-\varepsilon\end{matrix}; \tanh^2(2g)\biggr)
\sum_{k=0}^\infty \frac{1}{\bigl(\lambda+\varepsilon+k+\frac{1}{2}\bigr)\bigl(\lambda-\varepsilon+k+\frac{1}{2}\bigr)} \\*
&\eqspace{}+\sum_{n=1}^\infty \frac{\bigl(\frac{1}{2}\bigr)_n}{n!}\tanh^{2n}(2g)
\sum_{l=0}^{n-1}\frac{\lambda\bigl(\lambda-n+\frac{1}{2}\bigr)_l\bigl(-\lambda-n+\frac{1}{2}\bigr)_l}{(2n-2l-1)(\varepsilon-n)_{l+1}(-\varepsilon-n)_{l+1}}\Biggr)\sech(2g), \\
R_1^-(\lambda;g,\varepsilon)
&=\Biggl({}_3F_2\biggl(\begin{matrix}\frac{1}{2},\frac{1}{2}+\lambda,\frac{1}{2}-\lambda, \\ 1+\varepsilon,1-\varepsilon\end{matrix}; \tanh^2(2g)\biggr)
\sum_{k=0}^\infty \frac{(-1)^k}{\bigl(\lambda+\varepsilon+k+\frac{1}{2}\bigr)\bigl(\lambda-\varepsilon+k+\frac{1}{2}\bigr)} \\*
&\eqspace{}+\sum_{n=1}^\infty \frac{\bigl(\frac{1}{2}\bigr)_n}{n!}\tanh^{2n}(2g)
\sum_{l=0}^{n-1}\frac{\bigl(\lambda-n+\frac{1}{2}\bigr)_l\bigl(-\lambda-n+\frac{1}{2}\bigr)_l}{2(\varepsilon-n)_{l+1}(-\varepsilon-n)_{l+1}}\Biggr)\sech(2g). 
\end{align*}
\end{corollary}
The formula $R_1^+(0;g,0)$ coincides with the result by \cite{O3}. 

\begin{remark}
\begin{enumerate}
\item The even part of $R_1^+(\lambda;g,\varepsilon)$ and the odd part of $R_1^-(\lambda;g,\varepsilon)$ with respect to $\lambda$ are simplified as, for $\delta\in\{\pm\}$, 
\begin{align*}
&R_1^\delta(\lambda;g,\varepsilon)+\delta R_1^\delta(-\lambda;g,\varepsilon) \\
&=\sech(2g){}_3F_2\biggl(\begin{matrix}\frac{1}{2},\frac{1}{2}+\lambda,\frac{1}{2}-\lambda, \\ 1+\varepsilon,1-\varepsilon\end{matrix}; \tanh^2(2g)\biggr)
\sum_{k=-\infty}^\infty \frac{\delta^k}{\bigl(\lambda+\varepsilon+k+\frac{1}{2}\bigr)\bigl(\lambda-\varepsilon+k+\frac{1}{2}\bigr)} \\
&=\sech(2g){}_2F_1\biggl(\begin{matrix}\frac{1}{4}+\frac{\lambda+\varepsilon}{2},\frac{1}{4}-\frac{\lambda-\varepsilon}{2}, \\ 1+\varepsilon\end{matrix}; \tanh^2(2g)\biggr)
{}_2F_1\biggl(\begin{matrix}\frac{1}{4}-\frac{\lambda+\varepsilon}{2},\frac{1}{4}+\frac{\lambda-\varepsilon}{2}, \\ 1-\varepsilon\end{matrix}; \tanh^2(2g)\biggr) \\*
&\hspace{198pt}\times\sum_{k=-\infty}^\infty \frac{\delta^k}{\bigl(\lambda+\varepsilon+k+\frac{1}{2}\bigr)\bigl(\lambda-\varepsilon+k+\frac{1}{2}\bigr)} \\
&=\sech^3(2g){}_2F_1\biggl(\begin{matrix}\frac{3}{4}+\frac{\lambda+\varepsilon}{2},\frac{3}{4}-\frac{\lambda-\varepsilon}{2}, \\ 1+\varepsilon\end{matrix}; \tanh^2(2g)\biggr)
{}_2F_1\biggl(\begin{matrix}\frac{3}{4}-\frac{\lambda+\varepsilon}{2},\frac{3}{4}+\frac{\lambda-\varepsilon}{2}, \\ 1-\varepsilon\end{matrix}; \tanh^2(2g)\biggr) \\*
&\hspace{198pt}\times\sum_{k=-\infty}^\infty \frac{\delta^k}{\bigl(\lambda+\varepsilon+k+\frac{1}{2}\bigr)\bigl(\lambda-\varepsilon+k+\frac{1}{2}\bigr)}. 
\end{align*}
The 2nd and 3rd equalities follow from \cite[Exercise 2.13, (2.2.7)]{AAR}, 
\begin{gather*}
{}_3F_2\biggl(\begin{matrix}\frac{1}{2},a-b+\frac{1}{2},b-a+\frac{1}{2}\\a+b+\frac{1}{2},\frac{3}{2}-a-b\end{matrix};x\biggr)
={}_2F_1\biggl(\begin{matrix}a,b\\a+b+\frac{1}{2}\end{matrix};x\biggr){}_2F_1\biggl(\begin{matrix}\frac{1}{2}-a,\frac{1}{2}-b\\\frac{3}{2}-a-b\end{matrix};x\biggr), \\
{}_2F_1\biggl(\begin{matrix}a,b\\c\end{matrix};x\biggr)=(1-x)^{c-a-b}{}_2F_1\biggl(\begin{matrix}c-a,c-b\\c\end{matrix};x\biggr). 
\end{gather*}
\item Similarly, the generating function $V_\delta(t;\lambda,\varepsilon)$ in (\ref{def_Vdelta}) is computed as 
\begin{align*}
V_+(t;\lambda,\varepsilon)
&={}_3F_2\biggl(\begin{matrix}1,\frac{1}{2}+\lambda,\frac{1}{2}-\lambda, \\ 1+\varepsilon,1-\varepsilon\end{matrix}; t\biggr)
\sum_{k=0}^\infty \frac{1}{\bigl(\lambda+\varepsilon+k+\frac{1}{2}\bigr)\bigl(\lambda-\varepsilon+k+\frac{1}{2}\bigr)} \\*
&\eqspace{}+\sum_{n=1}^\infty \sum_{l=0}^{n-1}\frac{\lambda\bigl(\lambda-n+\frac{1}{2}\bigr)_l\bigl(-\lambda-n+\frac{1}{2}\bigr)_l}{(2n-2l-1)(\varepsilon-n)_{l+1}(-\varepsilon-n)_{l+1}}t^n, \\
V_-(t;\lambda,\varepsilon)
&={}_3F_2\biggl(\begin{matrix}1,\frac{1}{2}+\lambda,\frac{1}{2}-\lambda, \\ 1+\varepsilon,1-\varepsilon\end{matrix}; t\biggr)
\sum_{k=0}^\infty \frac{(-1)^k}{\bigl(\lambda+\varepsilon+k+\frac{1}{2}\bigr)\bigl(\lambda-\varepsilon+k+\frac{1}{2}\bigr)} \\*
&\eqspace{}+\sum_{n=1}^\infty \sum_{l=0}^{n-1}\frac{\bigl(\lambda-n+\frac{1}{2}\bigr)_l\bigl(-\lambda-n+\frac{1}{2}\bigr)_l}{2(\varepsilon-n)_{l+1}(-\varepsilon-n)_{l+1}}t^n. 
\end{align*}
Then the even part of $V_+(t;\lambda,\varepsilon)$ and the odd part of $V_-(t;\lambda,\varepsilon)$ with respect to $\lambda$ are simplified as, for $\delta\in\{\pm\}$, 
\begin{align*}
&V_\delta(t;\lambda,\varepsilon)+\delta V_\delta(t;-\lambda,\varepsilon) \\
&={}_3F_2\biggl(\begin{matrix}1,\frac{1}{2}+\lambda,\frac{1}{2}-\lambda, \\ 1+\varepsilon,1-\varepsilon\end{matrix}; t\biggr)
\sum_{k=-\infty}^\infty \frac{\delta^k}{\bigl(\lambda+\varepsilon+k+\frac{1}{2}\bigr)\bigl(\lambda-\varepsilon+k+\frac{1}{2}\bigr)}. 
\end{align*}
Especially, for $\varepsilon=0$, we have 
\begin{align}
V_+(t;\lambda,0)+V_+(t;-\lambda,0)
&=\int_{[0,1]^2} \frac{(1-uv)(u^\lambda v^\lambda+u^{-\lambda}v^{-\lambda})}{(1-uv)^2-t(u-v)^2}u^{-\frac{1}{2}}v^{-\frac{1}{2}}\,\rd u\rd v \notag\\
&=\frac{\pi^2}{\cos^2(\pi\lambda)}{}_2F_1\biggl(\begin{matrix}\frac{1}{2}+\lambda,\frac{1}{2}-\lambda, \\ 1\end{matrix}; t\biggr), \label{formula_V+even}\\
V_-(t;\lambda,0)-V_-(t;-\lambda,0)
&=\int_{[0,1]^2} \frac{(1+uv)(u^\lambda v^\lambda-u^{-\lambda}v^{-\lambda})}{(1+uv)^2-t(u+v)^2}u^{-\frac{1}{2}}v^{-\frac{1}{2}}\,\rd u\rd v \notag\\
&=-\frac{\pi^2\sin(\pi\lambda)}{\cos^2(\pi\lambda)}{}_2F_1\biggl(\begin{matrix}\frac{1}{2}+\lambda,\frac{1}{2}-\lambda, \\ 1\end{matrix}; t\biggr). \notag
\end{align}
The formula (\ref{formula_V+even}) coincides with \cite[(4.2), Proposition 4.2]{KW4}. 
\item By \cite{KW2}, by substituting 
\[ s=s(\tau)=-\frac{\theta_2(\tau)^4}{\theta_4(\tau)^4}=\frac{\eta(\tau)^8\eta(4\tau)^{16}}{\eta(2\tau)^{24}} \]
into $\widetilde{V}_+(s;0,0):=\frac{1}{1-s}V_+\bigl(\frac{s}{s-1};0,0\bigr)$, where $\theta_j(\tau)$ are the elliptic theta functions and $\eta(\tau)$ is the Dedekind eta function, 
we get 
\[ \widetilde{V}_+(s;0,0)=\frac{\pi^2}{2(1-s)}{}_2F_1\biggl(\begin{matrix}\frac{1}{2},\frac{1}{2}, \\ 1\end{matrix}; \frac{s}{s-1}\biggr)
=\frac{\pi^2}{2}\frac{\theta_4(\tau)^4}{\theta_3(\tau)^2}=\frac{\pi^2}{2}\frac{\eta(2\tau)^{22}}{\eta(\tau)^{12}\eta(4\tau)^8}, \]
which is a $\Gamma(2)$-modular form of weight 1. \cite[Problems 4.1--4.3]{KW4} proposed a problem of seeking similar modular properties for $V_+(t;\lambda,0)+V_+(t;-\lambda,0)$, 
and it seems natural to consider an analogous problem for $V_+(t;\lambda,\varepsilon)+V_+(t;-\lambda,\varepsilon)$ for general $\lambda,\varepsilon$. 
\end{enumerate}
\end{remark}


\end{document}